\documentclass[11pt]{article}
\usepackage{smile}
\usepackage{amsmath}
\usepackage{fullpage}
\usepackage{url}
\usepackage{csquotes}
\usepackage{lscape}
\usepackage{bigints}
\usepackage{framed}
\usepackage{mdframed}
\usepackage{enumerate}
\usepackage[inline]{enumitem}
\usepackage[T1]{fontenc}
\usepackage{moresize}
\usepackage{bm}
\usepackage{bbm}
\usepackage{dsfont}
\usepackage{amsmath}
\usepackage{amssymb}
\usepackage{amsthm}
\usepackage{amsfonts}
\usepackage{stmaryrd}
\usepackage{array}
\usepackage{mathrsfs}
\usepackage{mathtools} 
\usepackage{extarrows}
\usepackage{stackrel}
\usepackage{relsize,exscale}
\usepackage{scalerel}
\usepackage{colortbl}
\usepackage[nodisplayskipstretch]{setspace}
\usepackage{color}
\usepackage[dvipsnames]{xcolor}
\usepackage{cancel}
\usepackage{soul}
\usepackage{xfrac}
\usepackage{siunitx}
\usepackage{graphicx}
\usepackage{float}
\usepackage{rotating}
\usepackage{subcaption}
\usepackage[ruled,vlined]{algorithm2e}
\usepackage{overpic}
\usepackage[all]{xy}
\usepackage{tikz}
\usetikzlibrary{arrows,matrix,positioning,calc,automata,patterns}
\usepackage{booktabs}
\usepackage{dcolumn}
\usepackage{multirow}
\usepackage{diagbox}
\usepackage{tabularx}
\usepackage{verbatim}
\usepackage{listings}
\usepackage[ruled,vlined]{algorithm2e}
\usepackage{fancyvrb}
\usepackage{hyperref}
\usepackage[round]{natbib}
\usepackage{hyperref}
\usepackage{cleveref}
\usepackage{theoremref}
\usepackage{sectsty}

\hypersetup{
    unicode=false,          
    pdftoolbar=true,        
    pdfmenubar=true,        
    pdffitwindow=false,     
    pdfstartview={FitH},    
    pdftitle={My title},    
    pdfauthor={Author},     
    pdfsubject={Subject},   
    pdfcreator={Creator},   
    pdfproducer={Producer}, 
    pdfkeywords={key1, key2}, 
    pdfnewwindow=true,      
    colorlinks=true,        
    linkcolor=blue,         
    citecolor=blue,         
    filecolor=blue,         
    urlcolor=cyan           
}

\usepackage{stackengine}
\stackMath
\newcommand\tenq[2][1]{%
\def\useanchorwidth{T}%
\ifnum#1>1%
\stackunder[0pt]{\tenq[\numexpr#1-1\relax]{#2}}{\!\scriptscriptstyle\thicksim}%
\else%
\stackunder[1pt]{#2}{\!\scriptstyle\thicksim}%
\fi%
}

\makeatletter
\DeclareRobustCommand\widecheck[1]{{\mathpalette\@widecheck{#1}}}
\def\@widecheck#1#2{%
    \setbox\z@\hbox{\m@th$#1#2$}%
    \setbox\tw@\hbox{\m@th$#1%
       \widehat{%
          \vrule\@width\z@\@height\ht\z@
          \vrule\@height\z@\@width\wd\z@}$}%
    \dp\tw@-\ht\z@
    \@tempdima\ht\z@ \advance\@tempdima2\ht\tw@ \divide\@tempdima\thr@@
    \setbox\tw@\hbox{%
       \raise\@tempdima\hbox{\scalebox{1}[-1]{\lower\@tempdima\box
\tw@}}}%
    {\ooalign{\box\tw@ \cr \box\z@}}}
\makeatother

\def\given{\,|\,}

\def\tr{\mathop{\text{tr}}\kern.2ex}

\def\B{{\mathbb B}}
\def\D{{\mathbb D}}
\def\P{{\mathrm P}}
\def\Q{{\mathrm Q}}

\def\E{{\mathrm E}}
\def\R{{\mathbb R}}

\def\Ab{\mathbf{A}}
\def\bu{\bm{u}}
\def\sW{{\sf W}}
\def\sL{{\sf L}}
\def\rD{{\rm D }}

\def\d{{\mathrm d}}

\newcommand{\op}{\mathrm{op}}

\newcommand{\KL}{\mathrm{KL}}

\newcommand\yestag{\addtocounter{equation}{1}\tag{\theequation}}

\newcommand{\pa}[1]{\left(#1\right)}
\newcommand{\bp}[1]{\left[#1\right]}
\newcommand{\ap}[1]{\left|#1\right|}
\newcommand{\cp}[1]{\left\{#1\right\}}
\newcommand{\al}[1]{\begin{align*}#1\end{align*}}

\newcommand{\pr}[1]{\mathrm{P}\pa{#1}}

\newcommand{\ep}[1]{\E\bp{#1}}
\newcommand{\cdl}{\middle |}

\newcolumntype{L}[1]{>{\raggedright\let\newline\\\arraybackslash\hspace{0pt}}m{#1}}
\newcolumntype{C}[1]{>{  \centering\let\newline\\\arraybackslash\hspace{0pt}}m{#1}}
\newcolumntype{R}[1]{>{ \raggedleft\let\newline\\\arraybackslash\hspace{0pt}}m{#1}}
\newcolumntype{d}[1]{D{.}{.}{#1}}
\newcolumntype{H}{>{\setbox0=\hbox\bgroup}c<{\egroup}@{}}
\newcolumntype{Z}{>{\setbox0=\hbox\bgroup}c<{\egroup}@{\hspace*{-\tabcolsep}}}
\newcolumntype{s}{>{\hsize=.5\hsize}X}

\numberwithin{equation}{section}

\newtheorem{theorem}{Theorem}[section]
\newtheorem{lemma}{Lemma}[section]

\newtheorem{assumption}{Assumption}[section]

\newtheorem{innercustomgeneric}{\customgenericname}
\providecommand{\customgenericname}{}
\newcommand{\newcustomtheorem}[2]{%
  \newenvironment{#1}[1]
  {%
   \renewcommand\customgenericname{#2}%
   \renewcommand\theinnercustomgeneric{##1}%
   \innercustomgeneric
  }
  {\endinnercustomgeneric}
}
\newcustomtheorem{customdefinition}{Definition}
\newcustomtheorem{customdefinitions}{Definitions}
\newcustomtheorem{customtheorem}{Theorem}
\newcustomtheorem{customassumption}{Assumption}
\newcustomtheorem{customlemma}{Lemma}
\newcustomtheorem{customexample}{Example}
\theoremstyle{definition}
\newtheorem{definition}{Definition}[section]
\newtheorem{example}{Example}[section]
\newtheorem{remark}{Remark}[section]

\usepackage{enumitem}
\makeatletter
\newcommand{\mylabel}[2]{#2\def\@currentlabel{#2}\label{#1}}
\makeatother

\graphicspath{{./fig3/}}

\allowdisplaybreaks

\let\check\widecheck

\begin{document}

\title{\LARGE Generative modeling for the bootstrap}

\author{Leon Tran\thanks{Department of Statistics, University of Washington, Seattle, WA 98195, USA; e-mail: {\tt leontk@uw.edu}}, ~~Ting Ye\thanks{Department of Biostatistics, University of Washington, Seattle, WA 98195, USA. E-mail: \tt{tingye1@uw.edu}}, ~~Peng Ding\thanks{Department of Statistics, University of California, Berkeley, CA 94720, USA; email: {\tt pengdingpku@berkeley.edu}}, ~and~Fang Han\thanks{Department of Statistics, University of Washington, Seattle, WA 98195, USA; e-mail: {\tt fanghan@uw.edu}}
}

\date{\today}

\maketitle

\vspace{-1em}

\begin{abstract}
Generative modeling builds on and substantially advances the classical idea of simulating synthetic data from observed samples. This paper shows that this principle is not only natural but also theoretically well‐founded for bootstrap inference: it yields statistically valid confidence intervals that apply simultaneously to both regular and irregular estimators, including settings in which Efron’s bootstrap fails. In this sense, the generative modeling-based bootstrap can be viewed as a modern version of the smoothed bootstrap: it could mitigate the curse of dimensionality and remain effective in challenging regimes where estimators may lack root-$n$ consistency or a Gaussian limit.
\end{abstract}


\section{Introduction}

Simulating synthetic data from observed samples is by no means a new idea. In statistics, this principle has been proposed repeatedly in support of various data-analytic tasks. Its roots can be traced at least to the work of \citet{scott1954comparison} and \citet{neyman1956distribution}, who employed synthetic sampling for model checking and for detecting unsuspected patterns. Later, \citet{efron1979bootstrap} introduced the {\it bootstrap}, which relies on repeated sampling from the empirical distribution function for statistical inference.  \citet{rubin1993statistical} and \citet{little1993statistical} further explored the idea in the context of privacy protection, advocating the release of fully synthetic datasets and proposing to use {\it multiple imputation} \citep{rubin1987multiple}, which generates new data by sampling from a Bayesian posterior distribution.

From a different corner of the scientific landscape, machine learning---originally centered on prediction \citep{breiman2001statistical}---has undergone tremendous advances, particularly with the rise of deep learning. Against this backdrop, the seminal contributions of \cite{kingma2013auto} and \cite{goodfellow2014generative}, followed by \cite{chen2018neural}, \cite{song2020score}, and many others, sparked a new revolution. At the heart of this revolution, now known as {\it generative modeling}, lies the principle of ``creating data from noise'' \citep{song2020score}: a perspective that differs in intriguing ways from the early ideas explored by Scott, Neyman, Efron, and Rubin.

Motivated by these subtle differences, as well as by the remarkable empirical success of generative modeling, this paper proposes to leverage the principle of generative modeling for bootstrap inference. In particular, we develop a new framework that embraces the perspective of ``creating data from noise'' and enables bootstrap procedures based on repeated sampling from a generative model learned from the observed data. Alternatively, this framework could be viewed as generalizing
\begin{itemize}
\item[(a)] the parametric bootstrap \citep{efron2012bayesian}, which resamples from a learned parametric model;  
\item[(b)] the smoothed bootstrap \citep{efron1979bootstrap,silverman1987bootstrap}, which resamples from a nonparametric estimate of the data distribution. 
\end{itemize}
In this sense, generative modeling-based bootstrap can be regarded as a modern version of the smoothed bootstrap: it approximates the distribution of an unknown statistical estimator by resampling from a flexible, nonparametric estimator of the underlying data distribution.

From a theoretical standpoint, under the proposed general framework we establish broad criteria ensuring that {\it any} data distribution estimator satisfying these conditions yields a consistent bootstrap method for both regular and irregular statistical procedures---the latter being settings in which Efron's bootstrap is known to fail \citep{kosorok2008bootstrapping,sen2010inconsistency,groeneboom2024confidence,lin2024failure,lin2026consistency}. The resulting theory, presented in Theorems~\ref{thm_reg} and~\ref{thm:iso}, may thus be viewed as a modern counterpart of \citet[Theorem~2.1]{bickel1981some}, who already envisioned the possibility of resampling from a general estimator of the data distribution and developed the foundational theory more than forty years ago.

Specializing our framework to concrete generative modeling techniques, and building largely on the theoretical insights of \cite{biau2020some}, \cite{shen2023asymptotic}, and \cite{irons2022triangular}, we identify conditions under which generative adversarial networks (GANs) and flow-based generative models naturally fit within our setup, thereby giving rise to GAN-based and flow-based bootstrap procedures. From this perspective, the flow bootstrap is particularly appealing: unlike GAN-based approaches, flow estimators are typically more regular and guaranteed to be nondegenerate. Consequently, they lead to consistent bootstrap procedures that apply to both regular and irregular estimators, whereas the GAN bootstrap generally lacks comparable consistency guarantees in the irregular setting. This provides an additional, statistical inferential, perspective favoring flow-based over GAN-based generative models \citep{kobyzev2020normalizing}.

The rest of this paper is organized as follows. Section~\ref{sec:frame} introduces the general framework and provides illustrative examples. Sections~\ref{sec:regular} and~\ref{sec:iso} develop the corresponding bootstrap consistency theory for regular M-estimators and for the isotonic regression estimator, the latter being a prominent example of an irregular estimator. Section~\ref{sec:gai} presents conditions under which certain versions of GANs and flow-based models yield consistent bootstrap procedures. Simulation results are reported in Section~\ref{sec:sim}, and the proofs of the main theorems are collected in Section~\ref{sec:proofmain}.
Supporting lemmas are stated in Sections~\ref{sec:proofsupp} and~\ref{sec:otherresults}.

\section{Generative modeling-based bootstrap} \label{sec:frame}

\subsection{A general framework}\label{sec:framework}

Consider random vectors 
\[
\mZ,\mZ_1,\mZ_2, \ldots \in \cZ \subset \R^p 
\]
sampled independently from some unknown {\it data distribution} $\P_Z$ with an unknown support $\cZ$. In this paper, the {\it support} of the distribution $\P_Z$ of $\mZ$ refers to the smallest closed set $\cZ \subseteq \R^p$ such that $\P(\mZ \in \cZ) = 1$.  A common statistical task is to estimate and infer an {\it estimand} $\theta_0 = \theta_0(\P_Z)$ using an {\it estimator} $\hat\theta_n = \hat\theta_n(\mZ_1,\ldots,\mZ_n)$, which is a function of the data $\{\mZ_i : i \in [n] \}$ with size $n$ and $[n] := \{1,2,\ldots,n\}$.

Unlike estimation, inference requires a deeper understanding of the stochastic behavior of $\hat\theta_n$: in particular its (limiting) distribution. To approximate the distribution of $\hat\theta_n$, bootstrap methods are widely used and typically proceed in two steps:
\begin{itemize}
\item[] \textbf{Step 1:} For each bootstrap iteration, resample $n$ {\it synthetic observations} $\tilde\mZ_1,\ldots,\tilde\mZ_n \in \tilde{\cZ}_n$ from a (random) distribution $\bP_{\tilde Z,n}$, with support $\tilde{\cZ}_n$, that is learned from the data and intended to approximate $\P_Z$.
\item[] \textbf{Step 2:} Use the conditional distribution of $\hat\theta_n(\tilde\mZ_1,\ldots,\tilde\mZ_n)$ given the original sample to approximate the sampling distribution of $\hat\theta_n = \hat\theta_n(\mZ_1,\ldots,\mZ_n)$.
\end{itemize}
Different bootstrap procedures arise from different choices of $\bP_{\tilde Z,n}$. The choice $\bP_{\tilde Z,n} = \bP_n^Z$, the empirical measure of $\{\mZ_i\}_{i\in[n]}$, corresponds to the original proposal of \cite{efron1979bootstrap} and remains the most widely used form of bootstrap.

Adopting the generative modeling philosophy, we introduce a new class of choices for $\bP_{\tilde Z,n}$ by incorporating additional randomness. Let 
\[
\mU, \mU_1, \mU_2, \ldots ~~{\rm and}~~~ \tilde\mU_1, \tilde\mU_2, \ldots \in \cU \subset \R^p
\]
be  random vectors sampled independently from a {\it known} distribution $\P_U$, with support $\cU$, and independent of the data. One may regard $\mU_i$'s and $\tilde\mU_i$'s as {\it noise} and $\P_U$ as the corresponding {\it noise distribution}. A broad class of generative models approximates the data distribution $\P_Z$ by learning a {\it generator}
\begin{align}\label{eq:generator}
\hat\mG_n: \cU \to \tilde{\cZ}_n,
\end{align}
from either the paired observations $\{(\mZ_i,\mU_i)\}_{i\in[n]}$ or from $\{\mZ_i\}_{i\in[n]}$ alone. The goal of this learning process is to ensure that the pushforward distribution $\hat\mG_n \# \P_U$ is close, in 
some predefined metric, to the true data distribution $\P_Z$. The sample $\{\hat\mG_n(\tilde\mU_i)\}_{i\in[n]}$ then constitutes size-$n$ {\it synthetic data}, created from noise.

Because $\hat\mG_n \# \P_U$ is intended to approximate $\P_Z$, it is natural to introduce a new class of bootstrap procedures by setting
\[
\tilde\mZ_i = \hat\mG_n(\tilde\mU_i), \qquad i \in [n],
\]
and using the conditional distribution of
\[
\hat\theta_n\big(\tilde\mZ_1,\ldots,\tilde\mZ_n\big) \given \{(\mZ_i,\mU_i)\}_{i\in[n]}
\]
to approximate the sampling distribution of $\hat\theta_n(\mZ_1,\ldots,\mZ_n)$. In this paper, we refer to such procedures as {\it generative modeling-based bootstraps}.

\subsection{Examples}\label{sec:example}

Different generative models correspond to different choices of $\hat\mG_n$ in \eqref{eq:generator}. To introduce the generative models of interest, we begin with some additional notation. For any vector, let $\dim(\cdot)$ denote its dimension, and let $\|\cdot\|_2$, and $\|\cdot\|_\infty$ denote its $\ell_2$, and $\ell_\infty$ norms, respectively. Whenever ``$\leq$'' is used to compare two vectors, the comparison is done componentwise. For any (not necessarily square) matrix $\Ab$, let $\|\Ab\|_{\op}$ denote its spectral norm, and $\|\Ab\|_{\max}$ denote the maximum absolute value among its entries. For a square matrix, let $\det(\cdot)$ denote its determinant. Throughout the manuscript, the symbols ``$\vee$'' and ``$\wedge$'' represent the maximum and minimum, respectively, of two quantities.

 We first introduce the function class of {\it neural networks}.

\begin{definition}[Neural networks] \label{def:nn}
A neural network function class, denoted by $\cF_\alpha( L, W, B, q_1, q_2)$, consists of all neural networks with {\it depth} $L$, {\it width bound} $W$, {\it magnitude bound} $B$, input dimension $q_1$, output dimension $q_2$, and activation function 
\[
\alpha(\cdot): \R \rightarrow \R. 
\]
A function $\mf \in \cF_\alpha(L, W, B, q_1, q_2)$ is a mapping $\mf:\R^{q_1} \to \R^{q_2}$ defined recursively by $\mf(\mx) = \mx^{(L)}$, where $\mx^{(0)} = \mx$ and
\[
\mx^{(\ell)} = \alpha \Big(\Ab^{(\ell)} \mx^{(\ell-1)} + \mb^{(\ell)}\Big), \quad \ell \in [L],
\]
with $\alpha(\cdot)$ applied componentwise and the matrices $\Ab^{(1)},\ldots,\Ab^{(L)}$ and vectors $\mb^{(1)},\ldots,\mb^{(L)}$ satisfying
\begin{align*}
    \max_{\ell \in [L]} \|\Ab^{(\ell)}\|_{\op} \,\vee\, \max_{\ell \in [L]} \|\mb^{(\ell)}\|_{2} &\le B~~~{\rm and}~~~\max_{\ell \in [L]} \dim\big(\mb^{(\ell)}\big) \le W.
\end{align*}
\end{definition}

With Definition~\ref{def:nn}, we are ready to introduce the neural network-based (Wasserstein) GAN.

\begin{example}[Wasserstein GAN-based generative models, \cite{arjovsky2017wasserstein}] \label{exe:wgan}
Fix sequences of positive integers $L_n^{\rm gen}$,$ W_n^{\rm gen}$, $B_n^{\rm gen}$ and $L_n^{\rm disc}, W_n^{\rm disc}$ that may depend on the sample size $n$. Also, choose an activation function $\alpha(\cdot): \R \rightarrow \R$. Define the classes of {\it generator neural networks} and {\it discriminator neural networks} as
\[
\cG_n := \cF_\alpha(L_n^{\rm gen}, W_n^{\rm gen}, B_n^{\rm gen}, p, p),
\qquad 
\cD_n := \cF_\alpha(L_n^{\rm disc}, W_n^{\rm disc}, 1, p, 1).
\]
A Wasserstein GAN (W-GAN) aims to minimize the loss function 
\[
\sW: \cG_n \times \cD_n \times \R^p \times \R^p \to \R,
\qquad 
\sW(\mG, D, \mz, \bu) := D(\mG(\bu)) - D(\mz),
\]
which is closely connected to the Wasserstein metric \citep[Equation~(3)]{arjovsky2017wasserstein}. To train a GAN generator $\hat\mG_n^{\rm GAN}$, alternating maximization/minimization is performed so that
\[
\hat D_n^{(k)}
    \in 
    \arg\max_{D \in \cD_n}
    \sum_{i=1}^n \sW\big(\hat\mG_n^{(k-1)}, D, \mZ_i, \mU_i\big)
~~~
{\rm and}
~~~
\hat\mG_n^{(k)}
    \in 
    \arg\min_{\mG \in \cG_n}
    \sum_{i=1}^n \sW\big(\mG, \hat D_n^{(k)}, \mZ_i, \mU_i\big).
\]
Both updates are implemented using stochastic gradient descent over the corresponding neural network parameters. The final generator $\hat\mG_n^{\rm GAN}$ (and discriminator $\hat D_n^{\rm GAN}$) is then taken as $\hat\mG_n^{(k)}$ (and $\hat D_n^{(k)}$) for some sufficiently large $k$.
\end{example}

Flow-based generative models provide attractive alternatives to GAN-based approaches, offering more tractable and stable distributions \citep{kobyzev2020normalizing}. We illustrate this using the following {\it autoregressive flows} \citep{huang2018neural} coupled with affine transformers \citep{dinh2016density}, which are referred to as {\it affine autoregressive flows}, beginning with a definition of bijective monotone upper triangular functions.

\begin{definition}[Bijective monotone upper triangular functions]\label{def:bmut} A {\it bijective monotone upper triangular} function $\mF = (F_1,\ldots,F_p)^\top:\R^p \to \R^p$ is a function that satisfies (a) $\mF$ is bijective, (b) each $F_i$ is strictly increasing in each of its coordinates; and
(c) each $F_i$ depends only on the first $i$ coordinates of the input.  That is, for any $\mx = (x_1, \ldots, x_p)^\top \in \R^p$,
      \[
      F_i(\mx) = F_i(x_1, \ldots, x_i).
      \]
\end{definition} 

With Definition~\ref{def:bmut}, we are now ready to introduce the affine autoregressive flows and the corresponding generative models.

\begin{definition}[Affine autoregressive flows] \label{def:tf}
For any positive integer $\nu$, a function class $\cF_\nu$ is called a class of  affine autoregressive flows of depth $\nu$ if it can be expressed as
\[
    \cF_\nu 
    = 
    \Bigl\{
        \mF^\nu \circ \boldsymbol\Sigma_\nu \circ \cdots \circ \mF^1 \circ \boldsymbol\Sigma_1
        :
        \det(\boldsymbol\Sigma_j) \neq 0,\;
        \mF^j \in \cT_\uparrow(p)
    \Bigr\},
\]
where $\cT_\uparrow(p)$ denotes the set of all {\it bijective monotone upper triangular} functions with domain and range $\R^p$.
\end{definition}


\begin{example}[Affine autoregressive flow-based generative models] \label{exe:flow}
Assume the known noise distribution $\P_U$ admits a Lebesgue density, $p_U$. For any $\mS \in \cF_\nu$, write 
\[
\mS = \mF^\nu \circ \boldsymbol\Sigma_\nu \circ \cdots \circ \mF^1 \circ \boldsymbol\Sigma_1,
\quad 
\mF^i = (F_1^i,\ldots,F_p^i)^\top,
\]
and define the objective function \(\Gamma : \cF_\nu \times \R^p \to \R\) as
$$
\Gamma(\mS,\mz) 
= \log p_U(\mS(\mz)) + \sum_{i=1}^\nu \Big\{
        \log \det(\boldsymbol\Sigma_i) 
        + \sum_{j=1}^p \log\big( \sD_j F_j^i(\mx^{(i)}) \big)
    \Big\}, 
$$
where $\mx^{(i)} 
= \boldsymbol\Sigma_i \circ \cdots \circ \mF^1 \circ \boldsymbol\Sigma_1 (\mz)$ for $i = 1,\ldots,\nu$, and $\sD_j$ denotes the partial derivative with respect to the \(j\)-th coordinate for $j=1,\ldots, p$.

The change-of-variables formula gives
\[
    \Gamma(\mS,\mz) = \log p_{S^{-1}(U)}(\mz),
\]
where \(p_{S^{-1}(U)}\) is the Lebesgue density of the transformed random variable \(\mS^{-1}(\mU)\).  
In other words, the objective \(\Gamma(\mS,\mz)\) returns the log-density of \(\mS^{-1}(\mU)\) evaluated at the point \(\mz\).

Training a flow generator \(\hat\mG_n^{\rm flow}\) therefore reduces to the following (nonparametric) maximum likelihood estimation problem:
\begin{align}\label{eq:flow}
\hat\mG_n^{\rm flow}
    = (\hat\mS_n^{\rm flow})^{-1}~~\text{ with }
   \hat\mS_n^{\rm flow} \in \arg\max_{\mS \in \cF_\nu}
    \sum_{i=1}^n \Gamma(\mS, \mZ_i).
\end{align}
In \eqref{eq:flow}, the optimization problem is often solved by considering a smaller function class than $\cF_\nu$, leading to Real NVP \citep{dinh2016density} and many other popular normalizing flow models; cf. \citet[Section 3.1]{papamakarios2021normalizing}. Next, thanks to the affine autoregressive structure, $\hat\mS_n^{\rm flow}$ can 
be inverted analytically, leading to a computationally stable and tractable flow generator $\hat\mG_n^{\rm flow}$.
\end{example}

\subsection{Discussion}

We conclude this section with a brief discussion of the connections between the generative modeling-based bootstrap framework and the classical bootstrap literature, along with some related works. To begin with, we note that the framework introduced in Section~\ref{sec:framework} also encompasses many classical bootstrap procedures. For example, let $\P_U$ denote the Lebesgue measure on $[0,1]^p$, and define
\[
\hat\mG_n(\bu)
    = \sum_{i=1}^n \mZ_i \cdot 
      \mathds{1}\!\Big( \tfrac{i-1}{n} < u_1 \le \tfrac{i}{n} \Big),
    \qquad 
    \text{for any } \bu = (u_1,\ldots,u_p)^\top \in [0,1]^p,
\]
where $\mathds{1}(\cdot)$ denotes the indicator function. This construction immediately recovers Efron's original bootstrap. Such a connection is natural because, at least in one dimension, the quantile map constitutes the {\it optimal transport} from $\P_U$ to $\P_Z$ under any convex cost function \citep[Theorem~1.5.1]{panaretos2020invitation}.

In a similar vein, the smoothed bootstrap introduced by \cite{efron1979bootstrap} can also be accommodated within the framework of Section~\ref{sec:framework}. Specifically, given any {\it proper} estimator of the underlying data-generating distribution (e.g., the kernel density estimators), one may generate new samples by applying the Brenier map \citep{brenier1991polar} to uniformly distributed noise via the optimal transport. From this perspective, the smoothed bootstrap may be interpreted as a generative modeling-based bootstrap procedure, even though Efron's original motivation was rooted in a rather different philosophical standpoint.

Overall, despite the natural appeal of such generative modeling-based bootstrap methods, it is somewhat striking that the literature along this direction remains relatively sparse. Three notable exceptions are \cite{haas2020statistical}, which suggested using GANs to implement a version of the smoothed bootstrap; \cite{dahl2022time}, which explored GAN-based bootstrap inference for time series; and \cite{athey2024using}, which investigated the use of GANs for causal inference. In all cases, however, the scope is relatively specialized and the emphasis is predominantly empirical.

\section{Theory for regular M-estimators} \label{sec:regular}
We illustrate the validity of generative modeling-based bootstrap by first considering one of the most prevalent classes of estimators: {\it M-estimators}.  
Let 
\[
\sL : \R^q \times \R^p \to \R
\]
be a general objective function mapping a $q$-dimensional parameter 
and a $p$-dimensional data point 
to a real value.  
For a parameter space \(\cK \subset \R^q\), define the population and empirical maximizers
\[
\meta_0 := \arg\max_{\meta \in \cK} \E\big[\sL(\meta, \mZ)\big],
\qquad
\hat\meta_n \approx \arg\max_{\meta \in \cK} \sum_{i=1}^n \sL(\meta, \mZ_i),
\]
where the uniqueness of the maximizers is assumed and ``\(\approx\)'' allows for numerical optimization error.

In the bootstrap analogue, define similarly
\[
\tilde\meta_0 \in \arg\max_{\meta \in \cK} 
    \E\!\Big[\sL\big(\meta, \hat\mG_n(\mU)\big) \,\big|\, \mZ_1,\ldots,\mZ_n, \mU_1, ..., \mU_n\Big],
\qquad
\tilde\meta_n \approx \arg\max_{\meta \in \cK}
    \sum_{i=1}^n \sL\big(\meta, \hat\mG_n(\tilde\mU_i)\big).
\]
Bootstrap inference then proceeds by approximating the distribution of  
\(
\hat\meta_n - \meta_0
\)
using the conditional distribution of  
\(
\tilde\meta_n - \tilde\meta_0
\)
given the original data.

To establish consistency of generative modeling-based bootstrap procedures, we first lay out the required conditions on the data/noise space, the generator, and the M-estimators.

\begin{assumption}[Data space, I]\label{assump_ndata}
Assume that:
\begin{enumerate}[label=(\alph*)]
\item $\mZ, \mZ_1, \mZ_2, \ldots \in \cZ \subset \R^p$ are independently drawn from an unknown distribution $\P_Z$ that admits a continuous Lebesgue density $p_Z$ that has nonzero variance;
\item the set $\cZ$ is convex and compact.
\end{enumerate}
\end{assumption}

\begin{assumption}[Noise space, I]\label{assump_noise1}
Assume that $\mU, \mU_1, \mU_2, \ldots$ and $\tilde\mU_1, \tilde\mU_2, \ldots \in \cU$ are independently drawn from the known distribution $\P_U$, are independent of the data, and have nonzero variance.
\end{assumption}

\begin{assumption}[Generator, I]\label{assump_bootstrap} Assume that:
\begin{enumerate}[label=(\alph*)]
\item the generator $\hat\mG_n: \cU \rightarrow \tilde\cZ_n$, introduced in \eqref{eq:generator}, is a function of $\{(\mZ_i,\mU_i)\}_{i \in [n]}$;

\item the random measure $\bP_{\tilde Z \mid \cO}$ has nonzero variance $\P_{\cO}$-almost surely, and satisfies
\[
\sW_1\big(\bP_{\tilde Z \mid \cO},\, \P_Z\big) = o_{\P_{\cO}}(1),
\]
where $\sW_1(\cdot,\cdot)$ denotes the Wasserstein-1 distance (in Euclidean metric space) and $\bP_{\tilde Z \mid \cO}$ is the conditional distribution of $\tilde \mZ = \hat\mG_n(\mU)$ given  
\[
\cO := \sigma(\mZ_1,\mZ_2,\ldots,\mU_1,\mU_2,\ldots),
\]
the $\sigma$-field generated by data and noise, with the associated probability measure $\P_\cO$; 

\item for all $n$,  $\tilde{\cZ}_n \subseteq \tilde{\cZ}$  $\P_{\cO}$-almost surely, where $\tilde{\cZ}$ is a nonrandom and compact subset of $\R^p$.
\end{enumerate}

\end{assumption}

\begin{assumption}[Objective function]\label{assump_l}
For any $\mz \in \cZ \cup \tilde{\cZ}$ and $\meta \in \cK$, assume that:
\begin{enumerate}[label=(\alph*)]
\item $\cK$ is convex and compact;
\item the map $\boldsymbol\eta \mapsto \sL(\boldsymbol\eta,\mz)$ is twice continuously differentiable on $\cK$;
\item letting $\rD^k_{\meta}$ be the $k$-th derivative with respect to $\meta$, the maps
\[
\mz \mapsto \sL(\boldsymbol\eta, \mz),\qquad
\mz \mapsto \rD_{\boldsymbol\eta}\sL(\boldsymbol\eta, \mz),\qquad
\mz \mapsto \rD_{\boldsymbol\eta}^2\sL(\boldsymbol\eta, \mz)
\]
are continuous on $\cZ \cup \tilde \cZ$;
\item Fisher's information matrix $- {\rm D}_{\boldsymbol\eta}^2\E[\sL(\boldsymbol\eta_0, \mZ)]$ is  invertible.
\end{enumerate}
\end{assumption}

\begin{assumption}[M-estimator]\label{assump_eta}
Assume that:
\begin{enumerate}[label=(\alph*)]
\item $\boldsymbol\eta_0$ is an interior point of $\cK$;
\item $\boldsymbol\eta_0$ uniquely maximizes $\E\Big[\sL(\boldsymbol\eta,\mZ)\Big]$;
\item the estimator $\hat\meta_n$ is consistent for $\meta_0$ in the sense that
$
\|\hat{\boldsymbol\eta}_n - \boldsymbol\eta_0\|_2
    \xrightarrow[\;]{\P_{\cO}} 0;
$
\item $\hat\meta_n$ is an approximate empirical maximizer in the sense that  
      $\hat\meta_n \in \cK$ and
\[
\frac1n \sum_{i=1}^n \sL(\hat{\boldsymbol\eta}_n, \mZ_i) + o_{\P_{\cO}}(n^{-1})
    \;\ge\;
    \sup_{\boldsymbol \eta \in \cK} \frac1n \sum_{i=1}^n \sL(\boldsymbol\eta, \mZ_i)
    ,
\]
allowing for numerical optimization error.
\end{enumerate}
\end{assumption}

\begin{assumption}[Bootstrap M-estimator]\label{assump_bsmest}
Assume that:
\begin{enumerate}[label=(\alph*)]
\item The bootstrap estimator satisfies
$
\|\tilde{\boldsymbol\eta}_n - \tilde{\boldsymbol\eta}_0\|_2
    \xrightarrow[\;]{\P_{\cO\tilde U}} 0,
$
where $\P_{\cO\tilde U}$ refers to the joint distribution of  
$(\mZ_1,\mZ_2,\ldots)$ and $(\mU_1,\mU_2,\ldots,\tilde\mU_1,\tilde\mU_2,\ldots)$;
\item $\tilde\meta_n$ is an approximate empirical maximizer in the bootstrap world,  
      in the sense that $\tilde\meta_n \in \cK$ and
\[
\frac1n \sum_{i=1}^n \sL\big(\tilde{\boldsymbol\eta}_n,\, \hat\mG_n(\tilde\mU_i)\big) + o_{\P_{\cO\tilde U}}(n^{-1})
    \;\ge\;
    \sup_{\boldsymbol{\eta} \in \cK}\frac1n \sum_{i=1}^n \sL\big({\boldsymbol\eta},\, \hat\mG_n(\tilde\mU_i)\big).
\]
\end{enumerate}
\end{assumption}

Assumption~\ref{assump_l} corresponds to the ``classical conditions'' described in, e.g., \citet[Chapter~5.6]{Vaart_1998}. Assumption~\ref{assump_eta} represents the standard identifiability and consistency condition for M-estimators, while Assumption~\ref{assump_bsmest} serves as its bootstrap analogue. These assumptions hold automatically under Glivenko--Cantelli conditions for the loss function $\sL$ over $\meta$; see \citet[Theorem~5.7]{Vaart_1998}. Section~\ref{sec:gai} will provide sufficient conditions under which the generative models discussed in Section~\ref{sec:example} satisfy Assumption~\ref{assump_bootstrap}.

Although it is in principle possible to establish consistency of the generative modeling-based bootstrap under weaker smoothness conditions than, for instance, Assumption~\ref{assump_l}, by appealing to more refined empirical process techniques \citep[Chapter~3.6]{van1996weak}, we believe that doing so would add limited additional insight. The present theory already fulfills its intended purpose and underscores the main message: generative modeling-based bootstraps can consistently recover the distribution of regular M-estimators with appropriate theoretical guarantees. In this sense, they provide a viable alternative to existing bootstrap procedures.

In detail, with the above assumptions, the following theorem gives the bootstrap consistency of $\tilde{\boldsymbol\eta}_n$ to approximate the distribution of $\hat{\boldsymbol\eta}_n$.

\begin{theorem}[Bootstrap consistency, regular M-estimators]\label{thm_reg}
Under Assumptions \ref{assump_ndata}-\ref{assump_bsmest}, we have 
$$
\sup_{\mt \in \R^q} \Big|\pr{\sqrt{n} (\tilde{\boldsymbol\eta}_n - \tilde{\boldsymbol\eta}_0) \leq \mt| \mathcal{O}} - \pr{\sqrt{n}(\hat{\boldsymbol\eta}_n - \boldsymbol\eta_0) \leq \mt}\Big| = o_{\P_{\cO}}(1).
$$
\end{theorem}

\begin{remark}
For reasons similar to those discussed prior to Theorem~\ref{thm_reg}, we do not attempt to extend Theorem~\ref{thm_reg} to high-dimensional regimes in which the data dimension $p$ is large relative to the sample size $n$. Instead, this setting is examined empirically in Section~\ref{sec:sim}. The simulation results reported there provide encouraging evidence that generative modeling-based bootstrap methods implemented via GANs and normalizing flows can match the best performance of Efron's original bootstrap, while being substantially less affected by the curse of dimensionality than the smoothed bootstrap (based on kernel density estimators).
\end{remark}

\section{Theory for isotonic regression: an irregular estimator} \label{sec:iso}

Inference for irregular estimators---those that typically fail to achieve root-$n$ consistency and do not admit a Gaussian limit---has long been a central topic in mathematical statistics and econometrics. Prominent examples include shape-constrained inference \citep{groeneboom2014nonparametric} and Manski-type estimators \citep{manski1981structural,cattaneo2020bootstrap}. 

Because the limiting distributions in such problems are often intricate, bootstrap-based methods are particularly attractive. However, it is now well understood that Efron’s original bootstrap is generally inconsistent in these settings. As a result, inference for irregular estimators is typically conducted using variants of the smoothed bootstrap \citep{kosorok2008bootstrapping,sen2010inconsistency,groeneboom2024confidence}.

This section contributes to this literature by analyzing a canonical irregular estimator, the isotonic regression estimator, for which a comprehensive (smoothed) bootstrap consistency theory appears to remain unavailable. In addition, we investigate the use of generative modeling-based bootstrap methods as an alternative to the traditional smoothed bootstrap. Specifically, we establish general conditions under which the generative modeling-based bootstrap consistently approximates the sampling distribution of the estimator.

In detail, isotonic regression concerns pairs $\mZ_i = (X_i, Y_i) \in \cX \times \cY$, for $i \in [n]$, which are assumed to be independent and identically distributed (i.i.d.), with marginal distributions $\P_X$ and $\P_Y$, and corresponding supports $\cX$ and $\cY$. Assume the regression model
\[
Y_i = f_0(X_i) + \xi_i,
\]
where $\cX$ is {\it known}, $f_0 : \cX \to \R$ is an unknown, fixed, and nondecreasing function,  the errors $\xi_i$ are i.i.d., independent of the $X_i$'s, and $\E[\xi_1]=0$. The isotonic regression estimator of $f_0$ is a {\it shape-constrained} least squares, given by
\[
\hat f_n 
:= 
\underset{f : \cX \to \R \text{ nondecreasing}}{\arg\min} 
\sum_{i=1}^n (Y_i - f(X_i))^2.
\]
A common inferential goal is to construct confidence intervals for $f_0(x_0)$  based on $\hat f_n(x_0)$.

Let $\sigma^2 := \E[\xi_i^2]$. The following facts are well known.
\begin{enumerate}[label=(\alph*)]
\item {\bf (Cube-root asymptotics)}  
      Supposing $X_1$ is distributed uniformly on $[0,1]$, $x_0 \in (0,1)$, and $f_0$ and the errors obey mild regularity conditions, the estimator $\hat f_n(x_0)$ satisfies a cube-root rate of convergence to $f_0(x_0)$ and converges weakly to a Chernoff-type distribution \citep{brunk1969estimation,han2022berry}. Specifically,
      \begin{align}\label{eq:weaklimit-iso}
      \Big(\frac{n}{\sigma^2}\Big)^{1/3}
      \big(\hat f_n(x_0) - f_0(x_0)\big)
      \;\text{ converges weakly to }\;
      \Big(\frac{f_0'(x_0)}{2}\Big)^{1/3} \cdot \D,
      \end{align}
      where $f_0'$ denotes the derivative of $f_0$, and $\D$ is the Chernoff distribution:  
      \[
      \D = 2 \cdot \arg\max_{t \in \R} \Big\{\B(t) - t^2\Big\},
      \]
      with $\B$ denoting the two-sided Brownian motion.

\item \textbf{(Failure of the classical bootstrap)}  
      Efron's original bootstrap fails for isotonic regression: the bootstrap limiting distribution of $\hat f_n(x_0)$ does not match the distribution in \eqref{eq:weaklimit-iso}. See, for example, \cite{groeneboom2024confidence} and references therein.
\end{enumerate}

We now propose a generative modeling-based bootstrap approach for statistical inference in isotonic regression.  
To this end, consider the bootstrap data
\[
(\tilde X_i, \tilde Y_i)^\top := \hat\mG_n(\tilde\mU_i), \qquad i \in [n],
\]
with marginal distributions $\bP_{\tilde X  \mid \cO}$ and $\bP_{\tilde Y \mid \cO}$, supports $\tilde \cX_n$ and $\tilde \cY_n$, respectively, together with the induced regression structure
\begin{align*}
\tilde Y_i &= \tilde f_0(\tilde X_i) + \tilde \xi_i,\\
\text{where} \qquad 
\tilde f_0(x) &:= \E\!\big[\tilde Y_i \mid \tilde X_i = x,\, \cO\big], 
\qquad 
\tilde \xi_i := \tilde Y_i - \tilde f_0(\tilde X_i).
\end{align*}
Note that $\tilde f_0$ is not necessarily nondecreasing, and the bootstrap residuals $\{\tilde \xi_i\}_{i\in[n]}$ are not conditionally independent of the covariates $\{\tilde X_i\}_{i\in[n]}$ given $\cO$. Furthermore, the supports $\tilde\cX_n$ and $\tilde\cY_n$ may not necessarily equal $\cX$ and $\cY$.

The bootstrap isotonic regression estimator is defined as
\begin{align}\label{eq:bootstrap-isotonic}
\tilde f_n 
:=
\underset{f : \cX \to \R \text{ nondecreasing}}{\arg\min}
\sum_{i=1}^n \Big(\tilde Y_i - f(\tilde X_i)\Big)^2 
\, \mathds{1}(\tilde X_i \in \cX),
\end{align}
where the indicator function in \eqref{eq:bootstrap-isotonic} ensures that the optimization only involves those $\tilde X_i \in \cX$, on which the function $f_0$ is well defined.
 Our goal is to show that, under suitably mild conditions, the conditional distribution of  
$\tilde f_n(x_0)$ consistently approximates the sampling distribution of $\hat f_n(x_0)$.  
This naturally requires additional regularity conditions on both the data-generating mechanism and the generator~$\hat\mG_n$. 

In what follows, $\sD_x$ and $\sD_x^2$ denote the first and second partial derivative, respectively, with respect to the first argument. 
\begin{assumption}[Data space, II]\label{assump_iso}
Assume that:
\begin{enumerate}[label=(\alph*)]
\item the regression function $f_0: \cX \rightarrow \R$ is twice continuously differentiable on $\cX$, with derivative uniformly bounded away from $0$;
\item $(X_1,\xi_1),(X_2,\xi_2),\ldots$ are i.i.d., each $X_i$ is independent of $\xi_i$, and  satisfy $\E[\xi_1]=0$,  ${\mathrm{Var}}(X_1) >0 $ and $\sigma^2 > 0$;
\item $\mZ_1$ admits a Lebesgue density $p_Z$ such that
$$
(x,y) \mapsto p_Z(x,y) \text{ is twice continuously differentiable on $\cX \times \R$,}
$$
\item $X_1$ admits a Lebesgue density $p_X$ such that $p_X$ is uniformly bounded away from zero on $\cX$; 
      \item $\cX, \cY$ are bounded, closed intervals satisfying $\cZ = \cX \times \cY$, and $x_0$ is an interior point of $\cX$.
\end{enumerate}
\end{assumption}

\begin{assumption}[Generator, II]\label{ass:generator}
Assume that, for all sufficiently large $n$, the conditional distribution of $\tilde \mZ = \hat\mG_n(\mU)$ given $\cO$ admits a Lebesgue density $\tilde p_n(\cdot) = \tilde p_n(\cdot \mid \cO)$ so that
\begin{enumerate}[label=(\alph*)]
\item the map 
$(x,y) \mapsto \tilde p_n(x, y)$ is always continuously differentiable and almost surely twice continuously differentiable on $\tilde\cX_n \times \R $;
\item $\tilde X_1$ admits a Lebesgue density $\tilde p_X$ such that some universal constant $\tilde K > 0$ exists, for which
\[
\tilde K^{-1} \le \tilde p_X(x) \text{ for all } x \in \cX,
~~~{\rm and}~~~
|\tilde p_n(\mz)| \vee \|{\rD} \tilde p_n(\mz)\|_2 \,\vee\, \|{\rD}^2 \tilde p_n(\mz)\|_{\rm op} \le \tilde K  
\text{ for all } \mz  \in \tilde\cZ_n,
\]
almost surely; 
\item $\tilde \cX_n$ is an interval satisfying $\cX \subseteq \tilde\cX_n $ almost surely. 
\end{enumerate}

\end{assumption}

Assumption~\ref{ass:generator} resembles the classical conditions imposed in the smoothed bootstrap literature to ensure bootstrap consistency for irregular estimators; see, for example, \citet[Section~4]{sen2010inconsistency}. Section~\ref{sec:gai} will provide sufficient conditions under which a class of flow-based generative models satisfies Assumption~\ref{ass:generator}. 

With the above assumptions, the following theorem establishes bootstrap consistency for the isotonic regression estimator. 

\begin{theorem}[Bootstrap consistency, isotonic regression] \label{thm:iso}Assume Assumptions \ref{assump_noise1},  \ref{assump_bootstrap}, \ref{assump_iso}, and \ref{ass:generator}. We then have
        $$
    \sup_{t \in \R} \ap{\pr{n^{1/3}(\tilde f_n(x_0) - \tilde f_0(x_0)) \leq t \cdl \mathcal{O}} -\pr{n^{1/3}(\hat f_n(x_0) - f_0(x_0)) \leq t }} = o_{\P_{\cO}}(1).
    $$
\end{theorem}

\section{GAN and flow bootstraps} \label{sec:gai}

This section gives sufficient conditions, under which the GAN- and flow-based generative models are able to satisfy the requirements in Theorems \ref{thm_reg} and \ref{thm:iso}. To this end, we first regulate the noise distribution $\P_U$.


\begin{assumption}[Noise space, II]\label{assump_noise2}
Assume that: 
\begin{enumerate}[label=(\alph*)]
\item the support of $\P_U$, $\cU$, is convex, compact, and contains $\mathbf{0}$, and that $\P_U$ admits a continuously differentiable Lebesgue density $p_U$ on $\cU$; 
\item there exists some constant $r_0 > 0$ such that $p_U$ is uniformly lower bounded away from $0$ on the set $\{\bu \in \R^p : \|\bu \|_2 \leq r_0\}.$
\end{enumerate}
\end{assumption}

The following slightly stronger condition is needed for Theorem \ref{thm:iso}, particularly concerning Assumption $\ref{ass:generator}$.

\begin{assumption}[Noise space, III]\label{assump_noise3}
Supposing that $\mU = (U_1, U_2)^\top \in \R^2$, assume that $p_U$ is twice-continuously differentiable on $\cU_1 \times \R$, where $\cU_1$ is the support of $U_1$.
\end{assumption}

\subsection{W-GAN}

This section demonstrates that suitably trained W-GANs in Example~\ref{exe:wgan} satisfy Assumption~\ref{assump_bootstrap}. Our analysis builds upon the theoretical results of \cite{biau2020some} on GANs and \cite{shen2023asymptotic} on the asymptotic properties of neural networks.

\begin{assumption}[W-GAN]\label{trainedwgan}
Assume that:
\begin{enumerate}[label=(\alph*)]
\item the activation function $\alpha: \R\rightarrow \R$  is $1$-Lipschitz with $\alpha(0) = 0$, and the neural network parameters  $L^{\rm gen}$ and $B^{\rm gen} $ are fixed positive constants;

\item the W-GAN is well trained in the sense that
\[
\sup_{D \in {\rm Lip_1}(p,1)}\Big\{
\frac1n \sum_{i=1}^n 
\sW(\hat\mG_n^{\rm GAN}, D, \mZ_i, \mU_i)
-
\frac1n \sum_{i=1}^n 
\sW(\hat\mG_n^{\rm GAN}, \hat D_n^{\rm GAN}, \mZ_i, \mU_i)\Big\}
=
o_{\P_{\cO}}(1),
\]
where ${\rm Lip}_1(p,1)$ denotes the set of all $1$-Lipschitz functions from $\R^p$ to $\R$, and enjoys the universal approximation property in the sense that
\[
\frac1n \sum_{i=1}^n 
\sW(\hat\mG_n^{\rm GAN}, \hat D_n^{\rm GAN}, \mZ_i, \mU_i)
=
o_{\P_{\cO}}(1);
\]
\item $\mathrm{Var}(\hat \mG_n^{\rm GAN}(\mU)\mid \cO) > 0$ holds $\P_{\cO}$-almost surely.
\end{enumerate}
\end{assumption}

\begin{theorem}[GAN bootstrap]\label{theorem:gan}
Assume that Assumptions~\ref{assump_ndata}, \ref{assump_noise1}, \ref{assump_noise2},  and \ref{trainedwgan} hold.  
Then Assumption~\ref{assump_bootstrap} is satisfied by the GAN generator.
\end{theorem}

\subsection{Affine autoregressive flows}

This section concerns affine autoregressive flows in Example \ref{exe:flow}. To facilitate the analysis, we focus on the following subclass of $\cF_{\nu}$ that encourages more regularity and is encouraged by \cite{irons2022triangular}:
\[ 
\cF_{\nu,K,M} 
:= 
\Big\{
\mF^\nu \circ \boldsymbol\Sigma_\nu \circ \cdots \circ \mF^1 \circ \boldsymbol\Sigma_1
\Big\}
\subset \cF_\nu,
\]
where, for each $i \in [\nu]$:
\begin{enumerate}[label=(\alph*)]
\item $\mF^i \in \cT_\uparrow(p)$, $\mF^i(\mathbf{0}) = \mathbf{0}$, and $\mF^i$ is three-times continuously differentiable on $\R^p$;
\item  $\boldsymbol{\Sigma}_{i}$ is symmetric, and satisfies $K^{-1} \le \lambda_{\rm min}(\boldsymbol\Sigma_i)$ and  $\lambda_{\rm max}(\boldsymbol\Sigma_i)\ \le K,$ where $\lambda_{\rm min}(\boldsymbol\Sigma_i)$ and $\lambda_{\rm max}(\boldsymbol\Sigma_i)$ denote the smallest and largest eigenvalues of $\boldsymbol\Sigma_i$, respectively;
\item the absolute values of all first-, second-, and third-order partial derivatives of $\mF^i$ are uniformly bounded above by~$M$ and $M^{-1} \leq   \inf_{\mz \in \R^p}  \sD_j F^i_j(\mz)$ for $j\in[p]$.
\end{enumerate}

Fixing $\nu, K, M >1$, we follow  \cite{irons2022triangular} and focus on the following more regular affine autoregressive flow as an alternative to $\hat\mS_n^{\rm flow}$ introduced in \eqref{eq:flow}:
\[
\hat\mG_n^{\rm rflow}
    = (\hat\mS_n^{\rm rflow})^{-1}~~\text{ with }
   \hat\mS_n^{\rm rflow} \in \argmax_{\mS \in \cF_{\nu,K,M}}
    \sum_{i=1}^n \Gamma(\mS, \mZ_i).
\]

\begin{assumption}[Affine autoregressive flow]\label{assump_flows}
Assume that, for all sufficiently large $n$,
\begin{enumerate}[label=(\alph*)]
\item the triangular flow $\hat\mG_n^{\rm rflow} = (\hat\mS_n^{\rm rflow})^{-1}$ is well trained in the sense that
      \[
      \frac{1}{n}\sum_{i=1}^n \Gamma(\hat\mS_n^{\rm rflow}, \mZ_i)
      -
      \frac{1}{n}\sum_{i=1}^n \log p_Z(\mZ_i)
      =
      o_{\P_{\cO}}(1);
      \]
      
      \item for $r_0 > 0$ defined in Assumption \ref{assump_noise2}, we have
      $$
      \cZ \subseteq \Big\{\mz \in \R^p : \|\mz\|_2 \leq (KpM)^{-\nu} r_0\Big\} \subseteq \Big\{\mz \in \R^p : \|\mz\|_2 \leq r_0\Big\} \subseteq \cU.
      $$
\end{enumerate}
\end{assumption}

\begin{theorem}[Flow bootstrap]\label{flow_wcvg}
Assume Assumptions~\ref{assump_noise1},~\ref{assump_noise2}~and~\ref{assump_flows}, and $\E[|\log p_Z(\mZ)|] < \infty$.
\begin{enumerate}[label=(\alph*)]
    \item \textbf{Regular estimators:}  
          If Assumption~\ref{assump_ndata} holds, then Assumption~\ref{assump_bootstrap} is satisfied by the flow generator.
    \item \textbf{Irregular estimators:}  
          If Assumptions \ref{assump_iso}~and~\ref{assump_noise3} hold,  
          then Assumptions~\ref{assump_bootstrap} and~\ref{ass:generator} are satisfied by the flow generator.
\end{enumerate}
\end{theorem}

Comparing Theorems \ref{theorem:gan} and \ref{flow_wcvg}, we can see that the flow bootstrap has stronger theoretical guarantees for irregular estimators.

\section{Simulation} \label{sec:sim}

\subsection{Methods and implementation}

This section complements the theoretical developments with illustrative empirical results. To this end, we compare four bootstrap procedures: Efron's original bootstrap, which resamples from the empirical measure; Efron's smoothed bootstrap, which resamples from a kernel density estimator of the underlying distribution; the GAN bootstrap (Example \ref{exe:wgan}); and the flow bootstrap (Example \ref{exe:flow}).

To implement the smoothed bootstrap, we employ the tophat kernel and select the bandwidth according to Silverman's rule of thumb. The kernel density estimator is fitted using the implementation provided in \texttt{scikit-learn} \citep{scikit-learn}.

To implement the GAN bootstrap, we specify both the generator and the discriminator as fully connected neural networks with fixed width $200$ and depth $6$ across all simulation settings. A dropout probability of $0.4$ is applied to all hidden layers (but not the input nor output layers) during training. The weight matrices in the generator are initialized with i.i.d.\ Gaussian entries with mean zero and variance $0.02$, and the bias vectors are initialized as $\mathbf{0}$. The same initialization scheme is adopted for the discriminator. Both networks are trained using full-batch ADAM with learning rate $0.0001$ and parameters $\beta_1 = 0.5$ and $\beta_2 = 0.9$. The training procedure follows Algorithm~1 of \cite{gulrajani2017improved}, implemented as in \cite{cao_NF_imp}. We train the generator for $2000$ steps; for each generator update, the discriminator is updated $5$ times. The gradient penalty coefficient used in discriminator training is set as $\lambda = 1$.

To implement the flow bootstrap, we adopt the GLOW architecture \citep{kingma2018glow} using the implementation provided in \cite{duan_NF_imp}. The flow model has depth $10$. The parameters in the ActNorm layers are initialized as $\mathbf{0}$. The neural networks used in the affine coupling layers are fully connected $\tanh$ networks with width $8$ and depth $3$, initialized using the default PyTorch initialization. We select the \texttt{RealNVP} option in the implementation. Each invertible convolution is initialized as a random orthogonal matrix. Training is performed using full-batch ADAM with learning rate $0.005$ and the default PyTorch values of $\beta_1$ and $\beta_2$. The flow model is trained for $1000$ steps.

All implementations are carried out in PyTorch \citep{paszke2019pytorch}. The code to reproduce all simulation results is available at \url{https://github.com/leonkt/generative_modeling_for_bootstrap}.

\subsection{Regular estimator: ordinary least squares}

In our first simulation setting, we generate the following independent base random variables:
\[
S \sim \mathrm{Unif}[-4,4], \qquad 
\epsilon_1, \ldots, \epsilon_{p-1} \sim \mathrm{Unif}[-0.5,0.5], 
\qquad \text{and} \qquad 
\epsilon_p \sim \mathrm{Unif}[-7,7].
\]
We then construct the predictor vector $\mX \in \R^{p-1}$ as follows. First, we independently sample
\[
X_j \sim \mathrm{Beta}(2,5), 
\qquad \text{for } j \in [5].
\]
For the remaining coordinates, we set
\[
X_j 
= \sin\Big\{ \frac{(j+1)S}{p} \Big\}
  + \cos\Big\{ \frac{(j+1)S}{p} \Big\}
  + \epsilon_j,
\qquad \text{for } j \in [p-1] \setminus [5].
\]

The response variable $Y$ is generated according to
\[
Y = \boldsymbol{\beta}_0^\top \mX + \epsilon_p,
\]
where the regression coefficient $\boldsymbol{\beta}_0 = (1,\ldots,1)^\top \in \R^{p-1}$ is the parameter of interest. The observed data vector is $\mZ = (\mX^\top, Y)^\top \in \R^{p}$. The estimator $\hat{\boldsymbol{\beta}}_n$ of $\boldsymbol{\beta}_0$ is the ordinary least squares (OLS) estimator. Without loss of generality, we do not include an intercept term in the OLS specification.

Table~\ref{table_regularols} reports the empirical coverage probabilities for the OLS estimator. We vary the dimension $p \in \{24,50,100\}$ and the sample size $n \in \{500,1000,2000\}$. The empirical coverage probabilities are computed using elliptical confidence regions based on $500$ Monte Carlo replications.

More specifically, in each replication and under each bootstrap scheme, we generate $1000$ bootstrap samples and compute the corresponding least squares estimator, denoted by $\tilde{\boldsymbol{\beta}}_n$, for each resample. For the smoothed, GAN, and flow bootstraps, we additionally draw $50{,}000$ samples from the learned bootstrap distribution and compute the least squares estimator based on this large synthetic sample, denoted by $\tilde{\boldsymbol{\beta}}_0$. For the original bootstrap, the confidence ball is centered at $\hat{\boldsymbol{\beta}}_n$, which we set $\tilde{\boldsymbol{\beta}}_0$ to be.

More specifically, the empirical coverage is assessed by computing, in each Monte Carlo replication, the statistic
\[
n \Big\| \hat{\boldsymbol{\beta}}_n - \boldsymbol{\beta}_0 \Big\|_2^2,
\]
and comparing it with the empirical $(100 \cdot \alpha)\%$-quantile of
\[
n \Big\| \tilde{\boldsymbol{\beta}}_n - \tilde{\boldsymbol{\beta}}_0 \Big\|_2^2,
\]
approximated using the $1000$ bootstrap samples. We consider significance levels $\alpha \in \{0.90, 0.95\}$.

It can be readily observed that, in most cases, the smoothed bootstrap exhibits substantial distortion. In contrast, the original bootstrap, the GAN bootstrap, and the flow bootstrap perform markedly better. We emphasize that this setting is known to favor Efron's original bootstrap \citep{mammen1993bootstrap}. Therefore, the fact that the GAN and flow bootstraps are able to match its performance is particularly revealing.

\subsection{Isotonic Regression}

In our second simulation setting, we generate bivariate data $\mZ = (X,Y)^\top$ such that
\[
X \sim \mathrm{Unif}[0,1], 
\qquad 
Y = X + \epsilon,
\qquad 
\epsilon \sim \mathrm{Unif}[-0.01,0.01],
\]
where $X$ and $\epsilon$ are independent. In this setting, the true regression function is $f_0(x) = x$. We fix $x_0 = 0.5$ and evaluate the empirical coverage probabilities of the $100 \cdot \alpha\%$ confidence intervals for $f_0(x_0)$ constructed using the four bootstrap schemes.

Table~\ref{table_irregulariso} reports the empirical coverage probabilities of the different bootstrap schemes for varying sample sizes. Coverage probabilities are computed over $500$ Monte Carlo replications.

In each replication and under each bootstrap scheme, we generate $1000$ bootstrap samples, retain those with covariates in $[0,1]$, and compute the isotonic regression estimator evaluated at $x_0$, denoted by $\tilde{f}_n(x_0)$, for each bootstrap sample. For the smoothed, GAN, and flow bootstraps, we additionally generate $50{,}000$ samples from the learned bootstrap distribution and define $\tilde{f}_0(x_0)$ as the local average
\[
\tilde{f}_0(x_0)
=
\frac{1}{\big| \{ 1 \le i \le 50000 : 0.4999 \le \tilde{X}_i \le 0.5001 \} \big|}
\sum_{i : 0.4999 \le \tilde{X}_i \le 0.5001} \tilde{Y}_i.
\]
For Efron's original bootstrap, we set $\tilde{f}_0(x_0) = \hat{f}_n(x_0)$, where $\hat{f}_n(x_0)$ is computed from the original data.

Empirical coverage is evaluated by computing, in each Monte Carlo replication, the statistic
\[
n^{1/3} \big( \hat{f}_n(x_0) - f_0(x_0) \big),
\]
and comparing it with the equal-tailed $(100 \cdot \alpha)\%$ confidence interval of
\[
n^{1/3} \big( \tilde{f}_n(x_0) - \tilde{f}_0(x_0) \big),
\]
approximated using the $1000$ bootstrap samples. We consider significance levels $\alpha \in \{0.90, 0.95\}$.

It can be seen that the original bootstrap fails in this setting, as expected. In contrast, the remaining three approaches—the smoothed, GAN, and flow bootstraps—all deliver satisfactory empirical coverages. 

{
\renewcommand{\tabcolsep}{5pt}
\renewcommand{\arraystretch}{1}
\begin{table}
\caption{Empirical coverage probabilities for the four bootstrap procedures, regular settings}\label{table_regularols}\vspace{8pt}
\centering
\begin{tabular}{cccccccccc}
\toprule
\multirow{2}{*}{$p$} & \multirow{2}{*}{$n$} &
\multicolumn{4}{c}{90\% Coverage} &
\multicolumn{4}{c}{95\% Coverage} \\
\cmidrule(r{6pt}){3-6} \cmidrule(l{6pt}){7-10}
& & original & smoothed & GAN & flow & original& smoothed & GAN & flow\\
\midrule
24 & 500      &   0.852 & 0.588 & 0.922 & 0.902&
       0.938 & 0.702 & 0.964 & 0.944 \\
& 1000 & 0.932 & 0.648 & 0.944 & 0.924 &
       0.972 & 0.780 & 0.984 & 0.978 \\  
& 2000 & 0.920 & 0.648 & 0.966 & 0.934 &
       0.984 & 0.770 & 0.996 & 0.972 \\
       \midrule
50 & 500      &   0.920 & 0.726 & 0.882 & 0.926&
       0.931 & 0.828 & 0.942 & 0.974 \\
& 1000 & 0.934 & 0.780 & 0.904 & 0.972 &
       0.954 & 0.856 & 0.944 & 0.982 \\  
& 2000 & 0.850 & 0.700 & 0.920 & 0.990 &
       0.924 & 0.804 & 0.982 & 0.996 \\
       \midrule
100 & 500      &   0.972 & 0.918 & 0.846 & 0.980&
       0.990 & 0.936 & 0.882 & 0.994 \\
& 1000 & 0.926 & 0.812 & 0.860 & 1.000 &
       0.932 & 0.910 & 0.902 & 1.000 \\  
& 2000 & 0.898 & 0.878 & 0.898 & 1.000 &
       0.992 & 0.898 & 0.934 & 1.000 \\
\bottomrule
\end{tabular}
\end{table}
}

{
\renewcommand{\tabcolsep}{5pt}
\renewcommand{\arraystretch}{1}
\begin{table}
\caption{Empirical coverage probabilities of for the four bootstrap procedures, irregular settings}\label{table_irregulariso}\vspace{8pt}
\centering
\begin{tabular}{ccccccccc}
\toprule
 \multirow{2}{*}{$n$} &
\multicolumn{4}{c}{90\% Coverage} &
\multicolumn{4}{c}{95\% Coverage} \\
\cmidrule(r{6pt}){2-5} \cmidrule(l{6pt}){6-9}
 & original & smoothed & GAN & flow & original& smoothed & GAN & flow\\
\midrule
 1000      &   0.644 & 0.966 & 0.916 & 0.930&
       0.702 & 0.984 & 0.948 & 0.974 \\
 2000 & 0.698 & 0.912 & 0.896 & 0.924 &
       0.762 & 0.970 & 0.926 & 0.958 \\  
 3000 & 0.700 & 0.894 & 0.904 & 0.898 &
       0.762 & 0.954 & 0.932 & 0.940 \\
\bottomrule
\end{tabular}
\end{table}
}

\section{Proofs of main theorems} \label{sec:proofmain}

We start this section with an introduction to additional notation and conventions. 
We use $\P$ as shorthand for the joint distribution of 
$(\mZ, \mZ_1,\mZ_2,\ldots)$, $(\mU, \mU_1,\mU_2,\ldots)$,
and $(\tilde\mU_1,\tilde\mU_2,\ldots)$. 
The bootstrap samples are denoted by $\tilde \mZ = \hat \mG_n(\mU)$ and $\tilde \mZ_{i} = \hat\mG_n(\tilde \mU_i)$, the latter of which will sometimes be written as $\tilde \mZ_{i,n}$ when we need to emphasize the dependence of the distribution on $n$. 
We use $\bP_{\mid \cO}$ to denote the (regular) conditional probability of $\P$ given $\cO$. 

Consider $\mV, \mW, \mV_1, \mV_2, ...$ to be some general random variables in $\R^r$.  The distribution of $\mV$ under $\P$ is written as $\P_V$. Its conditional distribution under $\bP_{\mid \cO}$ will be written as $\bP_{V \mid \cO}$. 
Take a non-random sequence of real numbers, $a_n > 0$, converging to $0$. We say $\mV_{n} = O(a_n)$ if $\limsup_{n\rightarrow \infty} \|\mV_n\|_2/a_n < \infty$ almost surely. We say $\mV_n = o(a_n)$ if $\lim_{n\rightarrow\infty} \|\mV_n\|_2/a_n \rightarrow 0$ almost surely. 
We say $\mV_{n} = O_{\P}(a_n)$ if for every $\epsilon > 0$, there exists an $M_\epsilon > 0 $ such that $\limsup_{n\rightarrow\infty} \P(\|\mV_n\|_2/a_n \leq M_\epsilon) > 1-\epsilon$. We say $\mV_n = \Theta_\P(a_n)$ if $\mV_n = O_\P(a_n)$ and, for every $\epsilon > 0$, there exists an $m_\epsilon > 0$ such that $\limsup_{n\rightarrow\infty}\P(\|\mV_n\|_2/a_n \leq m_\epsilon) < \epsilon$ We say $\mV_n = o_{\P}(a_n)$ if for every $\epsilon > 0$, $\lim_{n\rightarrow\infty} \P(\|\mV_n\|_2/a_n > \epsilon) = 0$. When $a_n = 1$ for all $n=1,2,...$, an alternate notation to $\mV_n=o_\P(1)$ is $\mV_n \overset{\P}{\rightarrow} 0$.

The space of continuous, real-valued functions on some compact set $\cE \subseteq \R^r$ is denoted by $C(\cE)$. It is turned into a measurable space by endowing it with the supremum norm and giving it the Borel $\sigma$-algebra. In the following, consider $G_0, G_1, G_2, ...$ as a sequence of $C(\cE)$-valued random variables.


If $V$ is a real-valued random variable, then we define
$$
\|V\|_{\bP_{\mid \cO},\psi_2} = \inf\{C >0 : \E[\exp(V^2/C^2) \mid \cO] \leq 2 \text{ almost surely}\}.
$$

For a normed space $\cN$ with norm $\|\cdot \|$ and a positive number $\epsilon > 0$, we denote the closed $\epsilon$-ball centered at $x \in \cN$ as
$$
\cB(x, \epsilon, \|\cdot \|) = \Big\{m \in \cN : \|x_i - m\| \leq \epsilon\Big\}.
$$
Then, for any subset $\cS \subseteq \cN$, we denote its $\epsilon$-covering number by $N(\epsilon, \cS, \|\cdot\|)$, namely, 
\[
N(\epsilon, \cS, \|\cdot\|):= \inf\Big\{N: \text{there exist some } \{x_i\}_{i=1}^N, x_i \in \cN \text{ such that } \cS \subseteq\bigcup_{i=1}^n \cB(x_i, \epsilon, \|\cdot\|)\Big\}.
\]


A sequence of random variables $\mV_1, \mV_2, ...$ is said to converge weakly to $\mV$ (resp. $G_1, G_2, ...$ converges weakly to $G_0$) if, for every bounded, continuous function $f : \R^r \rightarrow \R$, (resp. $f : C(\cE) \rightarrow \R$)
$$
\E[f(\mV_n)] - \E[f(\mV)] \rightarrow 0 \qquad \text{(resp. $\E[f(G_n)] - \E[f(G_0)] \rightarrow 0$})
$$
A sequence of random variables $\mV_1, \mV_2, ...$ is said to converge weakly to $\mV$ conditionally on $\cO$ (resp. $G_1, G_2,...$ converges weakly to $G_0$ conditionally on $\cO$) if, for every bounded continuous function $f : \R^r \rightarrow \R$, (resp. $f : C(\cE) \rightarrow \R$)
$$
\E[f(\mV_n) \mid \cO] - \E[f(\mV)\mid \cO] \rightarrow 0 \qquad \text{(resp. $\E[f(G_n) \mid \cO] - \E[f(G_0) \mid \cO] \rightarrow 0$})
$$ 
almost surely. 

 
In the context of Section \ref{sec:iso}, taking any $\ell, u > 0$, we define positive sequences $\ell_n = x_0 - \ell n^{-1/3}$ and $u_n = x_0 + un^{-1/3}$, and introduce the following notation for local averages:
\al{
\overline{Y}_{[\ell_n, u_n]} &:= \frac{1}{|\{i: \ell_n \leq X_i \leq u_n\}|} \sum_{i=1}^n  Y_{i} \cdot\mathds{1}(\ell_n \leq X_i\leq u_n),\\
\overline{\widetilde{Y}}_{[\ell_n, u_n]} &:=\frac{1}{|\{i: \widetilde X_i \in [\ell_n, u_n] \cap \cX\}|} \sum_{i=1}^n  \widetilde Y_{i} \cdot\mathds{1}(\widetilde X_i \in [\ell_n, u_n] \cap \cX),\\
\overline{f}_{[\ell_n, u_n]} &:= \frac{1}{|\{i: \ell_n \leq X_i\leq u_n\}|} \sum_{i=1}^n  f_0(X_{i}) \cdot\mathds{1}(\ell_n \leq X_i \leq u_n),\\
\overline{\widetilde{f}}_{[\ell_n, u_n]} &:=\frac{1}{|\{i: \widetilde X_i \in [\ell_n, u_n] \cap \cX\}|} \sum_{i=1}^n  \widetilde f_0(\widetilde X_i) \cdot\mathds{1}(\tilde X_i \in [\ell_n, u_n] \cap \cX),\\
\overline{\xi}_{[\ell_n, u_n]} &:= \frac{1}{|\{i: \ell_n \leq X_i  \leq u_n\}|} \sum_{i=1}^n \xi_{i} \cdot\mathds{1}(\ell_n \leq X_i \leq u_n),\\
\overline{\tilde{\xi}}_{[\ell_n, u_n]} &:=\frac{1}{|\{i: \tilde X_i \in [\ell_n, u_n] \cap \cX\}|} \sum_{i=1}^n  \tilde \xi_{i} \cdot\mathds{1}(\tilde X_i \in [\ell_n, u_n] \cap \cX).
}
In addition, for any $\ell > 0$ and $u \ge 0$, define
\[
G_{\ell,u}
:=
\frac{\sigma}{\sqrt{p_X(x_0)}}
\cdot
\frac{\bB(u) - \bB(-\ell)}{u+\ell}
+
\frac{f_0'(x_0)}{2}(u-\ell).
\]

\subsection{Proof of Theorem \ref{thm_reg}}
\begin{proof}
    We appeal to Lemma \ref{lemma:prob_and_as} and thus, in the following, we will be explicit about the dependence of $\tilde{\boldsymbol{\eta}}_0$ on $n$, and denote it as $\tilde{\boldsymbol{\eta}}_{0,n}$. Take any subsequence $n_k$. It suffices to find a subsequence $n_{k_\ell}$ such that 
    $$
    \sup_{\mt \in \R^q} \Big|\pr{\sqrt{n_{k_\ell}} (\tilde{\boldsymbol\eta}_{n_{k_\ell}} - \tilde{\boldsymbol\eta}_{0,n_{k_\ell}}) \leq \mt| \mathcal{O}} - \pr{\sqrt{n_{k_\ell}}(\hat{\boldsymbol\eta}_{n_{k_\ell}} - \boldsymbol\eta_0) \leq \mt}\Big| \rightarrow 0
    $$
    almost surely. By Polya's theorem (Lemma \ref{lemma:polya_cdl}), it suffices to show that there is a subsequence $n_{k_\ell}$ and a Gaussian random variable $\mZ'$ such that $\sqrt{n_{k_\ell}} (\tilde{\boldsymbol\eta}_{n_{k_\ell}} - \tilde{\boldsymbol\eta}_{0,n_{k_\ell}})$ converges weakly to $\mZ'$ conditionally on $\cO$, and $\sqrt{n_{k_\ell}}(\hat{\boldsymbol\eta}_{n_{k_\ell}} - \boldsymbol\eta_0)$ converges weakly to $\mZ'$. Equivalently, by the Cramer-Wold device (Lemma \ref{lemma:cramerwold_cdl}), it suffices to show that, for any $\ma \in \R^q$, we have $\sqrt{n_{k_\ell}} \ma^\top(\tilde{\boldsymbol\eta}_{n_{k_\ell}} - \tilde{\boldsymbol\eta}_{0,n_{k_\ell}})$ converges weakly to $\ma^\top\mZ'$  conditionally on $\cO$ and $\sqrt{n_{k_\ell}}\ma^\top(\hat{\boldsymbol\eta}_{n_{k_\ell}} - \boldsymbol\eta_0)$ converges weakly to $\ma^\top\mZ'$.
    
~\\    
\textbf{Step 1: Reduce to an appropriate, almost surely converging subsequence $n_{k_\ell}$.}
  
Take any $\ma \in \R^q$. 
    By Lemma \ref{cdllinear}, $    \P(\sD_{\boldsymbol{\eta}}^2\E[\sL(\tilde{\boldsymbol\eta}_{0,n}, \tilde\mZ_{1,n}) \mid \cO]) \text{ is invertible}) \rightarrow 1$, and 
    $$
    \|\sD_{\boldsymbol{\eta}}^2\E[\sL(\tilde{\boldsymbol\eta}_{0,n}, \tilde\mZ_{1,n}) \mid \cO]^{-1} - \sD_{\boldsymbol{\eta}}^2\E[\sL(\boldsymbol\eta_0, \mZ)]^{-1}\|_{\rm op} = o_{\P}(1) ~~~\text{and}~~~ \|\tilde{\boldsymbol\eta}_{0,n} - \boldsymbol{\eta}_0\|_2 = o_{\P}(1).
    $$

    Furthermore,  $(\boldsymbol{\eta}, \mz) \mapsto \sD_{\boldsymbol{\eta}}\sL(\boldsymbol{\eta}, \mz)$ is continuous on $\cK \times (\cZ \cup \tilde\cZ)$ from Assumption \ref{assump_l}(b,c). Therefore,
\begin{equation}\label{eq:regthm:abovefact}
\sup_{\boldsymbol{\eta} \in \cK}\|{\rm Var}(\sD_{\boldsymbol{\eta}}\sL(\boldsymbol{\eta}, \tilde\mZ_{1,n}) \mid \cO) - {\rm Var}(\sD_{\boldsymbol{\eta}}\sL(\boldsymbol{\eta}, \mZ))\|_{\rm max} = o_{\P}(1)
\end{equation}
by Lemma \ref{lemma:w2cvg}. Continuity of $\boldsymbol{\eta} \mapsto  {\rm Var}(\sD_{\boldsymbol{\eta}}\sL(\boldsymbol{\eta}, \mZ))$ on $\cK$, by the bounded convergence theorem, combined with \eqref{eq:regthm:abovefact} implies
\al{
&\|{\rm Var}(\sD_{\boldsymbol{\eta}}\sL(\tilde{\boldsymbol{\eta}}_{0,n}, \tilde\mZ_{1,n}) \mid \cO) - {\rm Var}(\sD_{\boldsymbol{\eta}}\sL(\boldsymbol{\eta}_0, \mZ))\|_{\rm max} \\
\leq&  \|{\rm Var}(\sD_{\boldsymbol{\eta}}\sL(\tilde{\boldsymbol{\eta}}_{0,n}, \tilde\mZ_{1,n}) \mid \cO) - {\rm Var}(\sD_{\boldsymbol{\eta}}\sL(\tilde{\boldsymbol{\eta}}_{0,n}, \mZ))\|_{\rm max} + \|{\rm Var}(\sD_{\boldsymbol{\eta}}\sL(\tilde{\boldsymbol{\eta}}_{0,n}, \mZ)) - {\rm Var}(\sD_{\boldsymbol{\eta}}\sL({\boldsymbol{\eta}}_{0}, \mZ))\|_{\rm max} \\
=& o_{\P}(1) + \|{\rm Var}(\sD_{\boldsymbol{\eta}}\sL(\tilde{\boldsymbol{\eta}}_{0,n}, \mZ)) - {\rm Var}(\sD_{\boldsymbol{\eta}}\sL({\boldsymbol{\eta}}_{0}, \mZ))\|_{\rm max}
\tag{Lemma \ref{lemma:w2cvg}}\\
=& o_{\P}(1). \tag{$\|\tilde{\boldsymbol\eta}_{0,n} - \boldsymbol{\eta}_0\|_2 = o_{\P}(1)$ and continuity of $\boldsymbol{\eta} \mapsto {\rm Var}(\sD_{\boldsymbol{\eta}}\sL(\boldsymbol{\eta}, \mZ))$ on $\cK$}
}

Appealing to Lemma \ref{lemma:prob_and_as}, we select a subsequence $n_{k_\ell}$ such that $\sD_{\boldsymbol{\eta}}^2\E[\sL(\tilde{\boldsymbol\eta}_{0,n_{k_\ell}}, \tilde\mZ_{1,n_{k_\ell}}) \mid \cO]^{-1}$ exists for each $n_{k_\ell}$, almost surely, 
    $$
    \|{\rm Var}(\sD_{\boldsymbol{\eta}}\sL(\tilde{\boldsymbol{\eta}}_{0,n_{k_\ell}}, \tilde\mZ_{1,n_{k_\ell}}) \mid \cO) - {\rm Var}(\sD_{\boldsymbol{\eta}}\sL(\boldsymbol{\eta}_0, \mZ))\|_{\rm max} \rightarrow 0
    $$
    almost surely, 
    $$
    \|\sD_{\boldsymbol{\eta}}^2\E[\sL(\tilde{\boldsymbol\eta}_{0,n_{k_\ell}}, \tilde\mZ_{1,n_{k_\ell}}) \mid \cO]^{-1} - \sD_{\boldsymbol{\eta}}^2\E[\sL(\boldsymbol\eta_0, \mZ)]^{-1}\|_{\rm op} \rightarrow 0, ~~~\text{and}~~~ \|\tilde{\boldsymbol\eta}_{0,n_{k_\ell}} - \boldsymbol{\eta}_0\|_2 \rightarrow 0
    $$
    almost surely.  Applying Lemma \ref{cdllinear} again, we obtain
    $$
    \sqrt{n_{k_\ell}}\ma^\top(\hat{\boldsymbol\eta}_{n_{k_\ell}} - \boldsymbol\eta_0) = -\frac{1}{\sqrt{n_{k_\ell}}} \sum_{i=1}^{n_{k_\ell}} \ma^\top(\sD_{\boldsymbol{\eta}}^2\E[\sL(\boldsymbol\eta_0, \mZ)])^{-1} \sD_{\boldsymbol{\eta}}\sL(\boldsymbol\eta_0, \mZ_i) + o_{\P}(1)
    $$
    and
    $$
    \sqrt{n_{k_\ell}}\ma^\top(\tilde{\boldsymbol\eta}_{n_{k_\ell}} -\tilde{\boldsymbol\eta}_{0,n_{k_\ell}}) = -\frac{1}{\sqrt{n_{k_\ell}}} \sum_{i=1}^{n_{k_\ell}} \ma^\top(\sD_{\boldsymbol{\eta}}^2\E[\sL(\tilde{\boldsymbol\eta}_{0,n_{k_\ell}}, \tilde\mZ_{1,n_{k_\ell}}) \mid \cO])^{-1}\sD_{\boldsymbol{\eta}}\sL(\tilde{\boldsymbol\eta}_{0,n_{k_\ell}}, \tilde \mZ_{i,n_{k_\ell}}) + o_{\P}(1).
    $$

~\\
\textbf{Step 2: Apply the central limit theorem to the linear representation of $\sqrt{n_{k_\ell}}\ma^\top(\hat{\boldsymbol\eta}_{n_{k_\ell}} - \boldsymbol\eta_0)$.}

     The function $\mz \mapsto \sD_{\boldsymbol{\eta}}\sL(\boldsymbol\eta_0, \mz)$ is bounded on $\cZ$, since $\mz \mapsto \sD_{\boldsymbol{\eta}}\sL(\boldsymbol\eta_0, \mz)$ is continuous by Assumption \ref{assump_l}(c) and $\cZ$ is compact by Assumption \ref{assump_ndata}(b). The Cauchy-Schwarz inequality implies
     $$
     \E[\ma^\top(\sD_{\boldsymbol{\eta}}^2\E[\sL(\boldsymbol\eta_0, \mZ)])^{-1}\sD_{\boldsymbol{\eta}}\sL(\boldsymbol\eta_0, \mZ)\sD_{\boldsymbol{\eta}}\sL(\boldsymbol\eta_0, \mZ)^\top(\sD_{\boldsymbol{\eta}}^2\E[\sL(\boldsymbol\eta_0, \mZ)])^{-1} \ma] < \infty.
     $$ Furthermore, Assumption \ref{assump_ndata}(a) implies that
    $$  \ma^\top(\sD_{\boldsymbol{\eta}}^2\E[\sL(\boldsymbol\eta_0, \mZ)])^{-1} \sD_{\boldsymbol{\eta}}\sL(\boldsymbol\eta_0, \mZ_{n_{k_1}}), \ma^\top(\sD_{\boldsymbol{\eta}}^2\E[\sL(\boldsymbol\eta_0, \mZ)])^{-1} \sD_{\boldsymbol{\eta}}\sL(\boldsymbol\eta_0, \mZ_{n_{k_2}}), ...
    $$
    are independent, identically distributed and square integrable.
    Define $\mZ' \in \R^q$ as a Gaussian random variable with mean zero and variance matrix
    $$
    \boldsymbol{\Sigma} := \sD_{\boldsymbol{\eta}}^2\E[ \sL(\boldsymbol{\eta}_0, \mZ)]^{-1}\E[\sD_{\boldsymbol{\eta}}\sL(\boldsymbol\eta_0, \mZ)\sD_{\boldsymbol{\eta}}\sL(\boldsymbol\eta_0, \mZ)^\top]\sD_{\boldsymbol{\eta}}^2\E[ \sL(\boldsymbol{\eta}_0, \mZ)]^{-1}.
    $$
    The central limit theorem indicates that $\sqrt{n_{k_\ell}}\ma^\top(\hat{\boldsymbol\eta}_{n_{k_\ell}} - \boldsymbol\eta_0)$ converges weakly to a mean-zero Gaussian random variable with variance matrix
    $$
    \ma^\top\E[\sD_{\boldsymbol{\eta}}^2 \sL(\boldsymbol{\eta}_0, \mZ)]^{-1}\E[\sD_{\boldsymbol{\eta}}\sL(\boldsymbol\eta_0, \mZ)\sD_{\boldsymbol{\eta}}\sL(\boldsymbol\eta_0, \mZ)^\top]\E[\sD_{\boldsymbol{\eta}}^2 \sL(\boldsymbol{\eta}_0, \mZ)]^{-1}\ma = \ma^\top \boldsymbol{\Sigma}\ma.
    $$
    This is exactly the variance matrix of  $\ma^\top \mZ'$, so we have shown that $\sqrt{n_{k_\ell}}\ma^\top(\hat{\boldsymbol\eta}_{n_{k_\ell}} - \boldsymbol\eta_0)$ converges weakly to $\ma^\top \mZ'$.

~\\
\textbf{Step 3: Apply the Lyapunov central limit theorem (Lemma \ref{lemma:lyapunov_cdl})  to the linear representation of $\sqrt{n_{k_\ell}}\ma^\top(\tilde{\boldsymbol\eta}_{n_{k_\ell}} -\tilde{\boldsymbol\eta}_{0,n_{k_\ell}})$.}

 By Assumptions \ref{assump_noise1} and \ref{assump_bootstrap}(a), the random variables $\{\tilde\mZ_{i,n}\}_{n\geq 1, 1\leq i \leq n}$ form a triangular array conditional on $\cO$, so as the random variables
    $$
    \Big\{\ma^\top(\sD_{\boldsymbol{\eta}}^2\E[\sL(\tilde{\boldsymbol\eta}_{0,n_{k_\ell}}, \tilde\mZ_{1,n_{k_\ell}}) \mid \cO])^{-1}\sD_{\boldsymbol{\eta}}\sL(\tilde{\boldsymbol\eta}_{0,n_{k_\ell}}, \tilde \mZ_{i,n_{k_\ell}})\Big\}_{\ell \geq 1, 1 \leq i \leq n_{k_\ell}}.
    $$
    By the Cauchy-Schwarz inequality, for any $\check p \geq 1$, 
    \al{
    \E\Bigg[\ap{\E[\sD_{\boldsymbol{\eta}}^2\sL(\tilde{\boldsymbol\eta}_{0,n_{k_\ell}}, \tilde\mZ_{1,n_{k_\ell}}) \mid \cO])^{-1}\sD_{\boldsymbol{\eta}}\sL(\tilde{\boldsymbol\eta}_{0,n_{k_\ell}}, \tilde\mZ_{1,n_{k_\ell}})^\top \ma }^{\check p}\mid \cO\Bigg] \leq \\     \E\Bigg[\ap{\|\ma\|_2 \cdot \Big\|\E[\sD_{\boldsymbol{\eta}}^2\sL(\tilde{\boldsymbol\eta}_{0,n_{k_\ell}}, \tilde\mZ_{1,n_{k_\ell}}) \mid \cO])^{-1} \Big\|_{\rm op} \cdot \|\sD_{\boldsymbol{\eta}}\sL(\tilde{\boldsymbol\eta}_{0,n_{k_\ell}}, \tilde\mZ_{1,n_{k_\ell}})\|_2}^{\check p} \mid \cO \Bigg],
    }
which is further upper bounded by
    \begin{equation}
        \pa{\|\ma\|_2 \cdot \Big\|\E[\sD_{\boldsymbol{\eta}}^2\sL(\tilde{\boldsymbol\eta}_{0,n_{k_\ell}}, \tilde\mZ_{1,n_{k_\ell}}) \mid \cO])^{-1} \Big\|_{\rm op} \cdot \sup_{(\boldsymbol\eta, \mz) \in \cK \times \tilde\cZ}\|\sD_{\boldsymbol{\eta}}\sL(\boldsymbol\eta, \mz)\|_2}^{\check p}. \label{eq:secondthirdfinite}
    \end{equation}
     For notational concision, let
    $$
    \tilde V_{i, n_{k_\ell}}:= \ma^\top(\sD_{\boldsymbol{\eta}}^2\E[\sL(\tilde{\boldsymbol\eta}_{0,n_{k_\ell}}, \tilde\mZ_{1,n_{k_\ell}}) \mid \cO])^{-1}\sD_{\boldsymbol{\eta}}\sL(\tilde{\boldsymbol\eta}_{0,n_{k_\ell}}, \tilde \mZ_{i,n_{k_\ell}}) ~~~~~~\text{and}~~~~~~ S_{n_{k_\ell}}^2 := \sum_{i=1}^{n_{k_\ell}}  {\textrm{Var}}(\tilde V_{i, n_{k_\ell}} \mid \cO).
    $$
    When $\check p = 2$,  \eqref{eq:secondthirdfinite} verifies that
    $$
    \E\Big[\tilde V_{i, n_{k_\ell}}^2 | \cO\Big] < \infty
    $$
    almost surely. Also, $\E\Big[\tilde V_{i, n_{k_\ell}}^2 | \cO\Big]  >0$ due to Assumption \ref{assump_bootstrap}(b) that guarantees that $\tilde\mZ_{1}$ has nonzero variance, conditional on $\cO$, almost surely.
    We then check that the Lyapunov condition 
    $$
    \frac{1}{S_{n_{k_\ell}}^3} \sum_{i=1}^{n_{k_\ell}} \E\Big[|\tilde V_{i, n_{k_\ell}}|^3 \mid \cO\Big] \rightarrow 0
    $$
    holds almost surely, as $\ell\rightarrow \infty$. Since
    $\{\tilde V_{i, n_{k_l}}\}_{i=1}^{n_{k_l}}$ are identically distributed, the equivalent condition to check is that
    $$
    \frac{n_{k_\ell} \cdot \E[|\tilde V_{1, n_{k_\ell}}|^3 \mid \cO]}{n_{k_\ell}^{3/2}{\textrm{Var}}(\tilde V_{1, n_{k_\ell}} \mid \cO)^{3/2}} \rightarrow 0
    $$
    holds almost surely, as $\ell \rightarrow \infty$. For all $\ell \geq 1$,
    \al{
    |\tilde V_{1, n_{k_l}}|^3 \leq         \pa{\|\ma\|_2 \cdot \Big\|\E[\sD_{\boldsymbol{\eta}}^2\sL(\tilde{\boldsymbol\eta}_{0,n_{k_\ell}}, \tilde\mZ_{1,n_{k_\ell}}) \mid \cO])^{-1} \Big\|_{\rm op} \cdot \sup_{(\boldsymbol\eta, \mz) \in \cK \times \tilde\cZ}\|\sD_{\boldsymbol{\eta}}\sL(\boldsymbol\eta, \mz)\|_2}^3  
    }
    by \eqref{eq:secondthirdfinite} with $\check p=3$. Using the property of our subsequence, and the fact that $\|\mathbf{A}\|_{\rm op} \leq q\|\mathbf{A}\|_{\rm max}$ for matrices $\mathbf A \in \R^{q \times q}$,
    $$
    \|\E[\sD_{\boldsymbol{\eta}}^2\sL(\tilde{\boldsymbol\eta}_{0,n_{k_\ell}}, \tilde\mZ_{1,n_{k_\ell}}) \mid \cO]^{-1} \|_{\rm op} \rightarrow \|\E[\sD_{\boldsymbol{\eta}}^2\sL(\boldsymbol\eta_{0}, \mZ)]^{-1} \|_{\rm op} < \infty
    $$
    almost surely, by invertibility of the right-hand side by Assumption \ref{assump_l}(d). Handling the denominator using Lemma \ref{lemma:w2cvg} and our choice of subsequence, ${\textrm{Var}}(\tilde V_{1, n_{k_\ell}} \mid \cO)$ converges to a nonzero constant almost surely.
    
    Therefore, sending $\ell \rightarrow \infty$, we establish that
    $$
        \frac{n_{k_\ell} \cdot \E[|\tilde V_{1, n_{k_\ell}}|^3 \mid \cO]}{n_{k_\ell}^{3/2}{\textrm{Var}}(\tilde V_{1, n_{k_\ell}} \mid \cO)^{3/2}} \rightarrow 0
    $$
    almost surely. By the Lyapunov central limit theorem, Lemma \ref{lemma:lyapunov_cdl},
    $\sum_{i=1}^{n_{k_\ell}} \tilde{V}_{i,n_{k_\ell}}/ S_{n_{k_\ell}}$ then converges weakly to a standard normal conditionally on $\cO$ . 

~\\    
\textbf{Step 4: Verify the limiting distributions are the same, and conclude.}

    Notice that
    \begin{equation}\label{eq:regthm:therhs}
    \frac{1}{n_{k_\ell}} S_{n_{k_\ell}}^2  = {\textrm{Var}}(\tilde V_{1, n_{k_\ell}} \mid \cO)
    \end{equation}
by the identical distribution of $\{\tilde V_{i, n_{k_\ell}}\}_{i=1}^{n_{k_\ell}}$ conditional on $\cO$. The right-hand side of \ref{eq:regthm:therhs} is, by definition, 
$$
\ma^\top(\sD_{\boldsymbol{\eta}}^2\E[\sL(\tilde{\boldsymbol\eta}_{0,n_{k_\ell}}, \tilde\mZ_{1,n_{k_\ell}}) \mid \cO])^{-1}\E[\sD_{\boldsymbol{\eta}}\sL(\tilde{\boldsymbol\eta}_{0,n_{k_\ell}}, \tilde \mZ_{i,n_{k_\ell}}) \sD_{\boldsymbol{\eta}}\sL(\tilde{\boldsymbol\eta}_{0,n_{k_\ell}}, \tilde \mZ_{i,n_{k_\ell}})^\top \mid \cO]  (\sD_{\boldsymbol{\eta}}^2\E[\sL(\tilde{\boldsymbol\eta}_{0,n_{k_\ell}}, \tilde\mZ_{1,n_{k_\ell}}) \mid \cO])^{-1}\ma.
$$
First, 
\al{
\|\sD_{\boldsymbol{\eta}}^2\E[\sL(\tilde{\boldsymbol\eta}_{0,n_{k_\ell}}, \tilde\mZ_{1,n_{k_\ell}}) \mid \cO]^{-1} - \sD_{\boldsymbol{\eta}}^2\E[\sL(\boldsymbol\eta_0, \mZ)]^{-1}\|_{\rm max} &\leq \|\sD_{\boldsymbol{\eta}}^2\E[\sL(\tilde{\boldsymbol\eta}_{0,n_{k_\ell}}, \tilde\mZ_{1,n_{k_\ell}}) \mid \cO]^{-1} - \sD_{\boldsymbol{\eta}}^2\E[\sL(\boldsymbol\eta_0, \mZ)]^{-1}\|_{\rm op}\\
&\rightarrow 0 
}
almost surely, by the definition of $n_{k_\ell}$. In addition,
$$
\|{\rm Var}(\sD_{\boldsymbol{\eta}}\sL(\tilde{\boldsymbol\eta}_{0,n_{k_\ell}}, \tilde\mZ_{1,n_{k_\ell}}) \mid \cO) - {\rm Var}(\sD_{\boldsymbol{\eta}}\sL(\boldsymbol{\eta}_0, \mZ))\|_{\rm max} \rightarrow 0
$$
almost surely by the definition of $n_{k_\ell}$. Since $\sD_{\boldsymbol{\eta}}\sL(\tilde{\boldsymbol\eta}_{0,n_{k_\ell}}, \tilde\mZ_{1,n_{k_\ell}})$ has zero mean conditional on $\cO$, and $\sD_{\boldsymbol{\eta}}\sL(\boldsymbol{\eta}_0, \mZ)$ is zero mean,
$$
{\rm Var}(\sD_{\boldsymbol{\eta}}\sL(\tilde{\boldsymbol\eta}_{0,n_{k_\ell}}, \tilde\mZ_{1,n_{k_\ell}}) \mid \cO) = \E[\sD_{\boldsymbol{\eta}}\sL(\tilde{\boldsymbol\eta}_{0,n_{k_\ell}}, \tilde \mZ_{i,n_{k_\ell}}) \sD_{\boldsymbol{\eta}}\sL(\tilde{\boldsymbol\eta}_{0,n_{k_\ell}}, \tilde \mZ_{i,n_{k_\ell}})^\top \mid \cO]
$$
and
$$
{\rm Var}(\sD_{\boldsymbol{\eta}}\sL(\boldsymbol{\eta}_0, \mZ)) =  \E[\sD_{\boldsymbol{\eta}}\sL(\boldsymbol\eta_0, \mZ)\sD_{\boldsymbol{\eta}}\sL(\boldsymbol\eta_0, \mZ)^\top].
$$
Hence, 
$$
n_{k_\ell}^{-1/2}S_{n_{k_\ell}}/\sqrt{\ma^\top \boldsymbol\Sigma \ma} \rightarrow 1
$$
almost surely, so that Slutsky's theorem implies 
$$
n_{k_\ell}^{-1/2}\sum_{i=1}^{n_{k_\ell}} \tilde V_{i, n_{k_\ell}}\Big/ \sqrt{\ma^\top \boldsymbol{\Sigma} \ma}
$$
converges weakly to a standard normal  conditionally on $\cO$. Applying the continuous mapping theorem, 
$n_{k_\ell}^{-1/2}\sum_{i=1}^{n_{k_\ell}} \tilde V_{i, n_{k_\ell}}
$
converges weakly to $\sqrt{\ma^\top \boldsymbol{\Sigma}\ma} $ times a standard normal, conditionally on $\cO$. This limit is equal in distribution to $\ma^\top \mZ'$. Since $\ma$ was arbitrary, by the Cramer-Wold device we conclude $\sqrt{n_{k_\ell}}(\tilde{\boldsymbol\eta}_{n_{k_\ell}} -\tilde{\boldsymbol\eta}_{0,n_{k_\ell}})$ converges to $\mZ'$ conditionally on $\cO$ . 
By comparing the weak limits, Polya's theorem, and Lemma \ref{lemma:prob_and_as}, we establish
$$
    \sup_{\mt \in \R^q} \Big|\pr{\sqrt{n} (\tilde{\boldsymbol\eta}_{n} - \tilde{\boldsymbol\eta}_{0}) \leq \mt| \mathcal{O}} - \pr{\sqrt{n}(\hat{\boldsymbol\eta}_{n} - \boldsymbol\eta_0) \leq \mt}\Big| = o_{\P}(1).
$$
This completes the proof.
\end{proof}

\subsection{Proof of Theorem \ref{thm:iso}}
\begin{proof}
Take any $\epsilon > 0$. For the random variables $(\tilde L^* , \tilde U^*)$ defined in Lemma \ref{lemma:ul_o1}, and $(L_G^*, U_G^*)$ defined in Lemma \ref{lemma:glu_exists_tight}, we choose a $K_\epsilon > 0$ large enough and $k_\epsilon >0$ small enough so that
$$
\limsup_{n\rightarrow\infty}\Big\{\P(\tilde L^* > K_\epsilon) \vee \P( \tilde L^* < k_\epsilon) \vee \P(L_G^* > K_\epsilon) \vee \P(L_G^* <k_\epsilon)\Big\}   < \epsilon
$$
and
$$
\limsup_{n\rightarrow\infty}\Big\{\P(\tilde U^* > K_\epsilon) \vee \P( \tilde U^* < k_\epsilon ) \vee \P(U_G^* > K_\epsilon)\vee\P( U_G^* < k_\epsilon)\Big\} < \epsilon.
$$
By Lemma \ref{lemma:fidi_uclt},
$$
    \sup_{t \in \R, k_\epsilon \leq u,\ell \leq K_\epsilon} \ap{\pr{n^{1/3}(\overline {\tilde Y}_{[\ell_n, u_n]} - \tilde f_0(x_0)) \leq t \cdl \mathcal{O}} - \pr{G_{\ell, u}\leq t)}} = o_{\P}(1).
$$
In particular, we work on the event $ \{k_\epsilon \leq \tilde U^* , \tilde L^*, U_G^*, L_G^* \leq K_\epsilon\}$.We use the continuous mapping theorem on the space of random functions with bounded sample paths on $[k_\epsilon,K_\epsilon] \times [k_\epsilon, K_\epsilon]$, with the supremum norm. With respect to this metric, we use the continuous function $h \mapsto \sup_{k_\epsilon \leq \ell \leq K_\epsilon} \inf_{k_\epsilon \leq u \leq K_\epsilon}|h(\ell,u)|$ to obtain
$$
    \sup_{t \in \R} \ap{\P(\sup_{h_\epsilon < \ell \leq H_\epsilon} \inf_{h_\epsilon \leq u \leq H_\epsilon }n^{1/3}(\overline {\tilde Y}_{[x_0 - \ell n^{-1/3}, x_0 + u n^{-1/3}]} - \tilde f_0(x_0)) \leq t \mid  \mathcal{O}) - \P(\sup_{k_\epsilon < \ell \leq K_\epsilon} \inf_{k_\epsilon \leq u \leq K_\epsilon }G_{\ell, u}\leq t)} = o_{\P}(1).
$$
By Lemma \ref{lemma:ul_o1}, we can assume without loss of generality that the max-min formula holds so that
$$
    \sup_{t \in \R} \ap{\pr{n^{1/3}(\tilde f_n(x_0) - \tilde f_0(x_0)) \leq t\cdl \mathcal{O}} - \P(\sup_{k_\epsilon < \ell \leq K_\epsilon} \inf_{k_\epsilon \leq u \leq K_\epsilon} G_{\ell, u} \leq t)} = o_{\P_{\cO}}(1).
$$
Since we are working on the event where $k_\epsilon \leq  U_G^*, L_G^* \leq K_\epsilon$, we have, equivalently, 
$$
    \sup_{t \in \R} \ap{\pr{n^{1/3}(\tilde f_n(x_0) - \tilde f_0(x_0)) \leq t \cdl \mathcal{O}} - \P(\sup_{\ell > 0} \inf_{u \geq 0} G_{\ell, u} \leq t)} = o_{\P_{\cO}}(1).
$$
The theorem is then proven by Lemma \ref{lemma:iso_original_data}, then comparing the limits.
\end{proof}

\subsection{Proof of Theorem \ref{theorem:gan}}
\begin{proof}
First, Assumption \ref{trainedwgan}(c) corresponds to the condition that $\bP_{\tilde Z \mid \cO}$ has nonzero variance almost surely. Second, by definition of $\hat \mG_n^{\rm GAN}$, Assumption \ref{assump_bootstrap}(a) automatically holds. It remains to verify the remaining parts.

~\\
\textbf{Step 1: Show Assumption \ref{assump_bootstrap}(c).}

For a vector $\mx \in \R^p$ and a $1$-Lipschitz function $\alpha : \R \rightarrow \R$ applied componentwise,
\al{
\|\alpha(\mx) - \alpha(\mx)\|_2  &= \sqrt{(|\alpha(x_1) - \alpha(x_1')|)^2 + ... + (|\alpha(x_p) - \alpha(x_p')|)^2} \tag{definition of $\alpha$}\\
&\leq \sqrt{(|x_1 - x_1'|)^2 + ... + (|x_p-x_p'|)^2} \tag{$\alpha$ is 1-Lipschitz by Assumption \ref{trainedwgan} and square root function is increasing}\\
&= \|\mx - \mx'\|_2.
}\
Then, for $\Ab \in \R^{r \times p}, \mb \in \R^r$, we have
$$
\|\alpha(\Ab \mx + \mb) - \alpha(\Ab \mx + \mb)\|_2 \leq \|\Ab(\mx - \mx') \|_2 \leq \|\Ab\|_{\rm op} \cdot \|\mx - \mx'\|_2.
$$
Therefore, for $\mG \in \cF_\alpha(L^{\rm gen}, W_n^{\rm gen}, B^{\rm gen}, p, p)$,
$$
 \|\mG(\mx) - \mG(\mx')\|_2 \leq \|\mx - \mx'\|_2 \cdot \prod_{i=1}^{L^{\rm gen}} \|\Ab^{(i)}\|_{\rm op} \leq (B^{\rm gen})^{L^{\rm gen}} \cdot \|\mx - \mx'\|_2.
$$
That is, the class of functions $\cF_\alpha(L^{\rm gen}, W_n^{\rm gen},  B^{\rm gen}, p, p)$ is uniformly $(B^{\rm gen})^{L^{\rm gen}}$-Lipschitz on $\R^p$. Then, $D \circ \mG$, where $D \in {\rm Lip}(p,1)$, $D(\mathbf{0}) = 0$, and $ \mG \in \cF_\alpha(L^{\rm gen}, W_n^{\rm gen}, B^{\rm gen}, p, p)$ is $(B^{\rm gen})^{L^{\rm gen}}$-Lipschitz on $\R^p$ too. 

Next, for $\bu \in \cU$, with $\alpha$ applied componentwise,
\al{
\|\alpha(\Ab \bu + \mb)\|_2 &= \|\alpha(\Ab \bu + \mb) - \alpha(\mathbf{0})\|_2 \tag{Assumption \ref{trainedwgan}(a)}\\
&\leq \|\Ab \bu + \mb\|_2 \\
&\leq \|\Ab \bu\|_2 + \|\mb\|_2 \tag{triangle inequality}\\
&\leq B^{\rm gen} \cdot \|\bu\|_2 + \|\mb\|_2 \tag{definition of operator norm and $B^{\rm gen}$}\\
&\leq B^{\rm gen} \cdot \|\bu\|_2 + B^{\rm gen}. \tag{$\|\mb\|_2 \leq B^{\rm gen}$}
}
Hence, using the recursive definition of $\mG$, and the fact that $\cU$ contains $\mathbf{0}$, 
\begin{equation}\label{eq:theoremgan:GUcalc}
\|\mG(\bu)\|_2 \leq (B^{\rm gen})^{L^{\rm gen}} \cdot \|\bu\|_2 + \sum_{i=1}^{L^{\rm gen}}  (B^{\rm gen})^i.
\end{equation}
Since $\cU$ is compact by Assumption \ref{assump_noise2}, and $\mG$ is arbitrary, this shows that Assumption \ref{assump_bootstrap}(c) holds, because $\hat \mG_n^{\rm GAN} \in \cF_\alpha(L^{\rm gen}, W_n^{\rm gen},  B^{\rm gen}, p, p).$

~\\
\textbf{Step 2: Show Assumption \ref{assump_bootstrap}(b). }

Next, take $\epsilon > 0$. We aim to show
\al{
\P\pa{\sW_1(\bP_{\tilde Z \mid \cO}, \P_{Z}) \geq \epsilon } &= \P\pa{\sup_{D \in {\textrm{Lip}}(p,1)}\E\Big[ D( \hat \mG_n^{\rm GAN} (\mU)) - D(\mZ) \mid \cO\Big] \geq \epsilon} \tag{Lemma \ref{lemma:kr_dual}}\\
&= \P\pa{\sup_{D \in {\textrm{Lip}}(p,1), D(\mathbf{0}) = 0}\E\Big[ D( \hat \mG_n^{\rm GAN} ( \mU)) - D(\mZ) \mid \cO\Big]
\geq \epsilon}\\
&= \P\pa{\sup_{D \in {\textrm{Lip}}(p,1), D(\mathbf{0}) = 0}\E[\sW(\hat\mG_n^{\rm GAN}, D, \mZ, \mU) \mid \cO]
\geq \epsilon}
}
goes to zero. 

To do so, we first show 
\begin{align}\label{eq:thmgan:emp_proc_term}
 \sup_{D \in {\textrm{Lip}}(p,1), D(\mathbf{0}) = 0} \Big\{\frac{1}{n} \sum_{i=1}^n \sW(\hat\mG_n^{\rm GAN}, D, \mZ_i, \mU_i) - \E\Big[\sW(\hat\mG_n^{\rm GAN}, D, \mZ, \mU) \mid \cO\Big]\Big\} = o_{\P}(1).
\end{align}
Since $D(\mathbf{0}) = 0$, for any $\bu \in \cU$ and $\mz \in \cZ$,
$$
\|D\circ \hat\mG_n^{\rm GAN}(\bu)\|_2 \leq (B^{\rm gen})^{L^{\rm gen}} \cdot \|\bu\|_2 + \sum_{i=1}^{L^{\rm gen}}  (B^{\rm gen})^i
~~~\text{and}~~~
\|D(\mz)\|_2 \leq \|\mz\|_2,
$$
by \eqref{eq:theoremgan:GUcalc}. Summing up our calculations, and using the notation of Lemma \ref{lemma:lip_covering_num}, we obtain
$$
\Big\{D\circ \hat\mG_n^{\rm GAN} : D \in {\rm{Lip}}(p,1), D(\mathbf{0}) = 0 \Big\} \subseteq BL\pa{\cU, (B^{\rm gen})^{L^{\rm gen}} + (B^{\rm gen})^{L^{\rm gen}} \cdot \sup_{\bu \in \cU}\|\bu\|_2 + \sum_{i=1}^{L^{\rm gen}}  (B^{\rm gen})^i } 
$$ 
and
$$
\Big\{D : D \in {\rm{Lip}}(p,1), D(\mathbf{0}) = 0 \Big\} \subseteq BL\pa{\cZ, \sup_{\mz \in \cZ}\|\mz\|_2 + 1}.
$$
Accordingly,  Lemmas \ref{lemma:uniform_gc_class} and \ref{lemma:lip_covering_num} imply \eqref{eq:thmgan:emp_proc_term}. Using the first condition in Assumption \ref{trainedwgan}(b), it then suffices to show
$$
\frac1n \sum_{i=1}^n 
\sW(\hat\mG_n^{\rm GAN}, \hat D_n^{\rm GAN}, \mZ_i, \mU_i) = o_{\P}(1),
$$
which is exactly the second condition of Assumption \ref{trainedwgan}(b). We thus conclude the proof. 
\end{proof}

\subsection{Proof of Theorem \ref{flow_wcvg}}
\begin{proof}
\textbf{Step 1: Proof of Theorem \ref{flow_wcvg}(a).}

Assumption \ref{assump_bootstrap}(a) holds automatically by the definition of $\hat \mS_n^{\rm rflow}$. We then verify the remaining parts.

\noindent\textbf{Step 1a: Show that $\tilde \cZ_n \subseteq \tilde \cZ$ for some nonrandom compact set $\tilde \cZ \subseteq \R^p$.}

Take any $\mS = \mF^{\nu} \circ \boldsymbol{\Sigma}_\nu \circ ... \circ \mF^1 \circ \boldsymbol{\Sigma}_1 \in \cF_{\nu, K, M}$. Since $K >0$, part (a) of the definition of $\cT_\uparrow (p)$ and part (a) of the definition of $\cF_{\nu, K,M}$ indicate that $\mS$ is a continuously differentiable and bijective map from $\R^p$ to $\R^p$. Hence, $\mS^{-1}: \R^p \rightarrow \R^p$ exists. Since
\al{
\sD \mS(\mz) &= \sD(\mF^{\nu} \circ \boldsymbol{\Sigma}_\nu \circ ... \circ \mF^1 \circ \boldsymbol{\Sigma}_1)(\mz) \tag{definition of $S$}\\
&= \sD\mF^{\nu}  (\boldsymbol{\Sigma}_\nu \circ ... \circ \mF^2 \circ \mF^1 \circ \boldsymbol{\Sigma}_1 (\mz)) \sD(\boldsymbol{\Sigma}_\nu \circ ... \circ  \mF^1 \circ \boldsymbol{\Sigma}_1)(\mz) \tag{chain rule}\\
&= \sD\mF^{\nu}  (\boldsymbol{\Sigma}_\nu \circ ... \circ \mF^2 \circ \mF^1 \circ \boldsymbol{\Sigma}_1 (\mz)) \boldsymbol{\Sigma}_\nu  \sD(\mF^{\nu-1} \circ ... \circ  \mF^1 \circ \boldsymbol{\Sigma}_1)(\mz) \tag{pull out constant matrix $\boldsymbol{\Sigma}_\nu$}
}
for any $\mz \in \R^p$, we apply the chain rule $\nu-1$ more times in the same manner, to obtain
$$
\sD \mS(\mz) = \sD\mF^{\nu}(\mx^{\nu}) \boldsymbol{\Sigma}_\nu \, ... \,  \sD \mF^1(\mx^1) \boldsymbol{\Sigma}_1 \qquad \text{with} \qquad \mx^{i} = (\boldsymbol{\Sigma}_i \circ  ... \circ \mF^1 \circ \boldsymbol{\Sigma}_1)(\mz),
$$
which yields
$$
\det(\sD \mS(\mz)) = \prod_{i=1}^{\nu} \det(\sD\mF^{i}(\mx^{i})) \cdot \det(\boldsymbol{\Sigma}_i).
$$
Since $\mF^1, ..., \mF^\nu \in \cT_\uparrow (p),$ and the determinant of an upper triangular matrix is the product of its diagonal terms, we obtain
$$
\det(\sD \mF^i(\mz)) = \prod_{j=1}^p \sD_j F^i_j(\mz) 
$$
so that 
$$
\det(\sD \mS(\mz)) = \prod_{i=1}^{\nu} \prod_{j=1}^p \sD_j F^i_j(\mz) \cdot \det(\boldsymbol{\Sigma}_i).
$$
The definition of the class $\cF_{\nu, K, M}$ indicates that 
\begin{equation}
(K^pM)^{-\nu} \leq \det(\sD \mS(\mz)) \leq (K^pM)^{\nu}. \label{eq:flow_cvg:density_bd_ul}
\end{equation}
Since $\mz$ is arbitrary and applying Lemma \ref{lemma:c1inverse}, $\mS^{-1}$ is also continuously differentiable. Since $\hat \mS_n^{\rm rflow} \in \cF_{\nu, K, M}$ , this argument shows $\hat \mG_n^{\rm rflow} = (\hat \mS_n^{\rm rflow})^{-1}$ is continuously differentiable on $\R^p$. 

Moreover, by Lemma \ref{lemma:c1inverse}, we have
$$
\sD(\mS^{-1})(\mz) =  \sD\mS(\mS^{-1}(\mz))^{-1}.
$$
Accordingly, letting $\my = \mS^{-1}(\mz)$ and $\my^i = (\boldsymbol{\Sigma}_i \circ  ... \circ \mF^1 \circ \boldsymbol{\Sigma}_1)(\my)$ for each $i \in [\nu]$, we obtain
\al{
\sD\mS(\my)^{-1} &= 
(\sD\mF^{\nu}(\my^{\nu}) \boldsymbol{\Sigma}_\nu \, ... \,  \sD \mF^1(\my^1) \boldsymbol{\Sigma}_1)^{-1} \tag{formula for $\sD \mS$}\\
&=\boldsymbol{\Sigma}_1^{-1} (\sD \mF^1(\my^1))^{-1} ... \boldsymbol{\Sigma}_\nu^{-1} (\sD\mF^{\nu}(\my^{\nu}))^{-1}. \tag{formula for inverse of product of invertible matrices}
}
We thus reach that
\begin{align}
\|\sD\mS(\my)^{-1}\|_{\rm op} &\leq \|\boldsymbol{\Sigma}_1^{-1} (\sD \mF^1(\my^1))^{-1} ... \boldsymbol{\Sigma}_\nu^{-1} (\sD\mF^{\nu}(\my^{\nu}))^{-1}\|_{\rm op}\nonumber\\
&\leq \prod_{i=1}^\nu \|\boldsymbol{\Sigma}_i^{-1}\|_{\rm op} \cdot \|(\sD\mF^{i}(\my^{i}))^{-1}\|_{\rm op} \nonumber\tag{operator norm of product bound}\\
&= \prod_{i=1}^\nu \lambda_{\rm min}(\boldsymbol{\Sigma}_i)^{-1} \cdot \|(\sD\mF^{i}(\my^{i}))^{-1}\|_{\rm op} \nonumber\tag{$\boldsymbol{\Sigma}_i$ is symmetric and positive definite}\\
&\leq K^\nu \prod_{i=1}^\nu  \|(\sD\mF^{i}(\my^{i}))^{-1}\|_{\rm op} \nonumber\tag{definition of $\cF_{\nu, K, M}$}\\
&\leq K^\nu \cdot \prod_{i=1}^\nu M\sqrt{p}\pa{1+ M\|\sD\mF^{i}(\my^{i})\|_{\rm op}}^{p-1} \nonumber\tag{Lemma \ref{lemma:inverse_op_bound}}\\
&\leq K^\nu \cdot \prod_{i=1}^\nu M\sqrt{p}\pa{1+ M\sqrt{(pM)^2}}^{p-1}\nonumber\\
&= K^\nu \cdot \prod_{i=1}^\nu M\sqrt{p}\pa{1+ M^2p}^{p-1},\nonumber
\end{align}
where the final inequality comes from the fact that
$$
\|\Ab\|_{\rm op} \leq \sqrt{\max_{j} \sum_{i\leq j}  |A_{ij}| \cdot \max_{i'}\sum_{i' \geq j'} |A_{i'j'}|}
$$
and the fact that the magnitudes of all entries in $\sD\mF^{i}(\my^{i})$ are bounded by $M$ by definition of $\cF_{\nu, K, M}.$ Since $\mz$, and hence, $\my$ is arbitrary, we have
\begin{equation}
\sup_{\mz \in \R^p }\|\sD(\mS^{-1})(\mz)\|_{\rm op} \leq  K^\nu \cdot \prod_{i=1}^\nu M\sqrt{p}\pa{1+ M^2p}^{p-1}= K^\nu M^\nu p^{\nu/2}\pa{1+ M^2p}^{\nu(p-1)}\label{eq:flowcvg:ghat_lipschitz}
\end{equation}
for any $\mS \in \cF_{\nu, K, M}$.
Since $\hat \mG_n^{\rm rflow}$ is continuously differentiable on $\R^p$, the mean value theorem implies
\begin{align}
\|\hat\mG_n^{\rm rflow}(\mz) - \hat\mG_n^{\rm rflow}(\mz')\|_2 &\leq \sup_{\mz''\in\R^p}\|\sD\hat\mG_n^{\rm rflow}(\mz'')\|_{\rm op}\cdot\| \mz - \mz'\|_2 \nonumber\\
&\leq K^\nu M^\nu p^{\nu/2}\pa{1+ M^2p}^{\nu(p-1)}\cdot \|\mz - \mz'\|_2\nonumber\tag{$\hat\mS_{n}^{\rm rflow} \in \cF_{\nu, K,M}$ and \eqref{eq:flowcvg:ghat_lipschitz}}.
\end{align}

By setting $\mz' =\mathbf{0}$ and using the fact that each element of $\cF_{\nu, K,M}$ takes $\mathbf{0}$ to $\mathbf{0}$ by definition, the above inequality yields that
$$
\|\hat \mG_n^{\rm rflow}(\mz) - \hat \mG_n^{\rm rflow}(\mathbf{0})\|_2= \|\hat \mG_n^{\rm rflow}(\mz)\|_2 \leq K^\nu M^\nu p^{\nu/2}\pa{1+ M^2p}^{\nu(p-1)} \cdot \|\mz\|_2.
$$
Thus, for any $\textit{\textbf u} \in \cU$, we obtain
$$
\|\hat \mG_n^{\rm rflow}(\bu)\|_2 \leq K^\nu M^\nu p^{\nu/2}\pa{1+ M^2p}^{\nu(p-1)} \cdot \|\bu\|_2 \leq K^\nu M^\nu p^{\nu/2}\pa{1+ M^2p}^{\nu(p-1)} \cdot \sup_{\bu \in \cU}\|\bu\|_2,
$$
where Assumption \ref{assump_noise2} implies the boundedness of $\cU$ so that the right-hand side is finite. Since the right-hand side is nonrandom and does not depend on $n$, we reach the conclusion of {\bf Step 1a}.

~\\
\textbf{Step 1b: Show that $\tilde \cZ_n \supseteq \cZ$ almost surely.}

We have shown that each $\mS \in \cF_{\nu, K, M}$ is continuously differentiable on $\R^p$ and has an inverse that is also continuously differentiable on $\R^p$. Assumption \ref{assump_noise2} guarantees that $p_U$ exists, so that applying Lemma \ref{lemma:flowcvg:cov}, the density of $\mS^{-1}(\mU)$ is
\begin{equation}
p_{S^{-1}(U)}(\mz) = p_U(\mS(\mz)) \cdot |\det(\sD\mS(\mz))|. \label{eq:flowcvg:hatgn_density}
\end{equation}
\eqref{eq:flow_cvg:density_bd_ul} shows that the determinant is always positive, so that the above expression is $0$ if and only if $p_U(\mS(\mz)) = 0$, which happens if and only if $\mS(\mz) \not\in \cU$. 
Since $\mS$ is bijective, this shows that the support of $\mS^{-1}(\mU)$ is $\mS^{-1}(\cU)$, which denotes the image of the set $\cU$ under the mapping $\mS^{-1}$.

Next, by mean value theorem,
\begin{align}
\|\mS(\mz) - \mS(\mz')\|_2 
&\leq \sup_{\mz'' \in \R^p}\|\sD \mS(\mz'')\|_{\rm op} \cdot \|\mz - \mz'\|_2 \nonumber\\
&\leq \prod_{i=1}^\nu \|\boldsymbol{\Sigma}_i\|_{\rm op} \cdot \sup_{\mz'' \in \R^p}\|\sD\mF^{i}(\mz'')\|_{\rm op}\cdot \|\mz - \mz'\|_2  \nonumber\tag{upper bound by product of operator norms and take supremum}\\
&= \prod_{i=1}^\nu \lambda_{\rm max}(\boldsymbol{\Sigma}_i) \cdot  \sup_{\mz'' \in \R^p}\|\sD\mF^{i}(\mz'')\|_{\rm op} \cdot \|\mz - \mz'\|_2 \nonumber\tag{$\boldsymbol{\Sigma}_i$ is symmetric and positive definite}\\
&\leq K^\nu \prod_{i=1}^\nu \sup_{\mz'' \in \R^p}\|\sD\mF^{i}(\mz'')\|_{\rm op} \cdot \|\mz - \mz'\|_2 \nonumber \tag{definition of $\cF_{\nu,K,M}$}\\
&\leq (KpM)^{\nu} \cdot \|\mz - \mz'\|_2, \label{eq:flowcvg:slip}
\end{align}
where 
$$
 \sup_{\mz'' \in \R^p}\|\sD\mF^{i}(\mz'')\|_{\rm op} \leq pM
$$
from the inequality
$$
\|\Ab\|_{\rm op} \leq \sqrt{\max_{j} \sum_{i\leq j}  |A_{ij}| \cdot \max_{i'}\sum_{i' \geq j'} |A_{i'j'}|}
$$
and the fact that the absolute value all entries of $\sD\mF^{i}(\mz'')$ are upper bounded by $M$. Let $\mz' = \mathbf{0}$ and recall $\mS(\mathbf{0})=\mathbf{0}$. If $\mz \in \cZ$, then $\|\mz\|_2 \leq (KpM)^{-\nu}r_0$ for $r_0$ in Assumption \ref{assump_flows}(b). By \eqref{eq:flowcvg:slip}, $\|\mS(\mz)\|_2 \leq r_0$ implies
\begin{equation}
\mS(\cZ) \subseteq \Big\{\mx \in \R^p : \|\mx \|_2 \leq r_0\Big\} \label{eq:flowcvg:sz_subset_r0ball}
\end{equation}
and, in particular, $\mS(\mz) \in \cU$ by Assumption \ref{assump_flows}(b). Since $\mS$ is invertible, $\mz \in \mS^{-1}(\cU)$. Since $\mS$ is arbitrary, we set $\mS = \hat \mS_{n}^{\rm rflow}$, so that $\mz \in \hat\mG_n^{\rm rflow}(\cU) = \tilde\cZ_n$. Finally, as $\mz$ was arbitrary, we have $\cZ \subseteq \tilde \cZ_n$ almost surely, for each $n$, completing the proof of  {\bf Step 1b}.

~\\
\textbf{Step 1c: Show that $\bP_{\tilde Z \mid \cO}$ has nonzero variance, almost surely, and the Wasserstein distance between $\bP_{\tilde Z\mid \cO}$ and $\P_Z$ goes to $0$ in probability.}

Since $\hat \mG_n^{\rm rflow}$ is an invertible map and $\cU$ has nonzero variance, by Assumption \ref{assump_noise1}, $\hat\mG_n^{\rm rflow}(\mU)$ cannot be constant almost surely, and, thus, must have nonzero variance. This confirms $\bP_{\tilde Z\mid \cO}$ has nonzero variance, almost surely.

By \textbf{Step 1a} and \textbf{Step 1b}, we have shown $\cZ \subseteq \tilde \cZ_n \subseteq \tilde \cZ$ almost surely, for all $n=1,2, ...$. Now applying Lemma \ref{lemma:w1_kl_bddsupport}, we obtain
$$
    \sW_1(\bP_{\tilde Z \mid \cO}, \P_Z) \leq 2\cdot  \sup_{\mz \in \tilde\cZ } \|\mz\|_2 \cdot \sqrt{\frac{1}{2} \cdot \KL(\P_Z||\bP_{\tilde Z \mid \cO} )}
$$
almost surely. Due to the fact that $\tilde \cZ $ is bounded, it suffices to show
$$
\KL(\P_Z||\bP_{\tilde Z\mid \cO} ) = o_{\P}(1).
$$
We rewrite the KL-divergence as
$$
\KL(\P_Z || \bP_{\tilde Z\mid \cO}) = \E\Big[\log p_{Z}(\mZ)\Big] - \E\Big[\log \tilde p_n(\mZ) \mid \cO\Big].
$$
 By assumption, $\E[|\log p_Z(\mZ)|] < \infty$.
Taking any $\mS \in \cF_{\nu, K,M}$, \eqref{eq:flow_cvg:density_bd_ul} and \eqref{eq:flowcvg:hatgn_density} combined give,
$$
\inf_{\mz \in \cZ}p_U(\mS(\mz)) \cdot (K^pM)^{-\nu} \leq p_{S^{-1}(U)}(\mz) \leq \sup_{\mz \in \cZ} p_U(\mS(\mz)) \cdot (K^pM)^{\nu}.
$$
We have already shown $\mS(\cZ) \subseteq \{\mz \in \R^p : \|\mz\|_2 \leq r_0\}$. By Assumption \ref{assump_noise2}, $\inf_{\mz \in \cZ}p_U(\mS(\mz)) \geq \inf_{\mz \in \R^p : \|
\mz \| \leq r_0}p_U(\mz) > 0$. Furthermore, Assumption \ref{assump_noise2} implies continuity of $p_U$ on the compact set $\cU$. Assumption \ref{assump_flows}(b) guarantees $\sup_{\mz \in \cZ} p_U(\mS(\mz)) \leq \sup_{\bu \in \cU} p_U(\bu) < \infty$.  Put together, we have $\E[|\log p_{S^{-1}(U)}(\mZ)|] < \infty$ and 
\begin{equation}
\sup_{\mS \in \cF_{\nu, K,M}} \sup_{\mz \in \cZ}|\log p_{S^{-1}(U)}(\mz)| < \infty. \label{eq:flowcvg:gc_ub}
\end{equation}
Then, for any $\mz, \mz' \in \cZ$, 
$$
\Big|\log p_{S^{-1}(U)}(\mz) - \log p_{S^{-1}(U)}(\mz')\Big| \leq \frac{ (K^pM)^{\nu}}{\inf_{\bu : \|
\bu \| \leq r_0}p_U(\bu)} \cdot |p_{S^{-1}(U)}(\mz) - p_{S^{-1}(U)}(\mz')|
$$
since $p_{S^{-1}(U)}$ is bounded below by $\inf_{\bu : \|
\bu \| \leq r_0}p_U(\bu) \cdot (K^pM)^{-\nu}$.  Accordingly, by the definition of $p_{S^{-1}(U)}$ and \eqref{eq:flow_cvg:density_bd_ul},
\al{
 \Big|p_{S^{-1}(U)}(\mz) - p_{S^{-1}(U)}(\mz')\Big| &\leq (K^pM)^{\nu} \cdot \Big|p_U(\mS(\mz)) - p_U(\mS(\mz'))\Big|\\ &+ \sup_{\bu \in \cU} |p_u(\bu)| \cdot \Big|\det(\sD\mS(\mz)) - \det(\sD\mS(\mz')) \Big|.
}
The mean value theorem and convexity of $\cU_1 \times \R$ combined then imply
$$
  |p_U(\mS(\mz)) - p_U(\mS(\mz'))| \leq \sup_{\bu\in \cU} \|\sD p_U(\bu)\|_{\rm op} \cdot \sup_{\mS \in \cF_{\nu, K, M}} \sup_{\mz'' \in \R^p} \|\sD \mS(\mz'')\|_{\rm op} \cdot  \|\mz - \mz'\|_2.
$$
by the containment $\mS(\cZ) \subseteq \cU$ through \eqref{eq:flowcvg:slip} and Assumption \ref{assump_flows}(b).
The continuity of $\sD p_U$ on $\cU$ from Assumption \ref{assump_noise2} and the definition of $\cF_{\nu, K,M}$ means the right-hand side is finite. Furthermore, for some constant $L > 0$, depending on the dimension $p$, and parameters $\nu, M, K$, the Lipschitz condition
$$
\Big|\det(\sD\mS(\mz)) - \det(\sD\mS(\mz')) \Big| \leq L \cdot \|\mz - \mz'\|_2
$$
holds because the entries of $\sD\mS(\mz)$ are uniformly bounded in $\mz$ and $\mS$, and the determinant is a polynomial in the entries of its argument.
We have thus proven that $\{\log p_{S^{-1}(U)} : \mS \in \cF_{\nu, K,M}\}$ is uniformly bounded, and uniformly Lipschitz on $\cZ$. Hence, applying Lemma \ref{lemma:lip_covering_num} and then Lemma \ref{lemma:uniform_gc_class}, we obtain
$$
\sup_{\mS \in \cF_{\nu, K,M}} \ap{\frac{1}{n} \sum_{i=1}^n \log p_{S^{-1}(U)}(\mZ_i) - \E[\log p_{S^{-1}(U)}(\mZ)]}  = o_{\P}(1).
$$
Furthermore, we have assumed $\E[|\log p_Z(\mZ)|]< \infty $, so that the law of large numbers yields
$$
\frac{1}{n} \sum_{i=1}^n \log p_{Z}(\mZ_i) - \E[\log p_{Z}(\mZ)]  = o_{\P}(1).
$$
Combining the above two displays, we obtain
$$
\sup_{\mS \in \cF_{\nu, K,M}} \ap{\frac{1}{n} \sum_{i=1}^n \log p_Z(\mZ_i) - \log p_{S^{-1}(U)}(\mZ_i) - \E\Big[\log p_Z(\mZ) - \log p_{S^{-1}(U)}(\mZ)\Big]} = o_{\P}(1).
$$
Since $\hat\mS_n^{\rm rflow} \in \cF_{\nu, K,M}$, we also deduce
$$
\ap{\frac{1}{n} \sum_{i=1}^n \log p_Z(\mZ_i) - \log \tilde p_{n}(\mZ_i) - \E\Big[\log p_Z(\mZ) - \log \tilde p_{n}(\mZ) \mid \cO\Big]} = o_{\P}(1).
$$
Combined with Assumption \ref{assump_flows}(a), we finish the proof of {\bf Step 1c}.

~\\
\textbf{Step 2: Proof of Theorem \ref{flow_wcvg}(b).}

Assumption \ref{assump_iso} implies Assumption \ref{assump_ndata}, and thus  Assumption \ref{assump_bootstrap} holds under the conditions of Theorem \ref{flow_wcvg}(b). It remains to verify the remaining conditions.

\noindent\textbf{Step 2a: Show $\tilde \cX_n \supseteq \cX$ for all $n$, almost surely, and show $\tilde \cX_n$ is a closed interval.}

{\bf Step 1b} has shown that $\tilde \cZ_n \supseteq \cZ$ for all $n$, almost surely. To verify that $\tilde \cX_n$ is a closed interval, observe that $\hat \mG_n^{\rm rflow}$ is a continuous function and $\cU$ is a compact and connected set. Therefore, the image set, $\hat\mG_n^{\rm rflow}(\cU) = \tilde \cZ_n$, must be connected and compact. The map that projects onto the first component, $(x,y) \mapsto x$, is continuous. Therefore, $\tilde \cX_n$ must also be connected and compact. Since $\tilde \cX_n \subseteq \R$, it must be that $\tilde \cX_n$ is a closed interval.

~\\
\textbf{Step 2b: Show that $\tilde p_n$ is always twice-continuously differentiable on $\tilde \cX_n \times \R$.}

It suffices to show this stronger statement that implies Assumption \ref{ass:generator}(a). By Lemma \ref{lemma:flowcvg:cov},
$$
\tilde p_n(\mz) = p_U(\hat\mS_n^{\rm rflow}(\mz)) \cdot \det(\sD\hat\mS_n^{\rm rflow}(\mz))
$$
for $\mz \in \R^2$; note that due to \eqref{eq:flow_cvg:density_bd_ul} there is no absolute value. Observe that
$$
\hat\mS_n^{\rm rflow}(\tilde\cX_n \times \R) \subseteq \cU_1 \times \R.
$$
Note first that Assumption \ref{assump_noise3}(b) ensures that the function $p_U$ is twice continuously differentiable on $\hat\mS_n^{\rm rflow}(\tilde\cX_n \times \R)$. Second, the function $\hat\mS_n^{\rm rflow} \in \cF_{\nu, K,M}$ is twice continuously differentiable on $\R^2$. Then, using the chain rule, we obtain that $p_U \circ \hat\mS_n^{\rm rflow}$ is continuously differentiable on $\tilde\cX_n \times \R$. Using the product and chain rules on the first derivative confirms that $p_U \circ \hat\mS_n^{\rm rflow} $ is twice continuously differentiable on $\tilde\cX_n \times \R$.

The determinant of a matrix is a polynomial in the entries of the matrix. As a result, $\det(\sD\hat\mS_n^{\rm rflow}(\mz))$, as a function of $\mz$, is twice-continuously differentiable on $\tilde \cX_n \times \R$ if the map $\sD\hat\mS_n^{\rm rflow}(\mz)$ is twice-continuously differentiable on $\tilde \cX_n \times \R$. Equivalently, this happens if $\hat\mS_n^{\rm rflow}(\mz)$ is three-times continuously differentiable on $\tilde \cX_n \times \R$. By the product rule, this yields the conclusion of {\bf Step 2b}.

~\\
\noindent \textbf{Step 2c: Show that $\tilde p_X$ is lower bounded uniformly in $\cX$, and $\tilde p_n$, $\|\sD \tilde p_n\|_2$, and $\|\sD^2\tilde p_n\|_{\rm op}$ are upper bounded by a universal constant in $\tilde \cZ_n$.}

To show the lower bound, for each $x \in \cX$
\al{
\tilde p_{X}(x) &= \int_\R \tilde p_n(x,y) \, \d y \tag{marginal probability}\\ 
 &=   \int_\R p_U(\hat\mS_n^{\rm rflow}(x,y)) \cdot \det(\sD\hat\mS_n^{\rm rflow}(x,y)) \, \d y \tag{Lemma \ref{lemma:flowcvg:cov}}\\
 &\geq  (K^p M)^{-\nu} \cdot  \int_{\R} p_U(\hat\mS_n^{\rm rflow}(x,y)) \, \d y \tag{Equation \eqref{eq:flow_cvg:density_bd_ul}}\\
 &\geq (K^p M)^{-\nu} \cdot  \int_{\{y \in \cY: \|\hat\mS_n^{\rm rflow}(x,y)\|_2 \leq r_0\}} p_U(\hat\mS_n^{\rm rflow}(x,y)) \, \d y \tag{restrict the integration set}\\
 &\geq (K^p M)^{-\nu} \cdot \inf_{\bu \in \cU} p_U(\bu) \cdot  \int_{\{y \in \cY: \|\hat\mS_n^{\rm rflow}(x,y)\|_2 \leq r_0\}}  \, \d y \tag{ $p_U$ is uniformly lower bounded on $\{\mz \in \R^2 : \|\mz\|_2 \leq r_0\}$ by Assumption \ref{assump_noise2}}\\
 &= (K^p M)^{-\nu} \cdot \inf_{\bu \in \cU} p_U(\bu) \cdot  \int_{\cY}  \, \d y \tag{$(x,y) \in \cX \times \cY = \cZ$ from 
Assumption \ref{assump_iso}(e), then applying Equation \eqref{eq:flowcvg:sz_subset_r0ball}}\\
&> 0. \tag{ $\cY$ has positive measure from Assumption \ref{assump_iso}(e)}
}
Taking the infimum with respect to $x\in \cX$, we have shown $\tilde p_X$ is uniformly lower bounded in $\cX$.

Next, we aim to show that $\tilde p_n$ is upper bounded in $\tilde \cZ_n$. By Lemma \ref{lemma:flowcvg:cov}, 
\al{
\sup_{\mz \in \tilde \cZ_n} \tilde p_n(\mz) &= \sup_{\mz \in \tilde \cZ_n} p_U(\hat\mS_n^{\rm rflow}(\mz)) \cdot \det(\sD\hat\mS_n^{\rm rflow}(\mz)) \\
&= \sup_{\mz \in \tilde \cZ_n} p_U(\hat\mS_n^{\rm rflow}(\mz)) \cdot (K^pM)^{\nu} \tag{Equation \eqref{eq:flow_cvg:density_bd_ul}}\\&\leq \sup_{\bu \in \cU} p_U(\bu) \cdot (K^pM)^{\nu} \tag{for $\mz \in \tilde \cZ_n$, $\hat\mS_n^{\rm rflow}(\mz) \in \cU$}\\ 
&< \infty. \tag{Assumption \ref{assump_noise2} ensures $p_U$ is continuous on the compact set $\cU$}
}
For $\sD \tilde p_n$, 
\al{
\|\sD \tilde p_n(\mz)\|_2 &= \|\sD p_U(\hat\mS_n^{\rm rflow}(\mz)) \cdot \det(\sD\hat\mS_n^{\rm rflow}(\mz))\|_2 + \|p_U(\hat\mS_n^{\rm rflow}(\mz)) \cdot \sD \det(\sD\hat\mS_n^{\rm rflow}(\mz))\|_2 \tag{product rule and triangle inequality}\\
&\leq \|\sD p_U(\hat\mS_n^{\rm rflow}(\mz))\|_2 \cdot (K^pM)^\nu + p_U(\hat\mS_n^{\rm flow}(\mz)) \cdot \|\sD \det(\sD\hat\mS_n^{\rm rflow}(\mz)) \|_2\tag{Equation \eqref{eq:flow_cvg:density_bd_ul}}\\
&\leq \sup_{\bu \in \cU}\|\sD p_U(\bu)\|_2 \cdot (K^pM)^\nu + \sup_{\bu \in \cU}p_U(\bu) \cdot \|\sD \det(\sD\hat\mS_n^{\rm rflow}(\mz))\|_2 \tag{for $\mz \in \tilde \cZ_n$, $\hat\mS_n^{\rm rflow}(\mz) \in \cU$}.
}
We have $\sup_{\bu \in \cU}\|\sD p_U(\bu)\|_2 < \infty$ and $\sup_{\bu \in \cU}p_U(\bu) < \infty$ because Assumption \ref{assump_noise2} ensures $p_U$ and $\sD p_U$ are continuous on the compact set $\cU$. The function $\det(\sD\hat\mS_n^{\rm rflow}(\mz))$ is a polynomial in the first partial derivatives of $\hat \mS_{n}^{\rm rflow}$. By the chain rule, each component of $\|\sD \det(\sD\hat\mS_n^{\rm rflow}(\mz))\|_2$ is a polynomial in the first and second-order partial derivatives of $\hat \mS_n^{\rm rflow}$. By definition, $\hat \mS_n^{\rm rflow} \in \cF_{\nu, K,M}$, and has first and second order partial derivatives bounded by $M$ on $\R^p$. A uniform bound on $\|\sD \det(\sD\hat\mS_n^{\rm rflow}(\mz))\|_2$, in terms of $M$, holds as a result. The exact same reasoning leads to a uniform bound of $\|\sD^2 \det(\sD\hat\mS_n^{\rm rflow}(\mz))\|_{\rm op}$, since the definition of $\cF_{\nu, K,M}$ involves bounds on third partial derivatives over $\R^p$.

Applying the product rule again and using the twice-continuous differentiability of $p_U$ from Assumption \ref{assump_noise3} to show $\sup_{\bu \in \cU}\|\sD^2p_U(\bu)\|_{\rm op} < \infty$, we obtain that $\|\sD^2 \tilde p_n\|_{\rm op}$ is upper bounded by a universal constant, and thus complete the proof.
\end{proof}

\section{Supporting lemmas}\label{sec:proofsupp}
In the proofs of lemmas in Sections \ref{sec:proofsupp} and \ref{sec:otherresults}, we use  $C, C'$ to represent some generic constants whose values may change from statement to statement.

\subsection{Supporting lemmas for Theorem \ref{thm_reg}}
We demonstrate uniform convergence of expectations and variances of certain classes of functions.

\begin{lemma}\label{lemma:w2cvg}
    Assume Assumptions \ref{assump_ndata}, \ref{assump_bootstrap}, and \ref{assump_l}. Then, for any function $\mg: \cK \times (\cZ \cup \tilde\cZ) \rightarrow \R^r$ continuous in both arguments, we have
    \begin{align}\label{eq:han1}
        \sup_{\boldsymbol{\eta} \in \cK}\Big\|{\rm{Var}}(\mg(\boldsymbol{\eta}, \mZ_1)) - {\rm{Var}}(\mg(\boldsymbol{\eta}, \tilde \mZ_1) \mid  \cO)\Big\|_{\rm{max}} = o_{\P_{\cO}}(1)
    \end{align}
    and 
    \begin{align}\label{eq:han2}
    \sup_{\boldsymbol{\eta} \in \cK}\Big\|\E[\mg(\boldsymbol{\eta},\mZ_1)] - \E[\mg(\boldsymbol{\eta},\tilde \mZ_1) \mid  \cO]\Big\|_{\infty} = o_{\P_{\cO}}(1).
    \end{align}
\end{lemma}
\begin{proof}

We first prove \eqref{eq:han1}. Note that $\mg(\boldsymbol{\eta}, \mZ_1)$ and $\mg(\boldsymbol{\eta}, \tilde \mZ_1)$ are almost surely bounded so that their conditional covariance matrices and expectations exist almost surely for every $\boldsymbol{\eta} \in \cK$.
    
    Decomposing the covariance matrix and using the triangle inequality, we obtain
    \al{
     \Big\|{\rm{Var}}(\mg(\boldsymbol{\eta},\mZ_1)) -& {\rm{Var}}(\mg(\boldsymbol{\eta},\tilde \mZ_1) \mid  \cO)\Big\|_{\rm{max}} \leq \Big\|\E[\mg(\boldsymbol{\eta},\mZ_1)\mg(\boldsymbol{\eta},\mZ_1)^\top] - \E[\mg(\boldsymbol{\eta},\tilde \mZ_1)\mg(\boldsymbol{\eta}, \tilde \mZ_1)^\top \mid  \cO]\Big\|_{\rm{max}}\\
     &+ \Big\|\E[\mg(\boldsymbol{\eta},\mZ_1)]\E[\mg(\boldsymbol{\eta},\mZ_1)]^\top - \E[\mg(\boldsymbol{\eta},\tilde \mZ_1) \mid \cO] \E[\mg(\boldsymbol{\eta},\tilde \mZ_1) \mid  \cO]^\top\Big\|_{\rm{max}}
    }
and taking a supremum over $\boldsymbol{\eta} \in \cK$ on both sides, it suffices to show that 
$$
\sup_{\boldsymbol{\eta} \in \cK} \Big\|\E[\mg(\boldsymbol{\eta},\mZ_1)]\E[\mg(\boldsymbol{\eta},\mZ_1)]^\top - \E[\mg(\boldsymbol{\eta},\tilde \mZ_1) \mid \cO] \E[\mg(\boldsymbol{\eta},\tilde \mZ_1) \mid  \cO]^\top\Big\|_{\rm{max}} = o_{\P_{\cO}}(1)
$$
and 
$$
\sup_{\boldsymbol{\eta} \in \cK} \Big\|\E[\mg(\boldsymbol{\eta},\mZ_1)\mg(\boldsymbol{\eta},\mZ_1)^\top] - \E[\mg(\boldsymbol{\eta},\tilde \mZ_1)\mg(\boldsymbol{\eta}, \tilde \mZ_1)^\top \mid  \cO]\Big\|_{\rm{max}} = o_{\P_{\cO}}(1).
$$
As the proofs of these conclusions are alike, we only show the first equality.

We then have
\begin{align*}
&\sup_{\boldsymbol{\eta} \in \cK} \Big\|\E[\mg(\boldsymbol{\eta},\mZ_1)]\E[\mg(\boldsymbol{\eta},\mZ_1)]^\top - \E[\mg(\boldsymbol{\eta},\tilde \mZ_1) \mid \cO] \E[\mg(\boldsymbol{\eta},\tilde \mZ_1) \mid  \cO]^\top\Big\|_{\rm{max}}\\
\leq & \sup_{\boldsymbol{\eta} \in \cK}\Big\|\E[\mg(\boldsymbol{\eta}, \mZ_1)]\E[\mg(\boldsymbol{\eta}, \mZ_1)]^\top - \E[\mg(\boldsymbol{\eta}, \tilde \mZ_1) \mid \cO] \E[\mg(\boldsymbol{\eta},\tilde \mZ_1) \mid  \cO]^\top\Big\|_{\rm op}\\
\leq& \sup_{\boldsymbol{\eta} \in \cK}\|\big(\E[\mg(\boldsymbol{\eta}, \mZ_1)] - \E[\mg(\boldsymbol{\eta}, \tilde \mZ_1) \mid \cO] \big) \E[\mg(\boldsymbol{\eta}, \mZ_1)]^\top + \E[\mg(\boldsymbol{\eta}, \tilde\mZ_1) \mid \cO ] \big(\E[\mg(\boldsymbol{\eta}, \mZ_1) ] -  \E[\mg(\boldsymbol{\eta},\tilde\mZ_1) \mid  \cO]\big)^\top\|_{\rm op}\\
\leq & \sup_{\boldsymbol{\eta} \in \cK}\|\big(\E[\mg(\boldsymbol{\eta}, \mZ_1)] - \E[\mg(\boldsymbol{\eta}, \tilde \mZ_1) \mid \cO] \big) \E[\mg(\boldsymbol{\eta}, \mZ_1)]^\top\|_{\rm op}\\
&+ \sup_{\boldsymbol{\eta}' \in \cK}\|\E[\mg(\boldsymbol{\eta}', \tilde\mZ_1) \mid \cO ] \big(\E[\mg(\boldsymbol{\eta}', \mZ_1) ] -  \E[\mg(\boldsymbol{\eta}',\tilde\mZ_1) \mid  \cO]\big)^\top\|_{\rm op}\\
\leq &\sup_{\boldsymbol{\eta} \in \cK}\|\E[\mg(\boldsymbol{\eta},\mZ_1)] - \E[\mg(\boldsymbol{\eta},\tilde \mZ_1) \mid \cO] \|_2 \cdot \pa{\sup_{\boldsymbol{\eta}' \in \cK}\|\E[\mg(\boldsymbol{\eta}',\mZ_1)]\|_2 + \sup_{\boldsymbol{\eta}' \in \cK} \|\E[\mg(\boldsymbol{\eta}',\tilde \mZ_1) \mid \cO] \|_2}. \yestag \label{eq:han3}
\end{align*}
As $\mg(\boldsymbol{\eta}, \mZ_1)$ and $\mg(\boldsymbol{\eta}, \tilde\mZ_1)$ are almost surely contained in the compact set $\mg(\cK \times(\tilde \cZ \cup \cZ))$, so that
$$
\sup_{\boldsymbol{\eta}' \in \cK}\|\E[\mg(\boldsymbol{\eta}',\mZ_1)]\|_2 + \sup_{\boldsymbol{\eta}' \in \cK} \|\E[\mg(\boldsymbol{\eta}',\tilde \mZ_1) \mid \cO] \|_2 \leq 2\sup_{(\boldsymbol{\eta}, \mz) \in \cK \times (\tilde\cZ \cup \cZ)} \|\mg(\boldsymbol{\eta}, \mz)\|_2 < \infty.
$$
This yields the following upper bound for \eqref{eq:han3}:
\begin{align}\label{eq:han4}
2\sup_{(\boldsymbol{\eta}, \mz) \in \cK \times (\tilde\cZ \cup \cZ)} \|\mg(\boldsymbol{\eta}, \mz)\|_2 \cdot \sup_{\boldsymbol{\eta}\in \cK}\|\E[\mg(\boldsymbol{\eta},\mZ_1)] - \E[\mg(\boldsymbol{\eta},\tilde \mZ_1) \mid \cO] \|_2. 
\end{align}
Taking any joint distribution $\pi$ between $\bP_{\tilde Z \mid \cO}$ and $\P_{Z}$ and letting  $(\tilde \mZ, \mZ) \sim \pi$,  we have
\begin{equation}
\sup_{\boldsymbol{\eta} \in \cK}\|\E[\mg(\boldsymbol{\eta},\mZ_1)] - \E[\mg(\boldsymbol{\eta},\tilde \mZ_1) \mid \cO] \|_2 = \sup_{\boldsymbol{\eta} \in \cK}\|\E_\pi[\mg(\boldsymbol{\eta},\mZ) - \mg(\boldsymbol{\eta},\tilde\mZ)]\|_2 \label{eq:crossterm}
\end{equation}
since $\tilde\mZ \sim \bP_{\tilde Z\mid\cO}$ and $\mZ \sim \P_{Z}$ so that $\E[\mg(\boldsymbol{\eta}, \mZ_1)] = \E_\pi[\mg(\boldsymbol{\eta}, \mZ)]$ and $\E[\mg(\boldsymbol{\eta},\tilde \mZ_1) \mid \cO] =  \E_\pi[\mg(\boldsymbol{\eta},\tilde\mZ)]$. 
Take any $\epsilon > 0$. Since $\mg$ is continuous on a compact set $\cK \times (\cZ \cup \tilde\cZ)$, and, hence, uniformly continuous, we choose a $\delta >0 $ such that
\begin{equation}\label{eq:supeta:aboveineq}
\sup_{\boldsymbol{\eta} \in \cK}\sup_{\mz, \mz' \in \cZ \cup \tilde\cZ : \|\mz-\mz'\|\leq \delta} \|\mg(\boldsymbol{\eta}, \mz) - \mg(\boldsymbol{\eta}, \mz')\|_2 < \epsilon.
\end{equation}
Leveraging Jensen's inequality and \eqref{eq:supeta:aboveineq}, we obtain
\al{
\sup_{\boldsymbol{\eta} \in \cK}\|\E_\pi[\mg(\boldsymbol{\eta},\mZ) - \mg(\boldsymbol{\eta},\tilde\mZ)]\|_2  
&\leq \sup_{\boldsymbol{\eta} \in \cK}\E_\pi[\|\mg(\boldsymbol{\eta},\mZ) - \mg(\boldsymbol{\eta},\tilde\mZ)\|_2 \cdot \mathds{1}(\|\mZ - \tilde\mZ\|_2 > \delta)] + \epsilon\\
&\leq \frac{1}{\delta} \cdot 2\max_{(\boldsymbol{\eta},\mz) \in \cK \times (\cZ \cup \tilde\cZ)}\|\mg(\boldsymbol{\eta},\mz)\|_2 \cdot \E_\pi[\|\mZ - \tilde\mZ\|_2] + \epsilon \tag{Markov's inequality}.
}
Putting together and taking an infimum over $\pi$, an upper bound for \eqref{eq:han4} is then
$$
\frac{4}{\delta} \cdot \pa{\max_{(\boldsymbol{\eta},\mz) \in \cK \times (\cZ \cup \tilde\cZ)}\|\mg(\boldsymbol{\eta},\mz)\|_2 }^2\cdot \sW_1\big(\bP_{\tilde Z\mid \cO},\, \P_Z\big) + 2 \max_{(\boldsymbol{\eta},\mz) \in \cK \times (\cZ \cup \tilde\cZ)}\|\mg(\boldsymbol{\eta},\mz)\|_2\cdot \epsilon.
$$ 
By Assumption \ref{assump_bootstrap}(b),  $\sW_1\big(\bP_{\tilde Z\mid \cO}, \P_Z\big)  = o_{\P_{\cO}}(1)$, so taking $n\rightarrow \infty$ and then $\epsilon \rightarrow 0$, as $\epsilon$ was arbitrary, we obtain the desired conclusion. 

Equation \eqref{eq:han2} is established in an identical way, and we thus complete the whole proof.
\end{proof}

\begin{lemma}\label{cdllinear} 
\begin{enumerate}[label=(\alph*)]
\item Under Assumptions \ref{assump_ndata}, \ref{assump_l} and \ref{assump_eta},
\[
    \sqrt{n} (\hat{\boldsymbol\eta}_n - \boldsymbol\eta_0) = -\frac{1}{\sqrt{n}} \sum_{i=1}^n (\sD_{\boldsymbol{\eta}}^2\E[\sL(\boldsymbol\eta_0, \mZ)])^{-1} \sD_{\boldsymbol{\eta}}\sL(\boldsymbol\eta_0, \mZ_i) + o_{\P_{\cO}}(1).
\]
\item    Assuming further Assumption \ref{assump_noise1}, \ref{assump_bootstrap} and Assumption \ref{assump_bsmest}, it holds true that 
\begin{enumerate}[label=(\roman*)]
\item $\|\tilde{\boldsymbol\eta}_0 - \boldsymbol\eta_0\|_2 = o_{\P_{\cO}}(1)$; 
\item  $\P\Big(\sD_{\boldsymbol{\eta}}^2\E[\sL(\tilde{\boldsymbol\eta}_0, \tilde\mZ) \mid \cO]) \text{ is invertible}\Big) \rightarrow 1$;
\item  $\|\sD_{\boldsymbol{\eta}}^2\E[\sL(\tilde{\boldsymbol\eta}_0, \tilde\mZ) \mid \cO]^{-1} - \sD_{\boldsymbol{\eta}}^2\E[\sL(\boldsymbol\eta_0, \mZ)]^{-1}\|_{\rm op} = o_{\P}(1).$
\item  $\sqrt{n}(\tilde{\boldsymbol\eta}_n -\tilde{\boldsymbol\eta}_0) = -\frac{1}{\sqrt{n}} \sum_{i=1}^n (\sD_{\boldsymbol{\eta}}^2\E[\sL(\tilde{\boldsymbol\eta}_0, \tilde\mZ) \mid \cO])^{-1}\sD_{\boldsymbol{\eta}}\sL(\tilde{\boldsymbol\eta}_0, \tilde \mZ_i) + o_{\P_{\cO\tilde\mU}}(1)$.
        \end{enumerate}
    \end{enumerate}
\end{lemma}

\begin{proof}
Lemma \ref{cdllinear}(a) follows from Lemma 5.23 of \cite{Vaart_1998}. It remains to prove Lemma \ref{cdllinear}(b).

\noindent\textbf{Step 1: Show Lemma \ref{cdllinear}(b)(i).}

     Since $\boldsymbol{\eta}_0$ maximizes the function $\boldsymbol{\eta} \mapsto \E[\sL({\boldsymbol{\eta}}, \mZ)]$ by Assumption \ref{assump_eta}(b), we obtain
    $$
    0  \leq \E[\sL({\boldsymbol{\eta}}_0, \mZ) ]  -\E[\sL(\tilde{\boldsymbol{\eta}}_0, \mZ)].  
    $$
    Rewriting the right-hand side,
    \al{
    \E[\sL({\boldsymbol{\eta}}_0, \mZ) ]  -\E[\sL(\tilde{\boldsymbol{\eta}}_0, \mZ)] =& \E[\sL({\boldsymbol{\eta}}_0, \mZ) ]  -\E[\sL({\boldsymbol{\eta}}_0, \tilde\mZ) \mid \cO]\\
    &+\E[\sL({\boldsymbol{\eta}}_0, \tilde\mZ) \mid \cO]  - \E[\sL(\tilde{\boldsymbol{\eta}}_0, \tilde\mZ) \mid \cO]\\
    &+ \E[\sL(\tilde{\boldsymbol{\eta}}_0, \tilde\mZ) \mid \cO]-\E[\sL(\tilde{\boldsymbol{\eta}}_0, \mZ)].
    }
    The function $\sL(\cdot,\cdot)$ is jointly continuous by Assumption \ref{assump_l}(b) and (c). Also, $\boldsymbol{\eta}_0 \in \cK$ by Assumption \ref{assump_eta}(a). Additionally, $\tilde{\boldsymbol{\eta}}_0 \in \cK$ by definition. Applying Lemma \ref{lemma:w2cvg} then yields
    $$
    \E[\sL({\boldsymbol{\eta}}_0, \mZ) ]  - \E[\sL({\boldsymbol{\eta}}_0, \tilde\mZ) \mid \cO] + \E[\sL(\tilde{\boldsymbol{\eta}}_0, \tilde\mZ) \mid \cO]-\E[\sL(\tilde{\boldsymbol{\eta}}_0, \mZ)] = o_{\P_{\cO}}(1).
    $$
    By definition, $\tilde{\boldsymbol{\eta}}_0$ is a maximizer of $\boldsymbol{\eta} \mapsto \E[\sL({\boldsymbol{\eta}}_, \tilde\mZ) \mid \cO]$ so that
    $$
    \E[\sL({\boldsymbol{\eta}}_0, \tilde\mZ) \mid \cO]  - \E[\sL(\tilde{\boldsymbol{\eta}}_0, \tilde\mZ) \mid \cO]\leq 0
    $$
    almost surely.
    Putting this all together,
    $$
    0 \leq \E[\sL({\boldsymbol{\eta}}_0, \mZ) ]  -\E[\sL(\tilde{\boldsymbol{\eta}}_0, \mZ)] \leq o_{\P_{\cO}}(1).
    $$
    Since $\boldsymbol{\eta}_0$ is the unique maximizer of $\boldsymbol{\eta} \mapsto \E[\sL({\boldsymbol{\eta}}_, \mZ)]$ by Assumption \ref{assump_eta}(b), $\|\tilde{\boldsymbol\eta}_0 - \boldsymbol\eta_0\|_2 = o_{\P_{\cO}}(1)$ as desired.

~\\    
\noindent\textbf{Step 2: Show Lemma \ref{cdllinear}(b)(ii) and Lemma \ref{cdllinear}(b)(iii).}

By the mean value theorem and the bounded convergence theorem,
$$
\sD_{\boldsymbol{\eta}}^2\E[\sL(\tilde{\boldsymbol{\eta}}_0, \tilde\mZ) \mid \cO] = \E[\sD_{\boldsymbol{\eta}}^2\sL(\tilde{\boldsymbol{\eta}}_0, \tilde\mZ) \mid \cO] 
$$
almost surely. Thus, it suffices to show the same conclusion for $\E[\sD_{\boldsymbol{\eta}}^2\sL(\tilde{\boldsymbol{\eta}}_0, \tilde\mZ) \mid \cO]$. 

 First, note $\boldsymbol{\eta} \mapsto \E[\sD_{\boldsymbol{\eta}}^2\sL(\boldsymbol{\eta}, \mZ)]$ is continuous by the mean value theorem. Applying Lemma \ref{lemma:w2cvg} to the function $\sD_{\boldsymbol{\eta}}^2\sL(\boldsymbol{\eta}, \mz)$, we obtain
\begin{equation}\label{eq:asympexp:precedingdisplay}
\sup_{\boldsymbol{\eta} \in \cK }\Big\|\E[\sD_{\boldsymbol{\eta}}^2\sL(\boldsymbol{\eta}, \tilde\mZ) \mid \cO] - \E[\sD_{\boldsymbol{\eta}}^2\sL(\boldsymbol{\eta}, \mZ)]\Big\|_{\rm max} = o_{\P}(1). 
\end{equation}
Then, by \textbf{Step 1} and \eqref{eq:asympexp:precedingdisplay},
\al{
\|\E[\sD_{\boldsymbol{\eta}}^2\sL(\tilde{\boldsymbol{\eta}}_0, \tilde\mZ) \mid \cO] - \E[\sD_{\boldsymbol{\eta}}^2\sL(\boldsymbol{\eta}_0, \mZ) ]\|_{\rm max} &\leq \|\E[\sD_{\boldsymbol{\eta}}^2\sL(\tilde{\boldsymbol{\eta}}_0, \tilde\mZ) \mid \cO] - \E[\sD_{\boldsymbol{\eta}}^2\sL(\tilde{\boldsymbol{\eta}}_0, \mZ) ]\|_{\rm max} \\
&+\|\E[\sD_{\boldsymbol{\eta}}^2\sL(\tilde{\boldsymbol{\eta}}_0, \mZ)] - \E[\sD_{\boldsymbol{\eta}}^2\sL(\boldsymbol{\eta}_0, \mZ) ]\|_{\rm max} \\
&= o_{\P}(1).
}

Next, we note a consequence of the property of matrix inversion. Pick any $\epsilon' > 0$. Since the set of invertible matrices is an open set with respect to  $\|\cdot \|_{\max }$, and matrix inversion is continuous, there is a fixed $\delta' > 0$ for which any matrix $\mathbf{M} \in \R^{q \times q}$ that satisfies $\|\mathbf{M} - \E[\sD_{\boldsymbol{\eta}}^2\sL(\boldsymbol{\eta}_0, \mZ) ]\|_{\rm max} < \delta'$ is invertible and also satisfies $\|\mathbf{M}^{-1} - \E[\sD_{\boldsymbol{\eta}}^2\sL(\boldsymbol{\eta}_0, \mZ) ]^{-1}\|_{\rm max} < \epsilon'$. The existence of $\E[\sD_{\boldsymbol{\eta}}^2\sL(\boldsymbol{\eta}_0, \mZ) ]^{-1}$ is from 
Assumption \ref{assump_l}(d).

Then,
\al{
\P\pa{\E[\sD_{\boldsymbol{\eta}}^2\sL(\tilde{\boldsymbol{\eta}}_0, \tilde\mZ) \mid \cO]^{-1} \text{ exists}, \|\sD_{\boldsymbol{\eta}}^2\E[\sL(\tilde{\boldsymbol\eta}_0, \tilde\mZ) \mid \cO]^{-1} - \sD_{\boldsymbol{\eta}}^2\E[\sL(\boldsymbol\eta_0, \mZ)]^{-1}\|_{\rm max} < \epsilon' } \geq\\ \P\pa{\|\E[\sD_{\boldsymbol{\eta}}^2\sL(\tilde{\boldsymbol{\eta}}_0, \tilde\mZ) \mid \cO] - \E[\sD_{\boldsymbol{\eta}}^2\sL(\boldsymbol{\eta}_0, \mZ) ]\|_{\rm max} < \delta'} 
}
by the definition of $\delta'$. 
Since $\delta'$ is fixed,
$$
\P\pa{\|\E[\sD_{\boldsymbol{\eta}}^2\sL(\tilde{\boldsymbol{\eta}}_0, \tilde\mZ) \mid \cO] - \E[\sD_{\boldsymbol{\eta}}^2\sL(\boldsymbol{\eta}_0, \mZ) ]\|_{\rm max} < \delta'} \rightarrow 1,
$$
which means 
$$
\P\pa{\E[\sD_{\boldsymbol{\eta}}^2\sL(\tilde{\boldsymbol{\eta}}_0, \tilde\mZ) \mid \cO]^{-1} \text{ exists}, \|\sD_{\boldsymbol{\eta}}^2\E[\sL(\tilde{\boldsymbol\eta}_0, \tilde\mZ) \mid \cO]^{-1} - \sD_{\boldsymbol{\eta}}^2\E[\sL(\boldsymbol\eta_0, \mZ)]^{-1}\|_{\rm max} < \epsilon'} \rightarrow 1.
$$
\noindent\textbf{Step 4: Show Lemma \ref{cdllinear}(b)(iv).}

Since $\boldsymbol{\eta}_0$ is an interior point of $\cK$, we take $d > 0$ small enough so that $\cB({\boldsymbol{\eta}}_0, d, \|\cdot\|_2 )$ is contained in the interior of $\cK$. We have just shown $\|\tilde{\boldsymbol\eta}_0 - \boldsymbol\eta_0\|_2 = o_{\P_{\cO}}(1)$. By Assumption \ref{assump_bsmest}(a), $\|\tilde{\boldsymbol\eta}_n - \tilde{\boldsymbol\eta}_0\|_2 = o_{\P_{\cO\tilde U}}(1)$ so that we assume without loss of generality that $\tilde {\boldsymbol{\eta}}_n, \tilde{\boldsymbol{\eta}}_0$ are within this ball.

Due to the mean value theorem and the characterization of $\tilde {\boldsymbol{\eta}}_0$ as a maximizer, we interchange the derivative and expectation to obtain
\begin{equation}\label{eq:asympexp:zeroderivativecondition}
 \mathbf{0} = \sD_{\boldsymbol{\eta}} \E[\sL(\tilde{\boldsymbol\eta}_0, \tilde\mZ) \mid \cO] = \E[\sD_{\boldsymbol{\eta}}\sL(\tilde{\boldsymbol{\eta}}_0, \tilde\mZ) \mid \cO]
\end{equation}
almost surely.

We now verify the conditions \eqref{eq:thm552_bd1} and \eqref{eq:thm552_bd2} of Lemma \ref{lemma:vdv552}. An application of Taylor's theorem (Theorem 12.14 of \cite{apostol1974mathematical}, e.g.) then gives, using \eqref{eq:asympexp:zeroderivativecondition},
\al{
 |\E[\sL({\boldsymbol\eta}, \tilde\mZ) - \sL(\tilde{\boldsymbol\eta}_0, \tilde\mZ)\mid \cO]| \leq \frac{q^2}{2} \cdot \max_{\mz \in \tilde\cZ}\max_{\boldsymbol{\eta'} \in \cK} \|\sD_{\boldsymbol{\eta}}^2\sL(\boldsymbol{\eta}', \mz)\|_{\rm max} \cdot \|{\boldsymbol\eta} - \tilde{\boldsymbol\eta}_0\|_2^2
}
for any $\boldsymbol{\eta}$ in the interior of $\cK$. Thus, for all $\delta < d/3$, by the definition of $d$ and the triangle inequality, $\cB(\tilde{\boldsymbol{\eta}}_0, \delta, \|\cdot\|_2 )$ is contained in the interior of $\cK$. As a result, we have
$$
\sup_{\|\boldsymbol{\eta} - \tilde{\boldsymbol{\eta}}_0 \|_2 < \delta}  |\E[\sL({\boldsymbol\eta}, \tilde\mZ) - \sL(\tilde{\boldsymbol\eta}_0, \tilde\mZ)\mid \cO]|\leq \frac{q^2}{2} \cdot \max_{\mz \in \tilde\cZ}\max_{\boldsymbol{\eta'} \in \cK} \|\sD_{\boldsymbol{\eta}}^2\sL(\boldsymbol{\eta}', \mz)\|_{\rm max} \cdot \delta^2
$$
and thus verify the condition \eqref{eq:thm552_bd1}.

Condition \eqref{eq:thm552_bd2} is directly verified in Lemma \ref{lemma:Ldonsker}.  Consequently, $\|\tilde{\boldsymbol\eta}_n - \tilde{\boldsymbol\eta}_0\|_2 = O_{\P_{\cO\tilde U }}(n^{-1/2})$.

Twice-differentiability of  $\boldsymbol{\eta} \mapsto \E[\sL(\boldsymbol{\eta}, \tilde\mZ) \mid \cO]$ almost surely, combined with
$\sD_{\boldsymbol{\eta}}\E[\sL(\tilde{\boldsymbol{\eta}}_0, \tilde\mZ)\mid \cO] = \mathbf{0}$ almost surely, implies
$$
\E[\sL(\tilde{\boldsymbol{\eta}}_n, \tilde\mZ) \mid \cO] - \E[\sL(\tilde{\boldsymbol{\eta}}_0, \tilde\mZ) \mid \cO] - \frac{1}{2}(\tilde{\boldsymbol{\eta}}_n - \tilde{\boldsymbol{\eta}}_0)^\top\sD_{\boldsymbol{\eta}}^2\E[\sL(\tilde{\boldsymbol{\eta}}_0, \tilde\mZ) \mid \cO](\tilde{\boldsymbol{\eta}}_n - \tilde{\boldsymbol{\eta}}_0) = o_{\P_{\cO\tilde U }}(n^{-1}). 
$$
Using Lemma \ref{lemma:donsker_verify} on $\tilde{\boldsymbol{\eta}}_n$, and manipulating terms, we obtain then
\begin{equation}\label{eq:cdllinear:tilde}
 \pa{-\frac{1}{\sqrt{n}}\sum_{i=1}^n \sD_{\boldsymbol{\eta}}\sL(\tilde{\boldsymbol{\eta}}_0, \tilde{\mZ_i})}^\top \sqrt{n}(\tilde{\boldsymbol{\eta}}_n - \tilde{\boldsymbol{\eta}}_0) - \frac{1}{2}(\sqrt{n}(\tilde{\boldsymbol{\eta}}_n - \tilde{\boldsymbol{\eta}}_0))^\top\sD_{\boldsymbol{\eta}}^2\E[\sL(\tilde{\boldsymbol{\eta}}_0, \tilde\mZ) \mid \cO](\sqrt{n}(\tilde{\boldsymbol{\eta}}_n - \tilde{\boldsymbol{\eta}}_0)) = \tilde\epsilon_n,
\end{equation}
where 
$$
\tilde\epsilon_n :=  \sum_{i=1}^n   \sL(\tilde{\boldsymbol{\eta}}_0, \tilde{\mZ_i}) - \sL(\tilde{\boldsymbol{\eta}}_n, \tilde{\mZ_i}) + o_{\P}(1).
$$
Since $\P\pa{\E[\sD_{\boldsymbol{\eta}}^2\sL(\tilde{\boldsymbol{\eta}}_0, \tilde\mZ) \mid \cO]^{-1} \text{ exists}} \rightarrow 1$ by \textbf{Step 2}, it is no loss of generality to assume $\E[\sD_{\boldsymbol{\eta}}^2\sL(\tilde{\boldsymbol{\eta}}_0, \tilde\mZ) \mid \cO]^{-1}$ exists for the remainder of the proof. 

Define 
$$
\check{\boldsymbol{\eta}}_n = \tilde{\boldsymbol{\eta}}_0
-\frac{1}{n} \sum_{i=1}^n (\sD_{\boldsymbol{\eta}}^2\E[\sL(\tilde{\boldsymbol\eta}_0, \tilde\mZ) \mid \cO])^{-1}\sD_{\boldsymbol{\eta}}\sL(\tilde{\boldsymbol\eta}_0, \tilde \mZ_i)
$$
and observe, by the Lyapunov central limit theorem, that $\check{\boldsymbol{\eta}}_n - \tilde{\boldsymbol{\eta}}_0 = O_{\P}(n^{-1/2})$. Using $\check{\boldsymbol{\eta}}_n$ in Lemma \ref{lemma:donsker_verify}, we can show
\begin{equation}\label{eq:cdllinear:check}
 \pa{-\frac{1}{\sqrt{n}}\sum_{i=1}^n \sD_{\boldsymbol{\eta}}\sL(\tilde{\boldsymbol{\eta}}_0, \tilde{\mZ_i})}^\top \sqrt{n}(\check{\boldsymbol{\eta}}_n - \tilde{\boldsymbol{\eta}}_0) - \frac{1}{2}(\sqrt{n}(\check{\boldsymbol{\eta}}_n - \tilde{\boldsymbol{\eta}}_0))^\top\sD_{\boldsymbol{\eta}}^2\E[\sL(\tilde{\boldsymbol{\eta}}_0, \tilde\mZ) \mid \cO](\sqrt{n}(\check{\boldsymbol{\eta}}_n - \tilde{\boldsymbol{\eta}}_0)) = \check\epsilon_n,
\end{equation}
where 
$$
\check\epsilon_n :=  \sum_{i=1}^n   \sL(\tilde{\boldsymbol{\eta}}_0, \tilde{\mZ_i}) - \sL(\check{\boldsymbol{\eta}}_n, \tilde{\mZ_i}) + o_{\P}(1).
$$
Expanding the definition of $\check{\boldsymbol{\eta}}_n$ and simplifying Equation \eqref{eq:cdllinear:check}, we get
\begin{equation}\label{eq:cdllinear:check_simple}
\frac{1}{2}\pa{\frac{1}{\sqrt{n}}\sum_{i=1}^n \sD_{\boldsymbol{\eta}}\sL(\tilde{\boldsymbol{\eta}}_0, \tilde{\mZ_i})}^\top \sD_{\boldsymbol{\eta}}^2\E[\sL(\tilde{\boldsymbol{\eta}}_0, \tilde\mZ_1) \mid \cO]^{-1}\pa{\frac{1}{\sqrt{n}}\sum_{i=1}^n  \sD_{\boldsymbol{\eta}}\sL(\tilde{\boldsymbol{\eta}}_0, \tilde{\mZ_i})} = \check \epsilon_n.
\end{equation}
Subtracting \eqref{eq:cdllinear:check_simple} from \eqref{eq:cdllinear:tilde}, then completing the square, we get
\al{
\frac{1}{2}\Big\|(-\sD_{\boldsymbol{\eta}}^2\E[\sL(\tilde{\boldsymbol{\eta}}_0, \tilde\mZ) \mid \cO])^{1/2} \sqrt{n}(\tilde{\boldsymbol{\eta}}_n - \tilde{\boldsymbol{\eta}}_0) - \frac{1}{\sqrt{n}}\sum_{i=1}^n (-\sD_{\boldsymbol{\eta}}^2\E[\sL(\tilde{\boldsymbol{\eta}}_0, \tilde\mZ) \mid \cO])^{-1/2} \sD_{\boldsymbol{\eta}}\sL(\tilde{\boldsymbol{\eta}}_0, \tilde{\mZ}_i) \Big\|_2^2 =
\tilde \epsilon_n - \check \epsilon_n.
}
The matrix $-\sD_{\boldsymbol{\eta}}^2\E[\sL(\tilde{\boldsymbol{\eta}}_0, \tilde\mZ) \mid \cO]$ is almost surely symmetric, as $\sL$ is twice continuously differentiable and swapping second derivatives and expectations. In addition, $-\sD_{\boldsymbol{\eta}}^2\E[\sL(\tilde{\boldsymbol{\eta}}_0, \tilde\mZ) \mid \cO]$ is almost surely positive definite because $\tilde{\boldsymbol{\eta}}_0$ is a maximizer, as $\sD_{\boldsymbol{\eta}}^2\E[\sL(\tilde{\boldsymbol{\eta}}_0, \tilde\mZ) \mid \cO]$ is negative definite, almost surely. Hence, the matrix square root exists for $-\sD_{\boldsymbol{\eta}}^2\E[\sL(\tilde{\boldsymbol{\eta}}_0, \tilde\mZ) \mid \cO]$ and its inverse. 

Expanding the definition of $\tilde \epsilon_n$ and $\check \epsilon_n$, 
\al{
\tilde \epsilon_n - \check \epsilon_n &= \sum_{i=1}^n   \sL(\check{\boldsymbol{\eta}}_n, \tilde{\mZ_i}) - \sL(\tilde{\boldsymbol{\eta}}_n, \tilde{\mZ_i}) + o_{\P}(1)\\
&=\sum_{i=1}^n   \sL(\check{\boldsymbol{\eta}}_n, \tilde{\mZ_i}) - \sup_{\boldsymbol{\eta} \in \cK} \sum_{i=1}^n   \sL({\boldsymbol{\eta}}, \tilde{\mZ_i}) + \sup_{\boldsymbol{\eta} \in \cK} \sum_{i=1}^n   \sL({\boldsymbol{\eta}}, \tilde{\mZ_i})- \sL(\tilde{\boldsymbol{\eta}}_n, \tilde{\mZ_i}) + o_{\P}(1)\\
&=\sum_{i=1}^n   \sL(\check{\boldsymbol{\eta}}_n, \tilde{\mZ_i}) - \sup_{\boldsymbol{\eta} \in \cK} \sum_{i=1}^n   \sL({\boldsymbol{\eta}}, \tilde{\mZ_i}) + o_{\P}(1). \tag{Assumption \ref{assump_bsmest}(b)}\\
&\leq o_{\P}(1).
}
Therefore, 
$$
(-\sD_{\boldsymbol{\eta}}^2\E[\sL(\tilde{\boldsymbol{\eta}}_0, \tilde\mZ) \mid \cO])^{1/2} \sqrt{n}(\tilde{\boldsymbol{\eta}}_n - \tilde{\boldsymbol{\eta}}_0) - \frac{1}{\sqrt{n}}\sum_{i=1}^n (-\sD_{\boldsymbol{\eta}}^2\E[\sL(\tilde{\boldsymbol{\eta}}_0, \tilde\mZ) \mid \cO])^{-1/2} \sD_{\boldsymbol{\eta}}\sL(\tilde{\boldsymbol{\eta}}_0, \tilde{\mZ}_i) = o_{\P}(1).
$$
Factoring out $(-\sD_{\boldsymbol{\eta}}^2\E[\sL(\tilde{\boldsymbol{\eta}}_0, \tilde\mZ) \mid \cO])^{1/2}$ and noting, from \textbf{Step 2},
$$
\|(-\sD_{\boldsymbol{\eta}}^2\E[\sL(\tilde{\boldsymbol{\eta}}_0, \tilde\mZ) \mid \cO])^{1/2} - (-\sD_{\boldsymbol{\eta}}^2\E[\sL({\boldsymbol{\eta}}_0, \mZ) ])^{1/2}\|_{\rm op} = o_{\P}(1),
$$
we have
$$
\sqrt{n}(\tilde{\boldsymbol{\eta}}_n - \tilde{\boldsymbol{\eta}}_0) + \frac{1}{\sqrt{n}}\sum_{i=1}^n \sD_{\boldsymbol{\eta}}^2\E[\sL(\tilde{\boldsymbol{\eta}}_0, \tilde\mZ) \mid \cO]^{-1} \sD_{\boldsymbol{\eta}}\sL(\tilde{\boldsymbol{\eta}}_0, \tilde{\mZ}_i) = o_{\P}(1)
$$
and thus complete the whole proof.
\end{proof}

\subsection{Supporting lemmas for Theorem \ref{thm:iso}}
 
We first establish the uniform convergence of the regression functions and their first derivatives.

\begin{lemma}\label{lemma:isocvg}
Suppose that Assumptions \ref{assump_bootstrap}, \ref{assump_iso}, and \ref{ass:generator} hold. Then, for any bounded closed interval $\cC \subseteq \R$, it holds true that
\begin{align*}
&\sup_{\mz \in \cX \times \cC}\Big|\tilde p_n(\mz) - p_Z(\mz)\Big| = o_{\P}(1),~~~ \sup_{\mz \in  \cX \times \cC}\Big\|\sD\tilde p_n(\mz) - \sD p_Z(\mz)\Big\|_2 = o_{\P}(1),\\
&\sup_{x \in  \cX}\Big|\tilde p_X(x) - p_X(x)\Big| = o_{\P}(1),~~~\sup_{x \in  \cX}\Big|\tilde p_X'(x) - p_X'(x)\Big| = o_{\P}(1),\\
&\sup_{x \in  \cX}\Big|\tilde f_0(x) - f_0(x)\Big| = o_{\P}(1),~~~{\rm and}~~~ \sup_{x \in  \cX}\Big|\tilde f_0'(x) - f_0'(x)\Big| = o_{\P}(1).
\end{align*}
\end{lemma}
\begin{proof}

We denote $\tilde f_0$ as $\tilde f_{0,n}$, and $\tilde f'_0$ as $\tilde f'_{0,n}$ and analogously for the marginal densities in this proof to emphasize the approximating sequences' dependence on the sample size.  

~\\
{\bf Step 1.} We appeal to Lemma \ref{lemma:prob_and_as}: for each subsequence $n_k$, we aim to find a further subsequence $n_{k_{\ell_m}}$ such that 
$$
\sup_{\mz \in  \cX \times \cC}\Big| \tilde p_{n_{k_{\ell_m}}}(\mz) - p_Z(\mz)\Big| \rightarrow 0
$$
almost surely.

By Lemma \ref{lemma:margdiff}, $\tilde p_{n}$ is Lipschitz on $ \cX \times \R$, with probability $1$. Hence, almost surely, the sequence of functions $\tilde p_{n}$'s is uniformly Lipschitz on $ \cX \times \cC$ for all sufficiently large $n$. Also, they are uniformly bounded on $ \cX \times \cC$, by $\tilde K$ by Assumption \ref{ass:generator}(b).

Furthermore, $\tilde p_{n}$ is a $C( \cX\times \cC)$-valued random variable for sufficiently large $n$ due to Assumption \ref{ass:generator}(a,c). By Arzela-Ascoli theorem (see Theorem 11.28 of \cite{rudin1987real} e.g.), there then exists a compact $\tilde \cC \subseteq C( \cX \times \cC)$ for which
$$
\P(\tilde p_{n} \in \tilde \cC)  = 1
$$
for sufficiently large $n$, yielding tightness of the sequences. Take any subsequence $n_k$. Applying Lemma \ref{lemma:prokhorov}, then Assumption \ref{assump_bootstrap}(b) along with Lemma \ref{lemma:prob_and_as}, there exists a further subsequence $n_{k_\ell}$ for which 
$$
\tilde p_{{n_{k_\ell}}} \text{ converges weakly to } \phi \qquad \text{and \qquad } \sW_1(\bP_{\tilde Z_{1, n_{k_\ell}} \mid \cO} , \P_{Z}) \rightarrow 0 \text{ almost surely}
$$
for some $C( \cX \times \cC)$-valued random variable $\phi$. Note that $\phi$ is a density almost surely because $\tilde p_{n_{k_\ell}}$ are densities.

Our next goal is to identify the limit $\phi$ as the (deterministic) function $p_Z$. Choose any bounded, continuous function $f : \R^2 \rightarrow \R$. As $\mZ$ admits a Lebesgue density by Assumption \ref{assump_iso}(c), the function $\mz \mapsto f(\mz) \cdot \mathds{1}(\mz \in  \cX \times \cC)$ is $\P_{Z}$-almost surely continuous. Lemma \ref{lemma:weakcvg_wass} and $\sW_1(\bP_{\tilde Z_{1, n_{k_\ell}} \mid \cO} , \P_{Z}) \rightarrow 0$ then implies, almost surely,
$$
\tilde \mZ_{n_{k_\ell}} \text{ converges weakly to } \mZ \text{ conditional on }\cO.
$$
The continuous mapping theorem then implies
$$
f(\tilde \mZ_{n_{k_\ell}}) \cdot \mathds{1}(\tilde \mZ_{n_{k_\ell}} \in  \cX \times \cC) \text{ converges weakly to } f(\mZ) \cdot \mathds{1}(\mZ \in  \cX \times \cC) \text{ conditional on }
\cO.
$$
The sequence $f(\tilde \mZ_{n_{k_\ell}}) \cdot \mathds{1}(\tilde \mZ_{n_{k_\ell}} \in  \cX \times \cC)$ is uniformly bounded. Therefore, this sequence is uniformly integrable, so that
\begin{equation}
\E\Big[f(\tilde \mZ_{n_{k_\ell}}) \cdot \mathds{1}(\tilde \mZ_{n_{k_\ell}} \in  \cX \times \cC) \mid \cO\Big] - \E\Big[f(\mZ) \cdot \mathds{1}(\mZ \in  \cX \times \cC)\Big] \rightarrow 0. \label{eq:isocvg:rv_wkcvg}
\end{equation}
almost surely. 
Define the function $g_f : C( \cX\times \cC)  \rightarrow \R$ as
$$
g_f(h) := \int_{ \cX \times \cC} f(\mz) \cdot (h(\mz)-p_Z(\mz)) \, \d \mz,
$$
which, we note, is continuous.  By the continuous mapping theorem,
$$
g_f(\tilde p_{n_{k_\ell}}) \text{ converges weakly to } g_f(\phi).
$$
By definition, 
\al{
g_f(\tilde p_{n_{k_\ell}}) = &
\E\Big[f(\tilde \mZ_{n_{k_\ell}}) \cdot \mathds{1}(\tilde \mZ_{n_{k_\ell}} \in  \cX \times \cC) \mid \cO\Big] - \E\Big[f(\mZ) \cdot \mathds{1}(\mZ \in  \cX \times \cC)\Big]
}
 and we have shown $g_f(\tilde p_{n_{k_\ell}})$ converges to $0$ almost surely by \eqref{eq:isocvg:rv_wkcvg}. Thus, $g_f(\phi) =0$ almost surely, yielding
 $$
 \int_{ \cX \times \cC} f(\mz) \cdot (\phi(\mz) - p_Z(\mz)) \, \d \mz = 0
 $$
 almost surely, for a fixed choice of $f$. Choosing the sequence of bounded continuous functions $f_1, f_2, ... : (\cX \times \cC)  \rightarrow \R$ guaranteed by Lemma \ref{lemma:equality_of_measures_countable}, we have
 $$
 \int_{ \cX \times \cC} f_i(\mz) \cdot (\phi(\mz) - p_Z(\mz)) \, \d \mz = 0
 $$
 almost surely, for all $i=1,2,...$ simultaneously, as the countable intersection of almost sure events is also almost sure. Consequently, by Lemma \ref{lemma:equality_of_measures_countable}, the measures defined by $\phi$ and $p_Z$ are equal almost surely, which implies that the density functions $\phi$ and $p_Z$ are equal for Lebesgue almost all points $ \cX \times \cC$, almost surely. Since $\phi$ and $p_Z$ are continuous on $ \cX \times \cC$, we have equality of $\phi$ and $p_Z$ at all points $ \cX \times \cC$, almost surely. Lemma \ref{lemma:weak_to_prob} then yields
 $$
 \sup_{\mz \in  \cX \times \cC }\Big|\tilde p_{n_{k_\ell}}(\mz) - p_Z(\mz)\Big| = o_{\P}(1).
 $$
Finally, pick a further subsequence of $n_{k_\ell}$, $n_{k_{\ell_m}}$, so that Lemma \ref{lemma:prob_and_as} deducing
$$
 \sup_{\mz \in  \cX \times \cC }\Big|\tilde p_{n_{k_{\ell_m}}}(\mz) - p_Z(\mz)\Big| \rightarrow 0
$$
almost surely, which yields the final result.

~\\
\noindent{\bf Step 2.} Next, we show 
$$
 \sup_{\mz \in  \cX \times \cC}\|\sD\tilde p_{n}(\mz) - \sD p_Z(\mz)\|_2 = o_{\P}(1).
$$
We show this holds componentwise and denote the partial derivative with respect to the first and second argument of $\tilde p_n$ as $\sD_x$ and $\sD_y$, respectively.

We again appeal to Lemma \ref{lemma:prob_and_as}. Lemma \ref{lemma:margdiff} ensures that the functions $\sD_x \tilde p_{n}$'s, for all sufficiently large $n$, are uniformly bounded and uniformly equicontinuous on $ \cX \times \cC$. Accordingly, by Arzela-Ascoli theorem \citep[Theorem 11.28]{rudin1987real}, there exists a compact $\tilde\cC'' \subseteq C( \cX \times \cC)$ for which 
$$
\P(\sD_x \tilde p_{n} \in \tilde \cC'') = 1
$$
holds for all sufficiently large $n$, which implies tightness of the sequence. Take any subsequence $n_k$. Applying Lemma \ref{lemma:prokhorov}, take a further subsequence $n_{k_\ell}$ for which
$$
\sD_x \tilde p_{ n_{k_\ell}} \text{ converges weakly to }\varphi ~~~~~~\text{and}~~~~~~
\sup_{\mz \in  \cX \times \cC }|\tilde p_{n_{k_{\ell}}}(\mz) - p_Z(\mz)| \rightarrow 0 \text{ almost surely}
$$
for some $C( \cX \times \cC)$-valued random variable $\varphi$. 

Our next goal is to identify the limit $\varphi$ as the (deterministic) function $\sD_x p_Z$.  Fix a $\mz = (x,y) \in  \cX \times \cC$, and let $h > 0$ so $x + h \in \cX$. The fundamental theorem of calculus ensures
$$
\tilde p_{n_{k_\ell}}(x + h, y) - \tilde p_{n_{k_\ell}}(x,y) = \int_{x}^{x+h} \sD_x \tilde p_{n_{k_\ell}} (t,y) \, \d t,
$$
while implies
$$
\frac{\tilde p_{n_{k_\ell}}(x + h,y) - \tilde p_{n_{k_\ell}}(x,y)}{h} = \frac{1}{h}\int_{x}^{x+h} \sD_x\tilde p_{n_{k_\ell}} (t,y) \, \d t.
$$
The left-hand side satisfies
$$
\lim_{\ell \rightarrow \infty}\frac{\tilde p_{n_{k_\ell}}(x + h,y) - \tilde p_{n_{k_\ell}}(x,y)}{h} = \frac{p_Z(x+h,y) - p_Z(x,y)}{h} 
$$
almost surely, using uniform convergence almost surely by our choice of $n_{k_\ell}$. Since the function
$$
g \mapsto \frac{1}{h}\int_x^{x+h} g(t, y) \, \d t
$$
is a continuous map on $C( \cX \times \cC)$, the continuous mapping theorem ensures that
$$
\frac{1}{h}\int_x^{x+h} \sD_x \tilde p_{n_{k_\ell}}(t, y) \, \d t \text{ converges weakly to }\frac{1}{h} \int_x^{x+h} \varphi(t,y) \, \d t
$$
as $\ell \rightarrow \infty$. 

Thus, for any fixed $y \in \cC$ and any fixed and sufficiently small $h > 0$,
$$
\frac{1}{h}\int_x^{x+h} \varphi(t,y) \, \d t = \frac{p_Z(x+h,y) - p_Z(x,y)}{h} = \frac{1}{h}\int_x^{x+h} \sD_xp_Z(t, y) \, \d t
$$
almost surely by the uniqueness of limits. Using the continuity of $\varphi$ and $\sD_xp_Z$, we obtain $\varphi(\mz) = \sD_xp_Z(\mz)$ for all $\mz \in \cX \times \cC$.
Since $\sD_x p_Z$ is deterministic, we apply Lemma \ref{lemma:weak_to_prob}, yielding
$$
\sup_{\mz \in  \cX \times \cC}\Big|\sD_x \tilde p_{ n_{k_\ell}}(\mz) - \sD_xp_Z(\mz)\Big| = o_{\P}(1).
$$
Lemma \ref{lemma:prob_and_as} then implies the existence of a further subsequence $n_{k_{\ell_m}}$ with 
$$
\sup_{\mz  \in  \cX \times \cC}\Big|\sD_x\tilde p_{ n_{k_{\ell_m}}}(\mz) - \sD_x p_Z(\mz)\Big| \rightarrow 0
$$
almost surely. The exact same argument, taking $\sD_y$ instead, yields the full conclusion.

~\\
{\bf Step 3.} Assumption \ref{assump_bootstrap}(c) and \ref{assump_iso}(e) implies the existence of a closed, bounded  interval $\tilde \cY$ for which $\cY \cup \tilde \cY_n \subseteq \tilde \cY$ for all $n$, almost surely. {\bf Step 1} and {\bf Step 2} combined then yield
$$
\sup_{\mz \in  \cX \times \tilde \cY}\Big|\tilde p_n(\mz) - p_Z(\mz)\Big| = o_{\P}(1) ~~~~~~\text{and}~~~~~~ \sup_{\mz \in  \cX \times \tilde\cY}\Big\|\sD\tilde p_n(\mz) - \sD p_Z(\mz)\Big\|_2 = o_{\P}(1).
$$
The above is used to establish the remaining claims.

~\\
{\bf Step 4.} Next, we show $\sup_{x \in  \cX}|\tilde p_X(x) - p_X(x)| = o_{\P}(1)$. 
Expanding,
\al{
\sup_{x \in  \cX}\Big|\tilde p_X(x) - p_X(x)\Big| 
&= \sup_{x \in  \cX}\ap{\int_{\tilde \cY} \tilde p_n(x, y)  -  p_Z(x,y) \, \d y}
\tag{definition of $\tilde \cY$}\\
&\leq  \sup_{\mz \in  \cX \times \tilde \cY}\ap{\tilde p_n(\mz)  -  p_Z(\mz)} \cdot\int_{\tilde \cY} \, \d y\\
&= o_{\P}(1),
}
where the last equality comes from the fact that $\tilde \cY$ is bounded, so its Lebesgue measure is finite.

~\\
{\bf Step 5.}  Next,  we show $\sup_{x\in  \cX} |\tilde f_0(x) - f_0(x)| = o_{\P}(1)$. Expanding the regression functions,
\al{
\sup_{x \in  \cX}|\tilde f_0(x) - f_0(x)| &=\sup_{x \in  \cX} \ap{ \int_{\R} y \cdot \pa{\frac{\tilde p_n(x,y)}{\tilde p_X(x)} - \frac{p_Z(x,y)}{p_X(x)}} \, \d y}\\
&=\sup_{x \in  \cX} \ap{ \int_{\tilde \cY} y \cdot \pa{\frac{\tilde p_n(x,y)}{\tilde p_X(x)} - \frac{p_Z(x,y)}{p_X(x)}} \, \d y} \\
&\leq  \sup_{(x',y') \in  \cX \times \tilde \cY}\ap{\frac{\tilde p_n(x',y')}{\tilde p_X(x')} - \frac{p_Z(x',y')}{p_X(x')}} \cdot \int_{\tilde\cY} |y| \, \d y \\
&= o_{\P}(1),
}
almost surely. 
Here the last equality comes from the fact that $\tilde \cY$ is bounded, and $p_X$ is uniformly bounded below on $ \cX$ so that the continuous mapping theorem applies.

By the analogous argument, by replacing $p_X$ and $f_0$ with $p_X'$ and $\tilde f_0'$, respectively, we have
$$
\sup_{x\in  \cX} |\tilde p_X'(x) - p_X'(x)| = o_{\P}(1)~~~{\rm and}~~~
\sup_{x \in  \cX}|\tilde f_0'(x) - f_0'(x)| = o_{\P}(1).
$$
This completes the whole proof.
\end{proof}

Next, we calculate the bias of the local average of $\tilde f_0$.
\begin{lemma}\label{lemma:biascalc}
    Suppose that Assumptions \ref{assump_noise1}, \ref{assump_bootstrap}, \ref{assump_iso}, and \ref{ass:generator} hold. Then, for any $k,H > 0$,
$$
\sup_{\substack{(\ell,u) \in (0,H] \times [0,H]\\ \ell + u \geq k}} \ap{n^{1/3}(\overline{\tilde {f}}_{[\ell_n, u_n]} - \tilde f_0(x_0)) - \frac{f_0'(x_0)}{2}(u-\ell) } = o_{\P}(1).
$$
\end{lemma}
\begin{proof}
We take $n$ large enough so that $[x_0 - H n^{-1/3}, x_0 +H n^{-1/3}] \subseteq \cX$, which is possible since $x_0$ is set to be an interior point of $\cX$. Since $u,\ell \leq H$, we have $[x_0 - \ell n^{-1/3}, x_0 + u n^{-1/3}] \subseteq \cX$ for all $0 \leq u, \ell \leq H$ simultaneously. Without loss of generality, we assume $n$ is large enough that this holds for the remainder of the proof. As a result,
    $$
    \Big|\Big\{i : \tilde X_i \in [\ell_n, u_n] \cap \cX\Big\}\Big| = \Big|\Big\{i : \tilde X_i \in [\ell_n, u_n]\Big\}\Big| \quad \text{and} \quad \mathds{1}\Big(\tilde X_i \in [\ell_n, u_n] \cap \cX\Big) = \mathds{1}\Big(\tilde X_i \in [\ell_n, u_n]\Big).
    $$
        Rewriting the scaled bias,
    \begin{equation}
    n^{1/3}(\overline{\tilde f}_{[\ell_n,u_n]} - \tilde f_0(x_0)) = \frac{n \cdot n^{-1/3}}{|\{\tilde X_i : \ell_n \leq \tilde X_i \leq u_n\}|} \cdot \frac{1}{n^{1/3}}\sum_{i=1}^n ( \tilde f_0(\tilde X_i) - \tilde f_0(x_0)) \cdot \mathds{1}(\ell_n \leq \tilde X_i \leq u_n). \label{eq:biascalc:scaledbias}
    \end{equation}

    A routine application of Lemma \ref{lemma:wellvdv229} and Markov's inequality applied to the functions
    $$
    \Big\{f_{u,\ell}: x \mapsto n^{1/2}(\tilde f_0(x) - \tilde f_0(x_0)) \cdot\mathds{1}(-\ell n^{-1/3} \leq x - x_0\leq u n^{-1/3}) : (\ell,u) \in (0,H] \times [0,H]\Big\}
    $$
    demonstrates
    \begin{equation}
    \sup_{(\ell,u) \in (0,H] \times [0,H]}  \ap{\frac{1}{\sqrt{n}}\sum_{i=1}^n f_{u,\ell}(\tilde X_i) - \E[f_{u,\ell}(\tilde X_1) \mid \cO]} = O_{\P}(1) .\label{eq:biascalc:donsker}
    \end{equation}     
By similar reasoning applied to indicator functions, we also obtain
    $$
    \sup_{(\ell,u) \in (0,H] \times [0,H]}  \ap{\frac{|\{\tilde X_i : \ell_n \leq \tilde X_i \leq u_n\}|}{ n \cdot n^{-1/3}} -   \frac{\P(\ell_n \leq \tilde X_i \leq u_n \mid \cO)}{n^{-1/3} } }= o_{\P}(1).
    $$
    Next, for sufficiently large $n$, the mean value theorem combined with Lemma \ref{lemma:margdiff}  ensures
    \al{
    \sup_{(\ell,u) \in (0,H] \times [0,H]}  \ap{ \frac{\P(\ell_n \leq \tilde X_i \leq u_n \mid \cO) }{ n^{-1/3}} - (u + \ell) \cdot \tilde p_{X}(x_0)} \leq 2H^2\cdot\sup_{t \in \cX} |\tilde p_X'(t)| \cdot n^{-1/3}.
    }

By the first conclusion of Lemma \ref{lemma:isocvg}, we have 
\[ 
    \sup_{(\ell,u) \in (0,H] \times [0,H]}\Big|(u+\ell) \cdot (\tilde p_X(x_0) - p_X(x_0))\Big| = o_{\P}(1),
    \]
    so that the continuous mapping theorem, combined with the fact that $(u+\ell)\cdot p_X(x_0)$ is bounded uniformly below from Assumption \ref{assump_iso}(d), yields
    \begin{equation}
    \sup_{\substack{(\ell,u) \in (0,H] \times [0,H]\\ \ell + u \geq k}} \ap{\frac{n \cdot n^{-1/3}}{ |\{\tilde X_i : \ell_n \leq \tilde X_i \leq u_n\}|} - \frac{1}{(u+\ell) \cdot p_X(x_0)}} = o_{\P}(1). \label{eq:biascalc:margdensitycvg}
    \end{equation}
    
    Applying \eqref{eq:biascalc:donsker} and \eqref{eq:biascalc:margdensitycvg} to  \eqref{eq:biascalc:scaledbias} then yields
    $$
    \sup_{\substack{(\ell,u) \in (0,H] \times [0,H]\\ \ell + u \geq k}}  \ap{n^{1/3}(\overline{\tilde f}_{[\ell_n,u_n]} - \tilde f_0(x_0)) - \frac{n^{2/3}}{(u+\ell) \cdot p_X(x_0)} \cdot  \ep{ (\tilde f_0(\tilde X_1) - \tilde f_0(x_0)) \cdot \mathds{I}(\ell_n \leq \tilde X_1 \leq u_n) \mid \mathcal{O}}}  = o_{\P}(1).
    $$
   
    Next, we aim to show
    \begin{equation}\label{eq:biascalc:lefthandside}
    \sup_{\substack{(\ell,u) \in (0,H] \times [0,H]\\ \ell + u \geq k}}  \ap{\frac{n^{2/3}}{(u+\ell) \cdot p_X(x_0)} \cdot  \pa{\ep{ \pa{\tilde f_0(\tilde X_1) - \tilde f_0(x_0) - \tilde f_0'(x_0) \cdot (\tilde X_1 - x_0)} \cdot \mathds{I}(\ell_n \leq \tilde X_1 \leq u_n) \mid \mathcal{O}}}} = o_{\P}(1). 
    \end{equation}
    
    Since $n$ is large enough so that $[x_0 - \ell n^{-1/3}, x_0 + un^{-1/3}] \subseteq  \cX$, for all $(\ell,u) \in (0,H] \times [0,H]$, $\tilde f_0$ is twice continuously differentiable on $ \cX$,  by Lemma \ref{lemma:tildef0_twicediff}. Hence, Taylor's theorem with remainder yields an upper bound of \eqref{eq:biascalc:lefthandside} 
    \begin{equation}\label{eq:biascalc:innerintegral}
        \sup_{\substack{(\ell,u) \in (0,H] \times [0,H]\\ \ell + u \geq k}} \ap{\frac{n^{2/3}}{(u+\ell) \cdot p_X(x_0)} \cdot \sup_{x \in  \cX}\frac{\tilde f_0''(x)}{2}\cdot  \ep{ (\tilde X_1 - x_0)^2 \cdot \mathds{I}(\ell_n \leq \tilde X_1 \leq u_n) \mid \mathcal{O}}}.
    \end{equation}
    Evaluating the expectation in \eqref{eq:biascalc:innerintegral}, we deduce
    \al{
    \ep{ (\tilde X_1 - x_0)^2 \cdot \mathds{I}(\ell_n \leq \tilde X_1 \leq u_n) \mid \mathcal{O}} &= \int_{\ell_n}^{u_n} (x - x_0)^2 \tilde p_X(x) \, \d x \tag{definition of expectation}\\
    &\leq \sup_{t \in \cX} |\tilde p_X(t)| \int_{\ell_n}^{u_n} (x-x_0)^2 \, \d x \\
    &= \frac{\sup_{t \in \cX} |\tilde p_X(t)|\cdot(u^3 - \ell^3)}{3n}.
    }
    Thus, we obtain
    $$
   \sup_{\substack{(\ell,u) \in (0,H] \times [0,H]\\ \ell + u \geq k}} \ap{\frac{n^{2/3}}{(u+\ell) \cdot p_X(x_0)} \cdot \sup_{x \in  \cX}\frac{\tilde f_0''(x)}{2} \cdot \frac{\sup_{t \in \cX} |\tilde p_X(t)|\cdot(u^3 - \ell^3)}{3n}}  = o_{\P}(1),
    $$
so that
    \begin{equation}\label{eq:biascalc:theintegral}
\sup_{\substack{(\ell,u) \in (0,H] \times [0,H]\\ \ell + u \geq k}}\ap{n^{1/3}(\overline{\tilde f}_{[\ell_n,u_n]} - \tilde f_0(x_0)) - \frac{n^{2/3}\tilde f_0'(x_0) }{(u+\ell) \cdot p_X(x_0)} \cdot \E\Bigg[ (\tilde X_1 - x_0) \cdot \mathds{1}(\ell_n \leq \tilde X_1 \leq u_n) \mid \cO\Bigg]} = o_{\P}(1)
    \end{equation}
    Since Lemma \ref{lemma:isocvg} implies that $\sup_{x\in  \cX} |\tilde p_X(x) - p_X(x)| = o_{\P}(1)$, and $\sup_{x\in  \cX} |\tilde f_0'(x) - f_0'(x)| = o_{\P}(1)$, calculating the integral in \eqref{eq:biascalc:theintegral} demonstrates
    \begin{equation}\label{eq:biascalc:theexpectation}
        \sup_{\substack{(\ell,u) \in (0,H] \times [0,H]\\ \ell + u \geq k}}\ap{n^{1/3}(\overline{\tilde f}_{[\ell_n,u_n]} - \tilde f_0(x_0)) - \frac{n^{2/3}f_0'(x_0) }{(u+\ell) \cdot p_X(x_0)} \cdot \E\Bigg[ (X_1 - x_0) \cdot \mathds{1}(\ell_n \leq X_1 \leq u_n) \Bigg]} = o_{\P}(1).
    \end{equation}
    Finally, evaluating the expectation in \eqref{eq:biascalc:theexpectation}, we obtain
    $$
\sup_{\substack{(\ell,u) \in (0,H] \times [0,H]\\ \ell + u \geq k}} \ap{n^{1/3}(\overline{\tilde {f}}_{[\ell_n, u_n]} - \tilde f_0(x_0)) - \frac{f_0'(x_0)}{2}(u-\ell) } = o_{\P}(1),
$$
which finishes the proof. 
\end{proof}
Next, we note a consequence of the central limit theorem and Lemma \ref{lemma:biascalc}.

\begin{lemma}\label{lemma:fidi_uclt}
Suppose that Assumptions \ref{assump_noise1}, \ref{assump_bootstrap}, \ref{assump_iso}, and \ref{ass:generator} hold.
Then, for any $k, H >0$,
$$
    \sup_{t \in \R} \sup_{\substack{(\ell,u) \in (0,H] \times [0,H]\\ \ell + u \geq k} }\ap{\pr{n^{1/3}(\overline {\tilde Y}_{[\ell_n, u_n]} - \tilde f_0(x_0)) \leq t \cdl \mathcal{O}} - \pr{G_{\ell, u}\leq t)}} = o_{\P}(1).
    $$
\end{lemma}
\begin{proof}
 Rewrite the local average as
    \begin{equation}
    n^{1/3}\Big(\overline {\tilde Y}_{[\ell_n, u_n]} - \tilde f_0(x_0)\Big) = n^{1/3}\Big(\overline {\tilde \xi}_{[\ell_n, u_n]} + \overline{\tilde f}_{[\ell_n,u_n]}- \tilde f_0(x_0)\Big).\label{eq;fidi:firsteq}
\end{equation}
    We assume without loss of generality that $n$ is large enough so that $[x_0 - H n^{-1/3}, x_0 + Hn^{-1/3}] \subseteq \cX$.     
    
~\\    
{\bf Step 1.}  For \eqref{eq;fidi:firsteq}, we first obtain
    \begin{align}
   n^{1/3}\overline {\tilde \xi}_{[\ell_n, u_n]}= &\frac{n^{1/3} \cdot \sqrt{n \cdot n^{-1/3}}} {|\{i: \ell_n \leq \tilde X_i\leq u_n\}|} \cdot \frac{\sum_{i=1}^n  \{\tilde \xi_{i} \cdot\mathds{1}(\ell_n \leq \tilde X_i\leq u_n) - \E[\tilde\xi_{1} \cdot\mathds{1}(\ell_n \leq \tilde X_1\leq u_n) \mid \cO]\}}{\sqrt{n \cdot n^{-1/3}}} \nonumber\\
    &+ \frac{n^{1/3}}{|\{i: \ell_n \leq \tilde X_i\leq u_n\}|} \cdot \E[\tilde\xi_{1} \cdot\mathds{1}(\ell_n \leq \tilde X_1\leq u_n) \mid \cO]. \label{eq:fidi:bigterm}
    \end{align}
 Supposing $[\ell_n, u_n] \subseteq  \cX$, it holds true that
\begin{align}
\Big|\tilde\xi_{1} \cdot \mathds{1}(\ell_{n} \leq \tilde X_{1}\leq u_{n})\Big| &= \Big|(\tilde Y_{1} - \tilde f_{0}(\tilde X_{1})) \cdot \mathds{1}(\ell_{n} \leq \tilde X_{1}\leq u_{n})\Big| \nonumber \tag{definition of $\tilde \xi_{1}$}\\
&\leq (|\tilde Y_{1}| + |\tilde f_{0}(\tilde X_{1})|) \cdot \mathds{1}(\ell_{n} \leq \tilde X_{1}\leq u_{n}) \nonumber \tag{triangle inequality}\\
&\leq 2\cdot \sup_{\mz \in \tilde \cZ} \|\mz\|_2 \cdot \mathds{1}(\ell_{n} \leq \tilde X_{1}\leq u_{n}), \label{eq:fidi:xibound}
\end{align}
where the final bound comes from the fact that $\tilde f_0$ is defined as a conditional expectation of $\tilde Y_1$ and Assumption \ref{assump_bootstrap}(c). Equation \eqref{eq:biascalc:margdensitycvg} then ensures
$$
\sup_{\substack{(\ell,u) \in (0,H] \times [0,H]\\ \ell + u \geq k}}\ap{\frac{n^{1/3}}{|\{i: \ell_n \leq \tilde X_i\leq u_n\}|} \cdot \E[\tilde\xi_{1} \cdot\mathds{1}(\ell_n \leq \tilde X_1\leq u_n) \mid \cO] }= o_{\P}(1).
$$

~\\
\textbf{Step 2: Reduce to a subsequence.}

Assumptions \ref{assump_bootstrap}(c) and \ref{assump_iso}(e) ensure that there exists a closed, bounded  interval $\tilde \cY$ for which $\cY \cup \tilde \cY_n \subseteq \tilde \cY$ for all $n$, almost surely.  We then appeal to Lemma \ref{lemma:prob_and_as}. Take any subsequence $n_k$. 
 Lemma \ref{lemma:isocvg} and Assumption \ref{assump_bootstrap}(b) ensure the existence of a further subsequence $n_{k_m}$ such that
$$
\sup_{x \in  \cX}\Big|\tilde f_{0,n_{k_m}}(x) - f_0(x)\Big|\rightarrow 0, \sup_{x \in  \cX \times \tilde\cY}\Big|\tilde p_{n_{k_m}}(\mz) - p_Z(\mz)\Big|\rightarrow 0,  \text{ and }  \sW_1(\bP_{\tilde Z_{1, n_{k_m}} \mid \cO} , \P_{Z}) \rightarrow 0
$$
almost surely. Furthermore, by \eqref{eq:biascalc:margdensitycvg} and Lemma \ref{lemma:biascalc},  suppose that this subsequence also satisfies
$$
\sup_{\substack{(\ell,u) \in (0,H] \times [0,H]\\ \ell + u \geq k}}\ap{\frac{n_{k_m} \cdot n_{k_m}^{-1/3}}{ |\{i : \ell_{n_{k_m}} \leq \tilde X_{i,n_{k_m}} \leq u_{n_{k_m}}\}|} - \frac{1}{(u+\ell)\cdot p_X(x_0)}} \rightarrow 0
$$
and 
$$
\sup_{\substack{(\ell,u) \in (0,H] \times [0,H]\\ \ell + u \geq k}}\ap{n^{1/3}(\overline{\tilde {f}}_{[\ell_{n_{k_m}}, u_{n_{k_m}}]} - \tilde f_0(x_0)) - \frac{f_0'(x_0)}{2}(u-\ell) } \rightarrow 0 
$$
almost surely.

~\\
\textbf{Step 3: Check that the ratio of variances uniformly converges to $1$.}

This step aims to show that
$$
\sup_{(\ell,u) \in (0,H] \times [0,H]}\ap{\frac{{\textrm{Var}}(\tilde \xi_{1,n_{k_m}} \cdot\mathds{1}(\ell_{n_{k_m}} \leq \tilde X_{1,n_{k_m}}\leq u_{n_{k_m}}) \mid \cO)}{{\textrm{Var}}( \xi_{1} \cdot\mathds{1}(\ell_{n_{k_m}} \leq  X_{1}\leq u_{n_{k_m}})) }-1} \rightarrow 0 
$$
almost surely. Equivalently, we show 
\begin{equation}\label{eq:uclt:theratio}
\sup_{(\ell,u) \in (0,H] \times [0,H]}\ap{\frac{{\textrm{Var}}(\tilde \xi_{1,n_{k_m}} \cdot\mathds{1}(\ell_{n_{k_m}} \leq \tilde X_{1,n_{k_m}}\leq u_{n_{k_m}}) \mid \cO) - {\textrm{Var}}( \xi_{1} \cdot\mathds{1}(\ell_{n_{k_m}} \leq  X_{1}\leq u_{n_{k_m}})) }{{\textrm{Var}}( \xi_{1} \cdot\mathds{1}(\ell_{n_{k_m}} \leq  X_{1}\leq u_{n_{k_m}})) } } \rightarrow 0
\end{equation}
almost surely. Rewriting the denominator of \eqref{eq:uclt:theratio},
\begin{align}
{\textrm{Var}}( \xi_{1} \cdot\mathds{1}(\ell_{n_{k_m}} \leq  X_{1}\leq u_{n_{k_m}}))  &= \E[\xi_1^2 \cdot \mathds{1}(\ell_{n_{k_m}} \leq X_1 \leq u_{n_{k_m}})^2] - (\E[\xi_1 \cdot \mathds{1}(\ell_{n_{k_m}} \leq X_1 \leq u_{n_{k_m}})])^2 \notag\\
&= \E[\xi_1^2 \cdot \mathds{1}(\ell_{n_{k_m}} \leq X_1 \leq u_{n_{k_m}})^2] \nonumber\\ 
&=\E[\xi_1^2 \cdot \mathds{1}(\ell_{n_{k_m}} \leq X_1 \leq u_{n_{k_m}})] \notag\\ 
&= \sigma^2 \cdot \P(\ell_{n_{k_m}} \leq X_1 \leq u_{n_{k_m}})\notag\\ 
&= \sigma^2 \int_{\ell_{n_{k_m}}}^{u_{n_{k_m}}} p_X(t) \, \d t \nonumber\\ 
&\geq \sigma^2 \cdot \inf_{t \in \cX} p_X(t) \cdot (u + \ell) \cdot n_{k_m}^{-1/3}.\label{eq:fidi:ratio_denom}
\end{align}

Next, the numerator of \eqref{eq:uclt:theratio} is upper bounded by
\al{
\Big|\E\Big[\tilde \xi_{1,n_{k_m}}^2 \cdot \mathds{1}(\ell_{n_{k_m}} \leq \tilde X_{1,n_{k_m}}\leq u_{n_{k_m}}) - \xi_1^2 \cdot \mathds{1}(\ell_{n_{k_m}} \leq X_{1}\leq u_{n_{k_m}}) \mid \cO\Big]\Big| \\
+  \Big|\Big(\E\Big[\tilde \xi_{1,n_{k_m}} \cdot \mathds{1}(\ell_{n_{k_m}} \leq \tilde X_{1,n_{k_m}}\leq u_{n_{k_m}}) \mid \cO\Big]\Big)^2\Big|.
}
Similar arguments to \eqref{eq:fidi:xibound} give
\begin{equation}\label{eq:uclt:firstequality}
\Big|\Big(\E\Big[\tilde \xi_{1,n_{k_m}} \cdot \mathds{1}(\ell_{n_{k_m}} \leq \tilde X_{1,n_{k_m}}\leq u_{n_{k_m}}) \mid \cO\Big]\Big)^2\Big| = (u + \ell)^2 \cdot  O_{\P}(n_{k_m}^{-2/3})
\end{equation}
 and, by Lemma \ref{lemma:isocvg}, we have
\begin{equation}\label{eq:uclt:secondequality}
\Big|\E\Big[\tilde \xi_{1,n_{k_m}}^2 \cdot \mathds{1}(\ell_{n_{k_m}} \leq \tilde X_{1,n_{k_m}}\leq u_{n_{k_m}}) - \xi_1^2 \cdot \mathds{1}(\ell_{n_{k_m}} \leq X_{1}\leq u_{n_{k_m}}) \mid \cO\Big]\Big| = (u + \ell) \cdot o_{\P}(n_{k_m}^{-1/3}).
\end{equation}
Combining \eqref{eq:fidi:ratio_denom}, \eqref{eq:uclt:firstequality}, and \eqref{eq:uclt:secondequality} in \eqref{eq:uclt:theratio} then proves the claim. 

~\\
\textbf{Step 4: Obtain a uniform central limit theorem using Lemma \ref{lemma:uclt}, and conclude the proof.}

Consider the following function class
$$
\Big\{h_{\ell,u}(x,\xi) = \xi \cdot \mathds{1}(-\ell  \leq x  \leq u):0 \leq  \ell, u \leq K, \ell + u \geq k \Big\}
$$
coupled with the triangular array $
\{(n^{1/3}(\tilde X_{i,n} - x_0) ,  n^{1/6}\tilde\xi_{i,n})\}_{i \in [n]} $ and envelope function $(x, \xi) \mapsto \xi \cdot \mathds{1}(-K \leq x \leq K)$. Lemma \ref{lemma:uclt} combined with {\bf Step 3} demonstrates that the following term that comes from \eqref{eq:fidi:bigterm}, 
$$
\frac{1}{\sqrt{{n_{k_m}}}}\sum_{i=1}^{n_{k_m}}  {n_{k_m}^{1/6}} \cdot \Big(\tilde \xi_{i,n_{k_m}} \cdot\mathds{1}(\ell_{n_{k_m}} \leq \tilde X_{i,n_{k_m}}\leq u_{n_{k_m}}) - \E\Big[\tilde\xi_{1,n_{k_m}} \cdot\mathds{1}(\ell_{n_{k_m}} \leq \tilde X_{1,n_{k_m}}\leq u_{n_{k_m}}) \mid \cO\Big]\Big),
$$
converges weakly to some mean-zero Gaussian process uniformly in $u, \ell \in (0,H] \times [0,H]$ with $\ell + u \geq k$ with variance function $(\ell, u) \mapsto (\ell + u) \cdot \sigma^2 \cdot p_X(x_0)$,  conditionally on $\cO$. 


Note that, by our choice of the subsequence, it holds true that
$$
 \sup_{\substack{(\ell,u) \in (0,H] \times [0,H]\\ \ell + u \geq k}}\ap{\frac{n_{k_m} \cdot n_{k_m}^{-1/3}}{ |\{i : \ell_{n_{k_m}} \leq \tilde X_{i,n_{k_m}} \leq u_{n_{k_m}}\}|} - \frac{1}{(u+\ell) \cdot p_X(x_0)}} \rightarrow 0
$$
almost surely.  
We then obtain
$$
{n_{k_m}}^{1/3}(\overline {\tilde Y}_{[\ell_{n_{k_m}}, u_{n_{k_m}}]} - \tilde f_{0,n_{k_m}}(x_0)) \text{ converges weakly}\text{ to } G_{\ell, u} \text{ conditional on }\cO
$$
uniformly in $(\ell,u) \in (0,H] \times [0,H]$, with $\ell + u \geq k$, which further implies 
$$
    \sup_{t \in \R} \sup_{\substack{(\ell,u) \in (0,H] \times [0,H]\\ \ell + u \geq k}} \ap{\pr{n_{k_m}^{1/3}(\overline {\tilde Y}_{[\ell_{n_{k_m}}, u_{n_{k_m}}]} - \tilde f_{0,n_{k_m}}(x_0)) \leq t \cdl \mathcal{O}} - \pr{G_{\ell, u}\leq t)}} \rightarrow 0
$$
almost surely. Since $n_{k}$ was an arbitrary subsequence, we conclude that 
$$
    \sup_{t \in \R}\sup_{\substack{(\ell,u) \in (0,H] \times [0,H]\\ \ell + u \geq k}} \ap{\pr{n^{1/3}(\overline {\tilde Y}_{[\ell_{n}, u_{n}]} - \tilde f_{0}(x_0)) \leq t \cdl \mathcal{O}} - \pr{G_{\ell, u}\leq t)}} = o_{\P}(1),
$$
and thus finish the proof.
\end{proof}

\begin{lemma}\label{lemma:ul_o1}
Suppose that Assumptions \ref{assump_noise1}, \ref{assump_bootstrap}, \ref{assump_iso}, and \ref{ass:generator} hold.  Then, there exist some random variables $(\tilde L^*, \tilde U^*)$ such that
$$
\P\pa{ \tilde f_n(x_0) = \overline{\tilde Y}_{[\tilde L^*_n, \tilde U^*_n]}} \rightarrow 1
$$
as $n\rightarrow \infty$, where $\tilde L^*_n = x_0 - \tilde L^* n^{-1/3}$ and $\tilde U^*_n = x_0 + \tilde U^* n^{-1/3}$. Furthermore, these random variables satisfy 
\begin{enumerate}
    \item $\tilde L^*_n, \tilde U^*_n \in \cX$ always;
    \item  $\P(|\{i \in [n] : \tilde X_i \in [\tilde L_n^*, \tilde U_n^*]\}| \neq 0) \rightarrow 1$;
    \item $\tilde U^* = \Theta_\P(1)$ and  $\tilde L^* = \Theta_\P(1)$
\end{enumerate} 
\end{lemma}
\begin{proof}
\textbf{Step 1: Apply the max-min formula to obtain the representation of $\tilde f_n(x_0)$ and choose $\tilde L_n^*, \tilde U_n^*$ that satisfy the first two conditions.}

By Assumption \ref{assump_iso}(e), $\cX$ takes the form $\cX = [a^*, b^*]$ for some $a^*, b^* \in \R$. The probability of the event
$$
\cp{\tilde f_n(x_0) = \max_{a \in [n] :  \tilde X_a \leq x_0} \min_{b \in [n] : \tilde X_b\geq x_0} \overline{\tilde Y}_{[\tilde X_a, \tilde X_b]}}
$$
tends to $1$ by Lemma \ref{lemma:maxmin}.  A reparameterization yields 
$$
\max_{\ell > 0 :a^* \leq  \ell_n \leq \max_{i \in [n]: \tilde X_i \leq x_0} \tilde X_i} \min_{u \geq 0 :  b^* \geq u_n\geq \min_{i \in [n]: \tilde X_i \geq x_0} \tilde X_i} \overline{\tilde{Y}}_{[\ell_n, u_n]}
$$
always, for $\ell_n = x_0 - \ell n^{-1/3}$ and $u_n = x_0 + un^{-1/3}$ .  Indeed, when the set $\{i : \tilde X_i \in \cX\}$ is empty, both are defined by convention to be $0$.

When $\{i \in [n] : \tilde X_i \in \cX\}$ is nonempty, we then choose $\tilde L^*$ and $\tilde U^*$ to satisfy
$$
\overline{\tilde Y}_{[\tilde L_n^*, \tilde U_n^*]}  = \max_{\ell > 0 : a^* \leq \ell_n \leq \max_{i \in [n]: \tilde X_i \leq x_0} \tilde X_i} \min_{u \geq 0 :  b^* \geq u_n\geq \min_{i \in [n]: \tilde X_i \geq x_0} \tilde X_i} \overline{\tilde{Y}}_{[\ell_n, u_n]}
$$ 
for 
\[
\tilde{L}_n^* = x_0 - \tilde L^* n^{-1/3}~~~ {\rm and}~~~ \tilde{U}_n^* = x_0 + \tilde U^* n^{-1/3}. 
\]
When $\{i \in [n] : \tilde X_i \in \cX\}$ is empty, we define $\tilde L^*$ and $\tilde U^*$ to be arbitrary constants such that $\tilde L_n^*, \tilde U_n^* \in \cX$. 

By construction, the first and second conditions in Lemma \ref{lemma:ul_o1} are then automatically satisfied.

~\\
\textbf{Step 2: Show that $n^{1/3}(\tilde f_n(x_0) - \tilde f_0(x_0)) = O_{\P}(1)$.}
    
We have
    \begin{align}
        n^{1/3}(\tilde f_n(x_0) - \tilde f_0(x_0))
        &= n^{1/3}(\overline{\tilde Y}_{[\tilde L_n^*, \tilde U_n^*]} - \tilde f_0(x_0))\nonumber \\
        &\leq  n^{1/3}(\overline{\tilde Y}_{[\tilde L_n^*, x_0 + n^{-1/3}]} - \tilde f_0(x_0)) \nonumber\\
        &= n^{1/3}\Big(\overline{\tilde \xi}_{[\tilde L_n^*,x_0+ n^{-1/3}]} +  \overline{\tilde f}_{[\tilde L_n^*,x_0+ n^{-1/3}]}- \tilde{f}_0(x_0)\Big)\nonumber \\
        &\leq  n^{1/3}\Big(\sup_{\ell \geq 0}|\overline{\tilde \xi}_{[x_0 - \ell n^{-1/3},x_0+ n^{-1/3}]}| +  \overline{\tilde f}_{[\tilde L_n^*,x_0+ n^{-1/3}]}- \tilde{f}_0(x_0)\Big) \label{eq:ul_o1:upper}\notag\\
        &= O_{\P}(1) + n^{1/3}(\overline{\tilde f}_{[\tilde L_n^*,x_0+ n^{-1/3}]}- \tilde{f}_0(x_0)). \nonumber\tag{Lemma \ref{lemma:xi_o1}}
    \end{align}
    Assumption \ref{assump_iso}(a) ensures that  $\inf_{x \in \cX} f'_0(x) > 0$ and Lemma \ref{lemma:isocvg} ensures that $\sup_{x \in \cX} |\tilde f'_{0}(x) - f'_0(x)| =o_\P(1)$. For large enough $n$, $\tilde f_0$ is continuously differentiable almost surely on $\cX$ by Lemma \ref{lemma:tildef0_twicediff}. Also, $\cX$ is a closed interval, so we obtain
    $$
    \P\Big(\tilde f_0 \text{ is strictly increasing on }\cX\Big) \rightarrow 1
    $$
    as $n\rightarrow \infty$.
    Next, since $\tilde L^*_n \in \cX$,
     \al{
    n^{1/3}(\overline{\tilde f}_{[\tilde L_n^*, x_0 + n^{-1/3}]} - \tilde{f}_0(x_0)) &\leq n^{1/3}(\overline{\tilde f}_{[x_0, x_0 + n^{-1/3}]} -\tilde{f}_0(x_0))\\
    &= O_{\P}(1). \tag{Lemma \ref{lemma:biascalc}}
    }

A lower bound follows analogously, and we thus prove the claim.

~\\
\textbf{Step 3: Show that $\tilde U^*, \tilde L^* = O_{\P}(1).$}

Take any $\epsilon > 0$. Denote the events
$$
\Omega_{\epsilon} := \cp{|n^{1/3}(\tilde f_n(x_0) - \tilde f_0(x_0))| > f_0'(x_0) \cdot (H_\epsilon - 1)/8}
$$
and
$$
\check\Omega_\epsilon := \cp{-n^{1/3} \sup_{u \geq 0} |\overline{\tilde \xi}_{[x_0 -  n^{-1/3}, x_0 + un^{-1/3}]}| < -f_0'(x_0) \cdot (H_\epsilon - 1)/8},
$$
where $H_\epsilon > 0$ is chosen so that
$$
\limsup_{n\rightarrow \infty }\Big\{\P(\Omega_\epsilon) + \P(\check \Omega_\epsilon)\Big\}  < \epsilon.
$$
We then have
\al{
\P(\tilde U^* \geq H_\epsilon) \leq \P\Big(\tilde U^* \geq H_\epsilon, \Omega_\epsilon^c, \check{\Omega}_\epsilon^c \Big) + \P(\Omega_\epsilon) + \P(\check \Omega_\epsilon).
}
We now work on the intersection of events $\Omega_\epsilon^c\,$,$\, \check{\Omega}_\epsilon^c$, and $\tilde U^* \geq H_\epsilon$. In a calculation similar to \textbf{Step 2},
\al{
n^{1/3}\Big(-\sup_{u \geq 0}|\overline{\tilde \xi}_{[x_0- n^{-1/3},x_0 + u n^{-1/3}]}| +  \overline{\tilde f}_{[x_0 - n^{-1/3}, \tilde U_n^*]}- \tilde{f}_0(x_0)\Big)  \leq n^{1/3}(\tilde f_n(x_0) - \tilde f_0(x_0)).
}
On the events $\check{\Omega}_\epsilon^c$ and $\Omega_\epsilon^c$,
$$
-\frac{f_0'(x_0)}{8}(H_\epsilon - 1) + n^{1/3}(\overline{\tilde f}_{[x_0 - n^{-1/3}, \tilde U_n^*]}- \tilde{f}_0(x_0))  \leq \frac{f_0'(x_0)}{8}(H_\epsilon - 1).
$$
As we have assumed $\tilde f_0$ is strictly increasing in $\cX$, and $ \tilde U_n^* \in \cX$ without loss of generality, $x_0 + H_\epsilon n^{-1/3} \in \cX$ for large enough $n$. Then, the inequality 
$$
-\frac{f_0'(x_0)}{8}(H_\epsilon - 1) + n^{1/3}(\overline{\tilde f}_{[x_0 - n^{-1/3}, x_0 + H_\epsilon n^{-1/3}]}- \tilde{f}_0(x_0))  \leq \frac{f_0'(x_0)}{8}(H_\epsilon - 1)
$$
holds since $\tilde U^* \geq H_\epsilon$. By Lemma \ref{lemma:biascalc},  we then obtain
$$
\ap{n^{1/3}(\overline{\tilde f}_{[x_0 - n^{-1/3}, x_0 + H_\epsilon n^{-1/3}]}- \tilde{f}_0(x_0))  - \frac{f'_0(x_0)}{2} (H_\epsilon - 1)} = o_{\P}(1).
$$
Hence, $ \P(\tilde U^* \geq H_\epsilon, \Omega_\epsilon^c, \check{\Omega}_\epsilon^c ) \rightarrow 0$. By our choice of $H_\epsilon$, 
$$
\limsup_{n\rightarrow \infty }\P(\tilde U^* \geq H_\epsilon) \leq \epsilon 
$$
and since $\epsilon$ was arbitrary, we have shown $\tilde U^* = O_{\P}(1).$

Analogous arguments follow to show lower bounds so that $\tilde L^* = O_\P(1)$. Showing $\tilde U^*, \tilde L^* = \Theta_\P(1)$ is then done  by combining the argmax continuous mapping theorem (see Lemma 3.2.2 of \cite{van1996weak}, e.g.), Lemma \ref{lemma:fidi_uclt}, and Lemma \ref{lemma:glu_exists_tight}. 
\end{proof}
We state a result that summarizes those proven in  \cite{han2022berry}, demonstrating the maximizer and minimizer of the limit process, $G_{\ell, u},$ exist almost surely and are tight.

\begin{lemma}[Lemmas 4.4 and 4.5 of \cite{han2022berry}] \label{lemma:glu_exists_tight}
Suppose that Assumption \ref{assump_iso} holds. Define the random variables $(L_G^*, U_G^*)$ as the solutions to
$$
G_{L_G^*, U_G^*} = \max_{\ell > 0} \min_{u \geq 0} G_{\ell, u}.
$$
Then, $L_G^* = \Theta_\P(1)$ and $U_G^* = \Theta_\P(1)$.
\end{lemma}

Finally, we state the main conclusion of \cite{han2022berry} which summarizes the asymptotic theory for the the original estimator.

\begin{lemma}[Theorem 2.2 of \cite{han2022berry}] \label{lemma:iso_original_data}
Suppose Assumption \ref{assump_iso} holds. Then,
$$
    \sup_{t \in \R} \ap{\pr{n^{1/3}(\hat f_n(x_0) -  f_0(x_0)) \leq t \cdl \mathcal{O}} - \P(\sup_{\ell > 0} \inf_{u \geq 0} G_{\ell, u} \leq t)} = o_{\P_{\cO}}(1).
$$
    
\end{lemma}
\subsection{Supporting lemmas for Theorem \ref{theorem:gan}}

We note a condition for a uniform law of large numbers to hold uniformly over a class of probability measures, which is important in the triangular array setting. This is Theorem 2.8.1 of \cite{van1996weak} adapted to our setting.

\begin{lemma}[Theorem 2.8.1 of \cite{van1996weak}] \label{lemma:uniform_gc_class}
    Let $\cH$ be a class of uniformly bounded measurable functions from some subset $\cV \subseteq \R^r$ to $\R$. Suppose, as $n\rightarrow\infty$,
    $$
    \frac{\log N(\epsilon, \cH, \|\cdot\|_{\cV})}{n} \rightarrow 0
\text{ for every $\epsilon > 0$}    
$$
where $\|h\|_{\cV}$ denotes the infinity norm on $\cV$.
Then, for any $\epsilon > 0$, and for any collection of probability measures $\cQ$, on $\R^r$
$$
\sup_{\Q \in \cQ}\P\pa{ \sup_{h \in \cH} \ap{\frac{1}{n }\sum_{i=1}^n h(\mV_i) - \E[h(\mV_1)]}  > \epsilon} \rightarrow 0
$$
where $\mV_1, \mV_2, ..., \mV_n$ are independently distributed as $\Q$ inside the supremum.
\end{lemma}
Next is a covering number lemma for Lipschitz functions on a compact, convex domain.

\begin{lemma}[Theorem 2.7.1 of \cite{van1996weak}] \label{lemma:lip_covering_num}
Suppose that $\cV$ is a bounded, convex subset of $\R^r$. Fix $R >0$. Then, for the norm $\|f\|_{\cV} = \sup_{\mv \in \cV} |f(\mv)|$ defined for real-valued functions on $\cV$, we have
$$
\log N(\epsilon, BL(\cV,R), \|\cdot\|_{\cV}) \leq C(r,R,\cV) \cdot  \epsilon^{-r}.
$$
Here, $C(r,R,\cV)$ is a constant depending only on $(r, R, \cV)$, and $BL(\cV,R)$ denotes the set of bounded, real-valued Lipschitz functions $f$ on $\cV$ that satisfy
$$
\sup_{\mv \in \cV} |f(\mv)| + \sup_{\mv \neq \mv'} \frac{|f(\mv) - f(\mv')|}{\|\mv - \mv'\|_2} \leq R.
$$.
\end{lemma}

\subsection{Supporting lemmas for Theorem \ref{flow_wcvg}}

We first state a few analytical results to relate topological properties of a forward map to its inverse. The first is a direct consequence of the inverse function theorem, which we state without proof.
\begin{lemma}[Theorem 19.24 of \cite{rudin1976principles}]\label{lemma:c1inverse}
Suppose $\mh : \R^r \rightarrow \R^r$ is a bijective, continuously differentiable map with $\det(\sD \mh(\mv)) \neq 0$ for all $\mv \in \R^r$. Then, the inverse map $\mh^{-1}:\R^r \rightarrow \R^r$ is also continuously differentiable with
$$
\sD(\mh^{-1})(\mv) = \sD\mh(\mh^{-1}(\mv))^{-1}. 
$$
\end{lemma}
Next, we state a basic lemma for upper triangular matrices that bounds the operator norm of their inverse.
\begin{lemma}\label{lemma:inverse_op_bound}
 Suppose $\Ab \in \R^{r\times r}$ is an upper triangular matrix, with diagonal terms lower bounded by a constant $k > 0$. Then,
 $$
 \|\Ab^{-1}\|_{\rm op} \leq \frac{\sqrt{r}}{k}\pa{1+ \frac{\|\Ab\|_{\rm op}}{k}}^{r-1}.
 $$
\end{lemma}
\begin{proof}
Since $\Ab$ is upper triangular, the proof idea is to use backsubstitution and explicitly bound the terms of $\Ab^{-1}$. We omit the proof as it involves tedious algebraic details.
\end{proof}
Next is a bound of the Wasserstein distance via the KL-divergence, provided the measures' supports are uniformly bounded.
\begin{lemma}\label{lemma:w1_kl_bddsupport}
    Suppose that $\mV_1, \mV_2,...$ and $\mV$ are $\R^r$-valued random variables with supports $\cV_1,  \cV_2,... $ and $\cV$ contained in a fixed and compact set $\tilde \cV \subseteq \R^r$. Also, suppose $ \cV_i \supseteq \cV$ for all $i=1,2,...$  Then,
    $$
    \sW_1(\P_{V_i}, \P_V) \leq 2\cdot  \sup_{\mv \in \tilde\cV} \|\mv\|_2 \cdot \sqrt{\frac{1}{2} \cdot \KL(\P_V||\P_{V_i} )} < \infty
    $$
    for all $i=1,2, ...$
    \end{lemma}
\begin{proof}
    This  follows by using that $\tilde \cV$ is bounded, then Pinsker's inequality. 
\end{proof}

    We then note a dual representation of $\sW_1$.
    
\begin{lemma}[Remark 5.16 of \cite{villani2008optimal}]\label{lemma:kr_dual}
    Suppose that $\mV$ and $\mW$ are $\R^r$-valued random variables with $\E[\|\mV\|_2] < \infty$ and $\E[\|\mW\|_2]< \infty$. Denoting ${\rm{Lip}}(r,1)$ as the collection of all real-valued, $1$-Lipschitz functions on $\R^r$, we then have
    $$
    \sW_1(\P_V, \P_W) = \sup_{D \in {\rm Lip}(r,1)}\E\Big[D(\mV)-D(\mW)\Big].
    $$
\end{lemma}

Next is the change-of-variables formula for densities. 

\begin{lemma}[Corollary 4.7.4 of \cite{10.1093/oso/9780198572237.003.0008}]\label{lemma:flowcvg:cov}
Suppose that $\mV$ is an $\R^r$-valued random variable with Lebesgue density $p_V$. Take a continuously differentiable, bijective function with a continuously differentiable inverse $\mH^{-1}: \R^r \rightarrow \R^r$. The Lebesgue density of $\mH(\mV)$, denoted by $p_{H(V)}$, then exists and is given by 
$$
p_{H(V)}(\mv) = p_V(\mH^{-1}(\mv)) \cdot |\det(\sD (\mH^{-1})(\mv))|
$$
for any $\mv \in \R^r$.
\end{lemma}

\section{Auxiliary results}\label{sec:otherresults}

\subsection{Auxiliary results for Theorem \ref{thm_reg}}

We collect auxiliary technical lemmas in this section. We cite without proof the following as they apply directly to our situation. 

\begin{lemma}[Theorem 2.3.2 of \cite{Durrett2019}]\label{lemma:prob_and_as}
Let $\mV_1, \mV_2,...$ be a sequence of $\R^r$-valued random variables. The sequence converges in probability to some random variable $\mV$ if and only if for each subsequence $n_k$ of $n$, there is a further subsequence $n_{k_\ell}$ that converges almost surely to $\mV$.
\end{lemma}

Since our notion of weak convergence conditional on $\cO$ is nonstandard, we state analogues of weak convergence results. Their proofs follow by applying the classical results to each realization of the sequence of conditional measures. First is Polya's theorem.
\begin{lemma}[Problem 3.2.9 of \cite{Durrett2019}]\label{lemma:polya_cdl}
Let $\mV, \mV_1, \mV_2,...$ be $\R^r$-valued random variables, with $\mV$ continuous. Then, $\mV_1, \mV_2, ...$ converges weakly to $\mV$ conditionally on $\cO$ if and only if 
$$
\sup_{\mv \in \R^r} | \P(\mV_n \leq \mv \mid \cO) - \P(\mV \leq \mv)| \rightarrow 0 
$$
almost surely.
\end{lemma}
Next is the Cramer-Wold device.

\begin{lemma}[Theorem 3.10.6 of \cite{Durrett2019}]\label{lemma:cramerwold_cdl}
    Let $\mV, \mV_1, \mV_2, ...$ be $\R^r$-valued random variables. The sequence $\mV_1, \mV_2, ...$ converges weakly to $\mV$ conditionally on $\cO$ if, for any $\ma \in \R^r$, $\ma^\top \mV$ converges weakly to $\ma^\top\mV$  conditionally on $\cO$. 
\end{lemma}

Next is the Lyapunov central limit theorem.
\begin{lemma}[Problem 3.4.12 of \cite{Durrett2019}]\label{lemma:lyapunov_cdl}
    Let $\{V_{i,n}\}_{n \geq 1, 1 \leq i \leq n}$ be a triangular array, conditional on $\cO$,  of real-valued random variables. Suppose $0 < \E[V_{i,n}^2 \mid \cO] < \infty$ almost surely, for all $n\geq 1$ and $1 \leq i \leq n$. Furthermore, define $s_n^2 = \sum_{i=1}^n \mathrm{Var}(V_{i,n} \mid \cO)$ and suppose that, for some $\delta > 0$, the Lyapunov condition
    $$
    \frac{1}{s_n^{2+\delta}} \sum_{i=1}^n \E[|V_{i,n} - \E[V_{i,n} \mid \cO]|^{2 + \delta} \mid \cO] \rightarrow 0 \text{ almost surely}
    $$
    is satisfied. We then have
    $
    \frac{\sum_{i=1}^n V_{i,n} - \E[V_{i,n} \mid \cO]}{s_n}$  converges weakly to the standard normal distribution  conditionally on $\cO$. 
\end{lemma}

Next, we have Dudley's entropy integral stated conditionally.

\begin{lemma}[Corollary 2.2.9 of \cite{van1996weak}]\label{lemma:wellvdv229}
Suppose $\{V(\mt) : \mt \in \cT\}$ is a collection of real-valued random variables indexed by a subset $\cT \subseteq \R^r$. Suppose
$$
\|V(\mt) - V(\mt')\|_{\bP_{\mid \cO}, \psi_2} \leq C \|\mt - \mt'\|_2 
$$
almost surely, for all $\mt, \mt' \in \cT$ and some non-random constant $C > 0$. There then exists a non-random constant $C' > 0$ such that, for every $\delta > 0$, 
\al{
\E\Bigg[\sup_{\mt,\mt \in \cT : \|\mt - \mt'\|_2 \leq \delta} |V(\mt) - V(\mt')| \mid \cO \Bigg] \leq C' \int_{0}^\delta \sqrt{\log N(\epsilon/2, \cT, \|\cdot \|_2)} \, \d \epsilon
}
almost surely.
\end{lemma}

Next, we cite a covering number bound.

\begin{lemma}[Problem 7, Chapter 2.1.1 of \cite{van1996weak}]\label{lemma:ballcvnum}

Take $R > 0$ and $\mv \in \R^r$. Then, for all $\epsilon > 0$,
$$
 N(\epsilon, \cB(\mv, R, \|\cdot\|_2), \|\cdot \|_2) \leq \pa{
 \frac{3R}{\epsilon}}^r.
$$
\end{lemma}

Next, we show that empirical processes over certain classes of functions indexed by $\cK$ satisfy an asymptotic equicontinuity condition.

\begin{lemma}\label{lemma:Ldonsker}

Suppose Assumption \ref{assump_noise1},  \ref{assump_bootstrap}, and  \ref{assump_l} hold. 
Let 
    $$
    \tilde V_n(\boldsymbol{\eta}) := \frac{1}{\sqrt{n}} \sum_{i=1}^n \sL(\boldsymbol{\eta}, \tilde\mZ_i)- \E[\sL(\boldsymbol{\eta}, \tilde\mZ)\mid \cO].
    $$
    Also, denote, for $\boldsymbol{\eta} \in \cK$ and $\mz \in \tilde\cZ$,
    $$
    \check R(\boldsymbol{\eta}, \mz) = \|\boldsymbol{\eta} - \tilde{\boldsymbol{\eta}}_0\|_2^{-1} \cdot \pa{\sL(\boldsymbol{\eta}, \mz) - \sL(\tilde{\boldsymbol{\eta}}_0, \mz) - \sD_{\boldsymbol{\eta}}\sL(\tilde{\boldsymbol{\eta}}_0, \mz)^\top (\boldsymbol{\eta} - \tilde{\boldsymbol{\eta}}_0)},
    $$
    and define the process 
    $$
    \check V_n(\boldsymbol{\eta}) := \frac{1}{\sqrt{n}}\sum_{i=1}^n \check R(\boldsymbol{\eta},\tilde\mZ_i)- \E\Big[\check R(\boldsymbol{\eta},\tilde\mZ)\mid \cO\Big].
    $$
There then exist some constants $C, C'>0$ such that any $\delta > 0$ small enough, we have
$$
   \E\Bigg[\sup_{\boldsymbol{\eta} \in \cK : \|\boldsymbol{\eta} - \tilde{\boldsymbol{\eta}}_0\|_2 < \delta} \ap{\tilde V_n(\boldsymbol{\eta}) - \tilde V_n(\tilde{\boldsymbol{\eta}}_0)} \mid \cO \Bigg] \leq C\delta
~~~{\rm and}~~~
\E\Bigg[\sup_{\boldsymbol{\eta} \in \cK : \|\boldsymbol{\eta} - \tilde{\boldsymbol{\eta}}_0\|_2 < \delta} \ap{\check V_n(\boldsymbol{\eta})} \mid \cO \Bigg] \leq C'\delta
$$
almost surely for all $n$.
\end{lemma}
\begin{proof}
    The proof of this theorem is a consequence of using Taylor's theorem to verify the hypothesis of Lemma \ref{lemma:wellvdv229}. We omit the details as it only involves some algebraic manipulation.
\end{proof}


From Theorem 5.52 of \cite{Vaart_1998}, in the context of Section \ref{sec:regular}, we obtain the rate of convergence for the studied estimator sequence.

\begin{lemma}[Theorem 5.52 of \cite{Vaart_1998}]\label{lemma:vdv552}
Suppose Assumption \ref{assump_bootstrap} and Assumption \ref{assump_bsmest} holds. Furthermore, assume that there exists some constant $C > 0$ such that, for every $n$ and for every sufficiently small $\delta > 0$,
\begin{equation}
\sup_{\|\boldsymbol{\eta} - \tilde{\boldsymbol{\eta}}_0 \|_2 < \delta} \E[\sL(\boldsymbol{\eta}, \tilde \mZ) - \sL(\tilde{\boldsymbol{
\eta
}}_0, \tilde\mZ) \mid \cO] \leq -C\delta^2\label{eq:thm552_bd1}
\end{equation}
and
\begin{equation}
\E\Bigg[\sup_{\|\boldsymbol{\eta} - \tilde{\boldsymbol{\eta}}_0 \|_2 < \delta} \ap{\frac{1}{\sqrt{n}} 
\sum_{i=1}^n\sL(\boldsymbol{\eta}, \tilde\mZ_i) - \sL(\tilde{\boldsymbol{\eta}}_{0}, \tilde\mZ_i) - \E[\sL(\boldsymbol{\eta}, \tilde\mZ) - \sL(\tilde{\boldsymbol{\eta}}_{0}, \tilde\mZ) \mid \cO]} \mid \cO \Bigg] \leq C\delta \label{eq:thm552_bd2}
\end{equation}
almost surely. We then have
$$
n^{1/2}(\tilde{\boldsymbol{\eta}}_n - \tilde{\boldsymbol{\eta}}_0) = O_{\P_{\cO\tilde U}}(1).
$$
\end{lemma}

The following lemma is the conditional triangular array analogue of Lemma 19.24 of \cite{Vaart_1998}. As usual, we ignore measure-theoretic complications.

\begin{lemma}[Lemma 19.24 of \cite{Vaart_1998}] \label{lemma:triangular_array_1924}
 Suppose $\{\mV_{i,n}\}_{1 \leq i \leq n, n\geq 1}$ forms a  triangular array conditional on $\cO$ of $\R^r$-valued random variables, and $\mV_{1,n}, ..., \mV_{n,n}$ are identically distributed conditional on $\cO$. Furthermore, suppose that $\cH_n$'s are classes of functions from $\R^r$ to $\R$ such that
\begin{enumerate}[label=(\alph*)]
\item for each $n$, $\cH_n$ always contains the function $h_0: \R^r \rightarrow \R$ with $h_0(\mv) = 0$ for any $\mv\in \R^r$;

\item for each $n$ and all $h \in \cH_n$, we have $\E[h(\mV_{1,n})^2 \mid \cO] < \infty$ almost surely; 

\item letting $\rho_n: \cH_n \times \cH_n \rightarrow [0,\infty)$ be 
$$
\rho_n(h_1,h_2) := \Big(\E\Big[(h_1(\mV_{1,n}) - h_2(\mV_{1,n}))^2 \mid \cO\Big]\Big)^{1/2},
$$
the conditional asymptotic continuity condition at $h_0$,
$$
\lim_{\delta \rightarrow 0} \limsup_{n\rightarrow \infty} \P\pa{\sup_{\rho_n(h,h_0) \leq \delta} \ap{\frac{1}{\sqrt{n}} \sum_{i=1}^n h(\mV_{i,n})  - \E[h(\mV_{1,n})  \mid \cO]} > \epsilon \mid \cO} = 0,
$$
holds for any choice of $\epsilon > 0$.
\end{enumerate}
Finally, suppose $\hat H_n: \R^r \rightarrow \R$ satisfies that $\P(\hat H_n \not\in \cH_n) = o_{\P}(1)$ and $\E[\hat H_n(\mV_{1,n})^2 \mid \cO] = o_{\P}(1)$.
We then have 
$$
\frac{1}{\sqrt{n}}\sum_{i=1}^n \hat{H}_n(\mV_{i,n}) - \E[\hat{H}_n(\mV_{1,n}) \mid \cO] = o_{\P}(1).
$$
\end{lemma}

The following is a consequence of Lemma \ref{lemma:triangular_array_1924}.

\begin{lemma}\label{lemma:donsker_verify}

Suppose Assumptions \ref{assump_noise1}, \ref{assump_bootstrap}, \ref{assump_l} hold and that $\tilde{\boldsymbol{\eta}}_0$ is in the interior of $\cK$ with probability tending to $1$. Let $\boldsymbol{\eta}_n$ be a sequence of random variables on the bootstrap space such that $\|\boldsymbol{\eta}_n - \tilde{\boldsymbol{\eta}}_0\| = o_{\P}(1)$, and define
$$
\check R({\boldsymbol{\eta}}_n, \mz) = \|{\boldsymbol{\eta}}_n - \tilde{\boldsymbol{\eta}}_0\|_2^{-1} \cdot \pa{\sL({\boldsymbol{\eta}}_n, \mz) - \sL(\tilde{\boldsymbol{\eta}}_0, \mz) - \sD_{\boldsymbol{\eta}}\sL(\tilde{\boldsymbol{\eta}}_0, \mz)^\top ({\boldsymbol{\eta}}_n - \tilde{\boldsymbol{\eta}}_0)}.
$$
We then have
$$
\frac{1}{\sqrt{n}}\sum_{i=1}^{n}  \check R({\boldsymbol{\eta}}_n, \tilde \mZ_{i}) - \E[\check R({\boldsymbol{\eta}}_n, \tilde \mZ_1) \mid \cO] = o_{\P}(1).
$$
\end{lemma}
\begin{proof}
The proof of this lemma reduces to applying Lemma \ref{lemma:triangular_array_1924}. In turn, the non-trivial conditions of this lemma are verified by Lemma \ref{lemma:Ldonsker}, the twice-differentiability of $\sL$, and the consistency of $\boldsymbol{\eta}_n$. We omit the details.
\end{proof}

\subsection{Auxiliary results for Theorem \ref{thm:iso}}

We recount Prokhorov's theorem applied to $C(\cE)$, for $\cE \subseteq \R^r$ compact, which we recall is complete and separable with respect to the supremum norm.
\begin{lemma}[Theorems 5.1 and 5.2 of \cite{billing}]\label{lemma:prokhorov}
Suppose that $G_1, G_2, ... $ is a tight sequence of $C(\cE)$-valued random variables. Then, for each subsequence $n_k$, there exists a further subsequence $n_{k_\ell}$ for which $G_{n_{k_1}}, G_{n_{k_2}}, ...$ converges weakly to some random variable $G_0 \in C(\cE)$.
\end{lemma}
The next is a basic relation between weak convergence to a deterministic limit, and convergence in probability explicitly applied to the space $C(\cE)$. This follows from the discussion around Equation (3.7) of \cite{billing}.

\begin{lemma}[Equation (3.7) of \cite{billing}]\label{lemma:weak_to_prob}
    Suppose that $G_1, G_2, ...$ is a sequence of $C(\cE)$-valued random variables that converges weakly to a deterministic $G_0 \in C(\cE)$. Then, 
    $$
    \sup_{\mx \in \cE}\Big|G_n(\mx) - G_0(\mx)\Big| = o_{\P}(1).
    $$
\end{lemma}

The next lemma states that there is a \textit{countable} collection of functions that one must integrate against to test equality of measures.
\begin{lemma}[Problem 1.10 of \cite{billing}]\label{lemma:equality_of_measures_countable}
    Take two $\R^r$-valued random variables $\mV$ and $\mW$ with supports contained in a compact subset $\cE \subseteq \R^r$. There then exists a sequence of real-valued continuous functions on $\cE$, denoted $f_1, f_2, ... $, for which 
    $$
    \E[f_i(\mV)] = \E[f_i(\mW)] \text{ for all }i=1,2,... \text{ if and only if } \mV \text{ and } \mW \text{ have equal distributions.}
    $$
\end{lemma}

The following result connects weak convergence to convergence in Wasserstein distance.

\begin{lemma}[Theorem 6.9 of \cite{villani2008optimal}]\label{lemma:weakcvg_wass}
Take an integer $d \geq 1$. Suppose that $\mV, \mV_1, \mV_2$, ... are random variables in $\R^r$ with $\E[\|\mV\|_2^d] < \infty$ and $\E[\|\mV_i\|_2^d]< \infty$ for each $i=1,2,...$ Then 
$$
\sW_{d}(\P_{V_i}, \P_{V}) \rightarrow 0
$$
if and only if $\mV_1, \mV_2, ...$ converges weakly to $\mV$ and 
$$
\E[\|\mV_i\|_2^d] \rightarrow \E[\|\mV\|_2^d].
$$
Here, $\sW_{d}(\cdot,\cdot)$ is the Wasserstein-$d$ metric using the Euclidean 2-norm.
\end{lemma}
We next present a few analytical facts about the joint distribution. 

\begin{lemma}\label{lemma:margdiff}
    Suppose that Assumptions \ref{assump_bootstrap} and \ref{ass:generator} hold and assume that $n$ is sufficiently large. Then, $\tilde p_n(x,y)$ and $\sD \tilde p_n(x,y)$ are $\tilde K$- Lipschitz on $ \cX \times \R$ almost surely.
\end{lemma}
\begin{proof}
The conclusion follows from applying the mean value theorem to $\tilde p_n$ and $\sD \tilde p_n$ on the set $\cX \times \R$. Then, we observe $\|\sD\tilde p_n(\mz)\|_{2}$ and $\|\sD^2 \tilde p_n(\mz)\|_{\rm op}$ are zero outside $\tilde \cZ_n$, and bounded by $\tilde K$, by Assumption \ref{ass:generator}(b), on $\tilde \cZ_n$.
\end{proof}

Next, we show that $\tilde f_0$ is twice-continuously differentiable on $\cX$ with a uniformly bounded second derivative.
\begin{lemma}\label{lemma:tildef0_twicediff}
Suppose that $n$ is large enough so that Assumptions \ref{assump_bootstrap} and \ref{ass:generator} hold. Then, $\tilde f_0$ is twice-continuously differentiable on $\cX$. Furthermore, $\tilde f_0''$ is upper bounded by a universal constant in $\cX$.
\end{lemma}
\begin{proof}
    We provide a sketch of the proof, as it involves tedious algebraic details. For any $x \in \cX$, it is almost surely true that $\tilde p_X(x) \geq \tilde K$, so that we can rewrite the regression function as
    $$
    \tilde f_0(x) = \int_{\tilde \cY} y \cdot \frac{\tilde p_n(x,y)}{\tilde p_X(x)} \, \d y
    $$
    for some compact set $\tilde \cY \subseteq \R$. By the mean value theorem and the bounded convergence theorem, we can demonstrate the identities
    $$
    \tilde f_0'(x) = \int_{\tilde \cY} y \cdot \sD_x\pa{\frac{\tilde p_n(x,y)}{\tilde p_X(x)}} \, \d y \qquad \text{and} \qquad
    \tilde f_0''(x) = \int_{\tilde \cY} y \cdot \sD_x^2\pa{\frac{\tilde p_n(x,y)}{\tilde p_X(x)}} \, \d y
    $$
    hold almost surely. Continuity and boundedness of the latter display by a universal constant follows due to the continuity and boundedness of the second derivatives of $\tilde p_n$ and $\tilde p_X$.
\end{proof}
We cite an algebraic identity of \cite{robertson1988order}.

\begin{lemma}[Theorem 1.4.4 of \cite{robertson1988order}]\label{lemma:robertson_maxmin}
Suppose that $m \geq 1$ is a positive integer, $\cV$ is a nondegenerate closed interval, and that $(v_1, w_1), ..., (v_m,w_m) \in \cV \times \R$ are fixed points. Define $\check f: \cV \rightarrow \R$ as the solution to 
$$
\check f := 
\underset{f : \cV \to \R \text{ nondecreasing}}{\arg\min}
\sum_{i=1}^m \big(w_i - f(v_i)\big)^2.
$$
Take any point $v_0 \in \cV$. The following identity then holds true:
$$
\check f(v_0) = \max_{a \in [m] : v_a < v_0} \min_{b \in [m]: v_b \geq v_0}\frac{1}{|\{1 \leq i \leq m : v_i \in [v_a,v_b] \}|} \sum_{1 \leq i \leq m : v_i \in [v_a,v_b]} w_i.
$$
\end{lemma}

Lemma \ref{lemma:robertson_maxmin} then implies that the max-min identity holds true for the bootstrap data with probability tending to $1$.

\begin{lemma}\label{lemma:maxmin}
    Suppose Assumptions \ref{assump_iso} and \ref{ass:generator} hold. Then, as $n\rightarrow\infty$,
    $$
    \P\Big(\Big|\Big\{i \in [n] : \tilde X_i \in \cX  \cap [x_0, \infty)\Big\}\Big| \neq 0\Big) \wedge\P\Big(\Big|\Big\{i \in [n] : \tilde X_i \in \cX \cap (-\infty, x_0]\Big\}\Big| \neq 0\Big) \rightarrow 1.
    $$
    As a result, the probability of the event 
    $$
    \cp{\tilde f_n(x_0) =  \max_{a \in [n] : \tilde X_a \leq x_0} \min_{b \in [n] : \tilde X_b\geq x_0} \overline{\tilde Y}_{[\tilde X_a, \tilde X_b]}}
    $$
    tends to $1$ as $n\rightarrow \infty$.  That is, the bootstrap least-squares solution at $x_0$, $\tilde f_n(x_0)$, satisfies a max-min formula with high probability as $n\rightarrow\infty$.
\end{lemma}
\begin{proof}
To demonstrate the first conclusion, recall that 
$$
\P\Big(\tilde X_i \in \cX \cap (-\infty,x_0] \mid \cO\Big) \geq \int_{\cX\cap (-\infty,x_0]} \tilde K \, \d x > 0
$$
almost surely from Assumption \ref{ass:generator}(b). Since $\tilde X_i$'s are conditionally independent given $\cO$, we have 
$$
    \P\Big(\Big|\Big\{i \in [n] : \tilde X_i \in \cX \cap (-\infty,x_0]\Big\}\Big| = 0 \mid \cO\Big) \leq \pa{1 - \int_{\cX \cap (-\infty, x_0]}\tilde K\,\d x}^n \rightarrow 0
$$
almost surely.  The exact same reasoning for
$$
\P\Big(\Big|\Big\{i \in [n] : \tilde X_i \in \cX \cap [x_0, \infty)\Big\}\Big| = 0 \mid \cO\Big)
$$
combined with the bounded convergence theorem completes the proof of the first claim.

To show that the second claim is true, we first work on the intersection of the events $\{|\{i \in [n] : \tilde X_i \in \cX \cap (-\infty, x_0]\}| \neq 0\}$ and $\{|\{i \in [n] : \tilde X_i \in \cX \cap [x_0, \infty)\}| \neq 0\}$. 

Since $\tilde f_n$ solves \eqref{eq:bootstrap-isotonic},
Lemma \ref{lemma:robertson_maxmin} yields 
$$
\tilde f_n(x_0) = \max_{a \in [n] : \tilde X_a \leq x_0} \min_{b \in [n] : \tilde X_b\geq x_0} \frac{1}{|\{i : \tilde X_i \in \cX \cap [\tilde X_a,\tilde X_b]\}|} \sum_{1 \leq i \leq n : \tilde X_i \in \cX \cap [\tilde X_a,\tilde X_b]} \tilde Y_i.
$$
where $\cX = [a^*, b^*]$. This is exactly the claim of the lemma.
\end{proof}

Next is a consequence of Lemma 4.3 from \cite{han2022berry}.
\begin{lemma}[Lemma 4.3 of \cite{han2022berry}]\label{lemma:xi_o1}
   Suppose Assumptions \ref{assump_bootstrap} and \ref{ass:generator} hold. Then,
    $$
    n^{1/3} \sup_{u \geq 0} \Big|\overline{\tilde \xi}_{[x_0-n^{-1/3}, x_0 + un^{-1/3}]}\Big| = O_\P(1) \quad \text{and} \quad  n^{1/3} \sup_{\ell \geq 0} \Big|\overline{\tilde \xi}_{[x_0 - \ell n^{-1/3}, x_0 + n^{-1/3}]}\Big| = O_\P(1).
    $$
\end{lemma}
The next lemma states conditions, under which a uniform central limit theorem holds conditional on $\cO$ in a triangular array setting. This is the combination of Lemmas 2.8.2 and 2.8.7 in \cite{van1996weak}.

\begin{lemma}[Lemmas 2.8.2 and 2.8.7 in \cite{van1996weak}]\label{lemma:uclt}
Suppose $\cH_{\cT} = \{h_{t} : \mt \in \cT\}$ is a collection of real-valued functions defined on $\R^r$ indexed by a subset $\cT \subseteq \R^r$. Suppose that $\cH_{\cT}$ contains $h_0$, the function that maps all points to zero. Additionally, suppose that there exists a non-negative function $H : \R^r \rightarrow \R$ such that
\begin{enumerate}[label=(\alph*)]   
    \item $\sup_{h_t \in \cH_{\cT}} |h_t(\mv)| \leq H(\mv) \text{ for all }\mv \in \R^r;$
    \item $\sup_{n} \E[H(\mV_{1,n})^2 \mid \cO] < \infty$ almost surely; 
    \item $\limsup_{n\rightarrow \infty} \E[H(\mV_{1,n})^2 \cdot \mathds{1}(H(\mV_{1,n}) \geq \epsilon n^{1/2}) \mid \cO]=0$ almost surely, for every $\epsilon > 0$,
    \end{enumerate}
    where $\{\mV_{i,n}\}_{1\leq i\leq n, n\geq 1}$ are $\R^r$-valued random variables forming a triangular array conditional on $\cO$ and $\mV_{1,n}, ..., \mV_{n,n}$ are identically distributed conditional on $\cO$. Let $\mV$ be another $\R^r$-valued random variable. Let
    $$
    \rho_{n}(h_t, h_{t'}) = \E\Big[\Big(h_{t}(\mV_{1,n}) - h_{t'}(\mV_{1,n})\Big)^2 \mid \cO \Big]^{1/2} \quad \text{and} \quad \rho_0(h_t, h_{t'}) = \E\Big[\Big(h_{t}(\mV) - h_{t'}(\mV)\Big)^2 \Big]^{1/2}.
    $$
    Suppose that
\begin{enumerate}[label=(\alph*)]
        \item as $n\rightarrow \infty$, $\sup_{\mt, \mt' \in \cT} |\rho_n(h_t, h_{t'}) - \rho_0(h_t, h_{t'})| \rightarrow 0
        $ almost surely; 
        \item uniformly in $n=1,2,...$, $\cH_{\cT}$ is totally bounded with respect to $\rho_n$ almost surely; 
        \item for every $\epsilon > 0$,
        $$
        \lim_{\delta \rightarrow 0} \limsup_{n\rightarrow \infty} \P\pa{\sup_{\rho_n(h_t, h_{t'}) < \delta} \ap{\frac{1}{\sqrt{n}} \sum_{i=1}^n h_{t}(\mV_{i,n}) -h_{t'}(\mV_{i,n}) - \E[h_{t}(\mV_{i,n}) - h_{t'}(\mV_{i,n}) \mid \cO]} > \epsilon \mid \cO} = 0
        $$
        almost surely.
    \end{enumerate}
    Then,
    $$
    \sup_{\mt \in \cT, s \in \R} \ap{\P\pa{\frac{1}{\sqrt{n}}\sum_{i=1}^n h_{t}(\mV_{i,n}) - \E[h_t(\mV_{1,n})\mid \cO]  \leq s \mid \cO}
     - \P\pa{S \leq s}} \rightarrow 0
    $$
    almost surely, where $S$ is a mean-zero normal random variable with variance ${\rm Var}(h_t(\mV))$.
\end{lemma}

{\small
\bibliographystyle{apalike}
\bibliography{AMS}
}
\end{document}